\documentclass[12pt]{article}
\usepackage[top=1in, left=1in, right=1in, bottom=1in]{geometry}	
\usepackage[parfill]{parskip}	
\usepackage[export]{adjustbox}
\usepackage{graphicx, xcolor}
\usepackage{amssymb}
\usepackage[ruled,vlined]{algorithm2e}
\SetAlgoSkip{0pt}
\usepackage{mathrsfs}
\usepackage{enumitem}
\usepackage{epstopdf}
\usepackage{makecell}
\usepackage{bbm}
\usepackage{amssymb}
\usepackage{amsmath}
\usepackage{amsfonts}
\usepackage{bm, bbm}

\newcommand{\pp}{\mathbb P}
\usepackage{lscape}
\usepackage{rotating}
\usepackage{setspace}
\usepackage{threeparttable}
\usepackage{booktabs}
\usepackage{floatrow}
\usepackage{multirow}
\usepackage{caption}
\usepackage{subcaption}
\usepackage[compact]{titlesec}
\usepackage{lmodern}
\usepackage{comment}
\usepackage{floatrow}

\usepackage[rightcaption]{sidecap} 
\usepackage{graphicx}
\usepackage{caption}  

\usepackage{multibib}
\newcites{supp}{Supplementary References}

\fontfamily{lmtt}\selectfont
\usepackage[T1]{fontenc}
\usepackage{natbib}
\bibpunct{(}{)}{;}{a}{}{,}
\usepackage{hyperref}
\usepackage{titling}
\newcommand{\subtitle}[1]{
  \posttitle{
    \par\end{center}
    \begin{center}\large#1\end{center}
    \vskip0.5em}
}

\usepackage{amsthm}
\usepackage{ulem}

\newtheorem{definition}{Definition}[section]
\newtheorem{assumption}{Assumption}[section]

\newtheorem{proposition}[definition]{Proposition}
\newtheorem{corollary}[definition]{Corollary}

\newcommand{\E}{\mathbb{E}}

\newcommand\indep{\perp\!\!\!\perp}

\renewcommand{\hat}{\widehat}
\renewcommand{\bar}{\overline}

\usepackage{color}
\begin{document}
\pagestyle{plain}

\newcommand{\blind}{0}

\newcommand{\tit}{\vspace{-1cm}
\Large Understanding Spatial Regression Models from a Weighting Perspective in an Observational Study of Superfund Remediation}

\if0\blind

{\title{\Large \tit
\thanks{We thank Michael Cork, Veronica Ballerini, Avi Feller, and Heejun Shin for helpful comments and conversations. This work was partly supported by the Patient Centered Outcomes Research Initiative (PCORI, ME-2022C1-25648), the National Institutes of Health (R01MD016054, R01AG066793, RF1AG080948, R01ES34021, R01ES037156, R01ES036436-01A1, T32ES007142), Amazon Web Services, and a National Science Foundation Graduate Research Fellowship (DGE 2140743). Any opinion, findings, and conclusions or recommendations expressed in this material are those of the authors and do not necessarily reflect the views of our funders. The computations in this paper were run on the FASRC Cannon cluster, supported by the FAS Division of Science Research Computing Group at Harvard University, as well as the Regulated Data (ReD) Environment, operated by Harvard University Information Technology and supported by the Office of the Vice Provost for Research at Harvard University.}
\vspace*{.3in}}
\author{\normalsize Sophie M. Woodward\thanks{Department of Biostatistics, Harvard University; email: \url{swoodward@g.harvard.edu}.} \and \normalsize Francesca Dominici
\thanks{Department of Biostatistics and Harvard Data Science Initiative, Harvard University; email:  \url{fdominic@hsph.harvard.edu}.}\and \normalsize Jos\'{e} R. Zubizarreta
\thanks{Departments of Health Care Policy, Biostatistics, and Statistics, Harvard University; email: \url{zubizarreta@hcp.med.harvard.edu}.}}
\date{}

\maketitle
\date{}
}\fi

\if1\blind
\title{ \tit}
\date{}
\maketitle
\fi

\vspace{-1cm}

\begin{abstract}
A key challenge in environmental health research is unmeasured spatial confounding, driven by unobserved spatially structured variables that influence both treatment and outcome. A common approach is to fit a spatial regression that models the outcome as a linear function of treatment and covariates, with a spatially structured error term to account for unmeasured spatial confounding. However, it remains unclear to what extent spatial regression actually accounts for such forms of confounding in finite samples, and whether this regression adjustment can be reformulated from a design-based perspective. Motivated by an observational study on the effect of Superfund site remediation on birth outcomes, we present a weighting framework for causal inference that unifies three canonical classes of spatial regression models---random effects, conditional autoregressive, and Gaussian process models---and reveals how they implicitly construct causal contrasts across space. Specifically, we show that: (i) the spatial error term induces approximate balance on a latent set of covariates and therefore adjusts for a specific form of unmeasured confounding; and (ii) the covariance structure of the spatial error can be equivalently represented as regressors in a linear model. Building on these insights, we introduce a new estimator that jointly addresses multiple forms of unmeasured spatial confounding and develop visual diagnostics. Using our new estimator, we find evidence of a small but beneficial effect of remediation on the percentage of small vulnerable newborns.

\end{abstract}

\begin{center}
\end{center}
\clearpage
\doublespacing

\singlespacing
\pagebreak
\tableofcontents
\pagebreak
\doublespacing

\section{Introduction}
\vspace{-0.1in}
\label{sec:section_intro}
\subsection{Superfund remediation and birth outcomes}

The United States contains thousands of hazardous waste sites that potentially expose millions of Americans to toxic pollutants. More than 78 million Americans---24\% of the US population, including 17 million children and 11 million elderly individuals---live within three miles of a hazardous waste site \citep{EPA2021Superfund}
. Many of these sites are legacies of historical industrial activities, including landfills, manufacturing facilities, mining sites, and processing plants. In 1980, Congress enacted the Comprehensive Environmental Response, Compensation, and Liability Act (CERCLA), 
granting the Environmental Protection Agency (EPA) authority to identify and remediate these sites, which became known as Superfund sites. Although cleanup efforts have been associated with reductions in congenital anomalies in infants \citep{currie2011superfund} and blood lead levels in children \citep{klemick2020superfund}, remediation remains contentious due to high costs---often tens of millions of dollars per site---and disputes over whether the benefits justify the expense \citep{hamilton1999costly}.

Our goal is to estimate the causal effect of Superfund site remediation on two birth outcomes, while adjusting for a rich set of sociodemographic covariates. A key challenge is unmeasured spatial confounding. Unmeasured spatial confounders are unobserved spatially structured variables that influence both treatment assignment and outcomes, potentially leading to biased effect estimates and invalid confidence intervals \citep{gilbert2021causal, reich2021a}. 
Examples include community participation and local funding availability, which can shape both cleanup efforts and birth outcomes \citep{petrie2006environmental}. 

\subsection{Spatial regression and spatial confounding}
\vspace{-0.1in}

To address potential unmeasured spatial confounding, researchers frequently introduce a mean zero spatially autocorrelated error term into a linear regression model and interpret the treatment coefficient as a causal effect, a practice that remains widespread in many disciplines. Throughout this paper, we refer to such models as ``spatial regression models'' \citep{anselin2002under}. For example, a standard approach is to regress the outcome variable (e.g., birth weight) on a binary treatment variable (e.g., Superfund site remediation) and measured confounders, while incorporating random intercepts at the county or state level. Alternatively, researchers use Gaussian process regression, modeling the correlation between the error terms of two Superfund sites as a function of their geographic proximity, or employ a conditional autoregressive model, which models spatial autocorrelation based on an adjacency or weight structure \citep{reich2021a}.

The literature, however, offers conflicting perspectives on whether spatial regression models should be used to adjust for unmeasured spatial confounding. A prevailing view is that spatial regression is inappropriate for this purpose, because the estimated treatment coefficient can be biased in finite samples \citep{hanks2015restricted, hodges2010adding, khan2023re,  narcisi2024effect, nobre2021effects, paciorek2010a, page2017estimation, schnell2020a}. This view also arises from the fact that a spatial regression model is inherently misspecified when there is unmeasured spatial confounding. Specifically, spatially structured error is typically modeled independently of treatment under the assumption that it does not capture unmeasured spatial confounding, which, if present, would induce dependence between the error term and the treatment \citep{hodges2010adding}.

Recently, \cite{gilbert2025consistency} showed that, under certain conditions, spatial regression models with Gaussian process covariance yield consistent treatment effect estimates under unmeasured confounding that is continuous in space. Although the bias of the treatment coefficient is non-zero in finite samples, the authors show that it vanishes asymptotically. 
While this result is of substantial interest, it remains uncertain whether similar consistency guarantees extend to other commonly used spatial regression models---such as random effects models and conditional autoregressive models---or what implications this has for the magnitude of unmeasured spatial confounding bias in finite samples.
 
\subsection{Contribution and outline}
\vspace{-0.1in}
Although this literature has advanced our understanding of spatial regression under unmeasured confounding, several critical aspects remain underexplored. 
First, most bias analyses focus primarily on linear outcome models with constant treatment effects and no treatment-covariate interactions. 
Second, while studies such as \cite{gilbert2025consistency} and \cite{paciorek2010a} use particular forms for the error covariance matrix, such as a Gaussian process structure, it is worth exploring how these choices adjust for bias in finite samples. 
This question extends to other error structures, such as random effects and conditional autoregressive models. 
Third, although \cite{dupont2023demystifying}, \cite{khan2023re}, and \cite{narcisi2024effect} accommodate a general error covariance matrix, they characterize bias through the eigenstructure of the precision matrix rather than through more interpretable measures of spatial autocorrelation. 
Finally, current methods often rely on strong distributional assumptions regarding the unmeasured confounder $U$, such as multivariate normality. 

More broadly, the adjustment mechanism by which spatial regression accounts for unmeasured spatial confounding remains to be understood from a design-based standpoint, leaving open questions about the hypothetical experiment and the causal contrast these models implicitly define in space. 
In this paper, we provide an exact finite-sample analysis of spatial regression models in binary treatment settings. 
Our contributions are fourfold. 
First, we introduce a general weighting framework for causal inference with spatial data that encompasses three canonical types of spatial regression models: random effects models, conditional autoregressive models, and Gaussian process models.  
We establish that generalized least squares (GLS) estimators of the treatment effect derived from spatial regression models are equivalent to minimal dispersion approximately balancing weights estimators, which minimize a measure of weight dispersion subject to specific covariate balance requirements \citep{bruns2023augmented, chattopadhyay2023implied, robins2007comment}. 
To our knowledge, these are the first results to show that a GLS estimator is equivalent to a weighting estimator that achieves exact balance on the means of measured covariates and approximate balance on an additional latent set of covariates---specifically, the eigenvectors of the error covariance matrix. 
This result elucidates how the error structure adjusts estimates to account for spatial patterns. 
This perspective is related to the work of \citet{hodges2010adding}, who showed that the spatially correlated error term in a linear model can be represented through additional canonical regressors with random coefficients. 
Importantly, we further demonstrate that the resulting GLS regression is equivalent to a form of penalized regression, and we provide a mathematical programming representation that characterizes the exact nature of this regression adjustment.

Second, we characterize the class of unmeasured confounders that these three spatial regression models are designed to adjust for. 
We do so by linking the unmeasured confounding bias to the confounder's Moran's I statistic, a standard measure of spatial autocorrelation. 
This mathematical characterization is accompanied by three visual diagnostics to complement traditional weighting diagnostics: a map of the implied weights, visualizations of unmeasured confounders that are approximately balanced, and a plot of the maximum confounding bias as a function of the unmeasured confounder's Moran’s I.

Third, we propose a more general spatial weighting estimator that circumvents key limitations of traditional spatial regression. 
Unlike standard approaches, our estimator  accounts for multiple forms of spatially autocorrelated unmeasured confounding while accommodating nonlinear outcome models and treatment effect heterogeneity. 
We derive bounds on the finite-sample bias of our estimator and evaluate its performance through simulation. 

Fourth, we apply our proposed spatial weighting estimator to study the impact of Superfund site remediation on birth outcomes. 
To our knowledge, this study represents the first national-scale linkage of Medicaid claims with EPA remediation data to evaluate health impacts of Superfund cleanup. 
Our approach reveals how the causal contrast of remediated and non-remediated Superfund sites is constructed across the geography of the United States, making explicit the implicit geography of the comparison. 
Our results suggest that cleanup efforts between 1991 and 2015 had a small but statistically significant beneficial effect on the rate of small, vulnerable newborns during 2016--2018. 

The remainder of this paper is organized as follows. 
Section \ref{sec:motivatingdata} provides institutional and scientific background on the Superfund program and describes the linked dataset used for analysis. 
Section \ref{sec:threeexamples} then reviews three canonical types of spatial regression models: random effects, conditional autoregressive, and Gaussian process models. 
Section \ref{sec:weighting} establishes a general weighting framework that encompasses GLS estimators derived from these spatial regression models as a special case.
Section \ref{sec:space-diagnostics} characterizes unmeasured confounding bias through Moran's I and introduces new visual diagnostics. 
Building on this, Section \ref{sec:extension} proposes an alternative estimator to spatial regression and describes its finite-sample properties in the presence of unmeasured spatial confounding.
Simulation results are detailed in the Supplementary Materials.  
Section \ref{sec:application} applies our approach to estimate the causal effect of Superfund site remediation on Medicaid-covered birth outcomes. 
Section \ref{sec:guidance} concludes with practical guidance for evaluating environmental interventions in spatial settings. 

\vspace{-0.05in}
\section{Superfund remediation in the United States}
\label{sec:motivatingdata}
\vspace{-0.1in}
\subsection{Background}
\vspace{-0.1in}
Remediation of Superfund sites is a federally managed process that aims to identify, prioritize, and clean hazardous waste sites to protect human health and the environment. 
Remediation typically proceeds through several stages, beginning with site identification and preliminary assessment, followed by placement on the National Priorities List (NPL), detailed investigation of contamination, design and implementation of remedial actions, and finally site deletion once cleanup goals are achieved.  

Superfund site remediation may improve birth outcomes in nearby areas by reducing maternal exposure to contaminants---such as arsenic, lead, other metals, volatile organic compounds, and persistent organic pollutants---that have been linked to adverse pregnancy and birth outcomes, or by interrupting these exposure pathways \citep{ouidir2020association, perera2003effects, porpora2019environmental, stillerman2008environmental}. However, quantitative evidence on the causal effects of remediation on maternal and infant health remains limited.  
\citet{currie2011superfund} report that remediation may reduce the incidence of congenital anomalies by 20–25\% but find no significant effects on low birth weight, preterm birth, or infant mortality. 
Most prior studies have examined associations between residential proximity to contaminated sites and birth outcomes rather than the causal impact of remediation itself \citep{brender2011residential, henn2016prenatal, keeler2023residential, kihal2017systematic, langlois2009maternal}. 
 
Evaluating the causal effect of remediation is complicated by several variables that influence both site remediation and birth outcomes (confounders) that are difficult to measure or obtain. One example is the severity of the contamination itself: sites that undergo cleanup are often those with less severe or more tractable contamination, and the severity of contamination likely affects the severity of adverse outcomes. A second example is local funding availability. Remediation is typically financed by potentially responsible parties, but when such parties cannot be identified or are unable to pay, funding may come from a combination of the Superfund trust fund, taxpayers, states, or an excise tax on chemical manufacturers, making funding availability complex \citep{kiel2001estimating}. Local funding availability may also correlate with community wealth and access to healthcare. 
While both of these variables are challenging to measure, they are likely to be spatially autocorrelated, providing an opportunity to account for them through a spatial confounding adjustment.

\subsection{Data}
\vspace{-0.1in}
In this study, we evaluate the impact of Superfund site remediation on birth outcomes through a novel synthesis of publicly available EPA and high-resolution census data with Medicaid birth claims from 2016--2018, which cover about 40\% of all U.S. births during this period, while accounting for potential bias from unmeasured spatial confounders.
The dataset includes $n=1429$ Superfund sites in the United States currently on the National Priorities List (NPL) as of October 30, 2025, excluding those for which remedial action had started before January 1, 1991. 
The binary treatment variable indicates whether a Superfund site was remediated between 1991 and 2015. More specifically, we define  
a site as \textit{treated} if the site was deleted from the NPL before December 31, 2015, and 
\textit{control} otherwise. Figure \ref{fig:treatment} plots binary treatment and Figure \ref{fig:treatment-definition} in the Supplementary Material illustrates the definition of binary treatment. 

\vspace{-0.15in}
\begin{figure}[!ht]
    \centering
    \includegraphics[width=\linewidth]{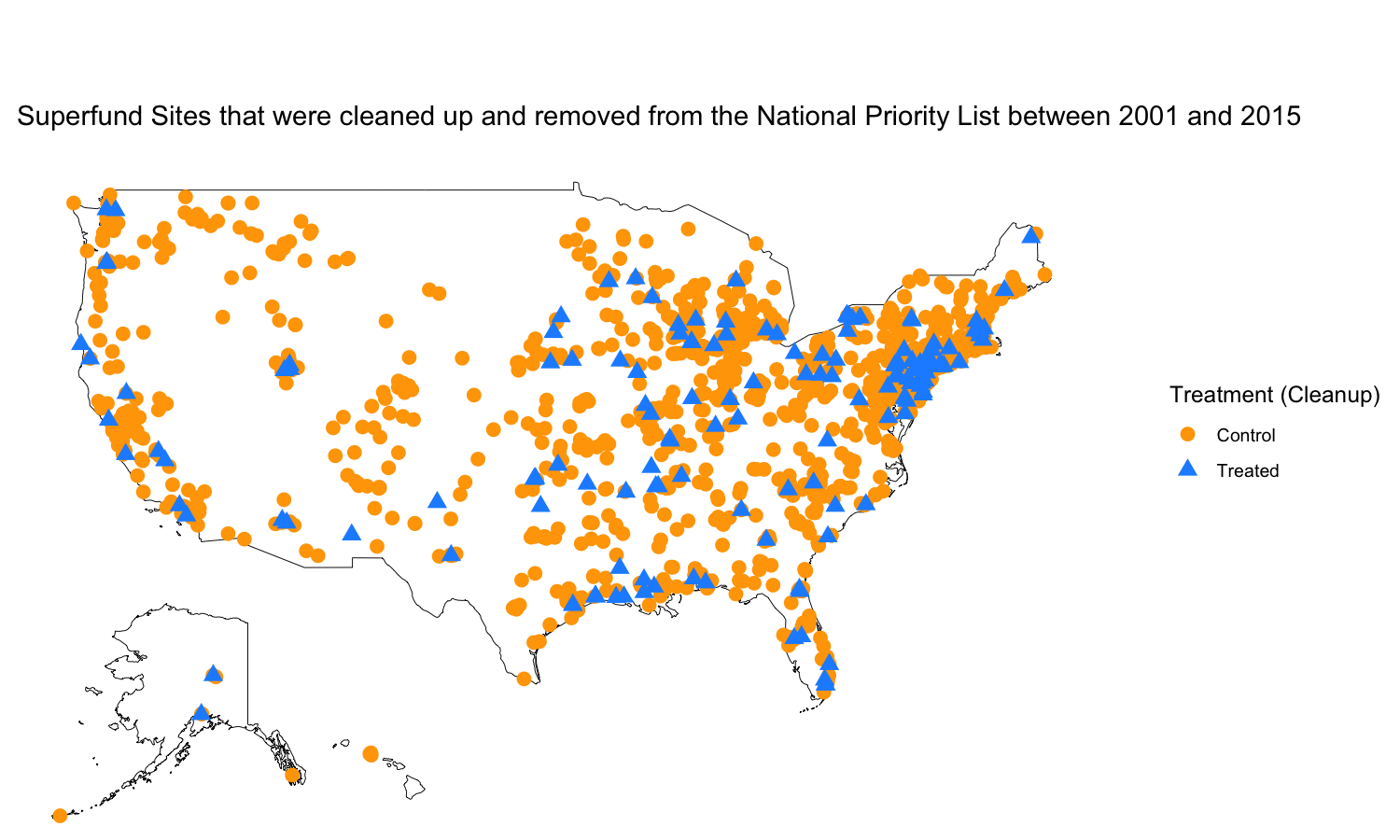}
    \caption{As of January 1, 1991, remediation action had not started at $n = 1429$ Superfund sites. During the fifteen year period between January 1, 1991 and December 31, 2015, $n_1 = 256$ sites (blue triangles) were remediated and deleted from the National Priorities List.}
    \label{fig:treatment}
\end{figure}
\vspace{-0.15in}

We consider eleven sociodemographic confounders derived from the 1990 Decennial Census, measured within each site's boundary and the surrounding 2-kilometer buffer. The buffer size is motivated by prior work examining infant health near Superfund sites \citep{currie2011superfund}. The census covariates, originally obtained at the census tract-level and subsequently summarized to the buffer-level, include population density, percentage of individuals identifying as Hispanic, Black, Indigenous, or Asian, percent owner occupied, median household income, median house value, percentage below the poverty line, percentage with a high school diploma, and the median year of housing construction \citep{USCensus1990, Schroeder2025NHGIS}. We additionally include five binary indicators: two indicating whether at least one metal or at least one volatile organic compound (VOC) was present, and three indicating whether at least one contaminant was detected in soil or other solids, water, or gas
. Finally, we include site score, a numerical rating assigned to each site by the EPA during a hazardous risk assessment.  
Treatment and covariate data are publicly available on Harvard Dataverse at \url{https://doi.org/10.7910/DVN/EKSCCU} and fully reproducible code is accessible at \url{https://github.com/NSAPH-Projects/spatial_regression_weighting}.

We consider two birth outcomes measured within each site's boundary and the surrounding $2$-kilometer buffer during the period 2016 to 2018. Importantly, this paper focuses on the design stage of spatial regression, and the outcome data are not used until Section \ref{sec:application}. The two outcomes are (i) the percentage of newborns classified as \textit{small vulnerable}---defined as preterm, small for gestational age, or low birth weight \citep{ashorn2023small}---and (ii) the percentage of newborns with congenital anomalies. Both variables were obtained at the zip code tabulation area (ZCTA) level from Medicaid T-MSIS Analytic File Inpatient Claims \citep{cms_taf_ip}. These ZCTA-level statistics were aggregated to the Superfund site and $2$-kilometer buffer areas (see Supplementary Section \ref{suppsec:dataapp} for more details). 

\vspace{-0.1in}
\section{Three common spatial regression models}
\label{sec:threeexamples}
\vspace{-0.1in}

In this section, we describe a general form of a spatial regression model and introduce three widely used classes of such models for estimating causal effects in spatial settings. 
Write $Z_i \in \{0,1\}$ for the binary treatment variable, $\bm{X}_i \in \mathbb R^{18}$ for the vector of $17$ observed covariates and the intercept, and $Y_i \in \mathbb R$ for one of the two birth outcomes for site $i$. 
Further write $\bm{Z} = (Z_1, \ldots, Z_n)^\top \in \mathbb{R}^n$ for the vector of binary treatments,  
$\bm{X} = (\bm{X}_1, \ldots, \bm{X}_n)^\top \in \mathbb{R}^{n \times 18}$ for the matrix of observed covariates, and $\bm{Y} = (Y_1, \ldots, Y_n)^\top \in \mathbb{R}^n$ for the vector of outcomes. Under the Stable Unit Treatment Value Assumption (SUTVA), define $\{Y_i(1), Y_i(0)\}$ as the potential outcomes under treatment and control, respectively, such that $Y_i = Z_iY_i(1) + (1-Z_i)Y_i(0)$.

The general form of a spatial regression model is
\begin{equation}
\label{eqngeneral}
\bm{Y} = \bm{X}\bm{\beta} + \tau \bm{Z} + \bm{\epsilon}, \qquad \bm{\epsilon} \sim \mathcal{N}\left(\bm{0},\;\bm{\Sigma} = \sigma^2 \bm{I}_n + \rho^2 \bm{S}\right), \qquad \bm{\epsilon} \indep (\bm{X}, \bm{Z}),
\end{equation}
where the covariance matrix of the error terms $\bm{\Sigma}$ includes both an independent and identically distributed term ($\sigma^2 \bm{I}_n$) and a spatially autocorrelated  term ($\rho^2 \bm{S}$). 
Typically, the $(i,j)$ entry of $\bm{S}$ decreases as the spatial distance between sites $i$ and $j$ increases, reflecting Tobler's First Law of Geography \citep{miller2004tobler}. 
The common assertion that spatial regression cannot adjust for unmeasured spatial confounding stems from the independence assumption, which treats the errors as unrelated to covariates or treatment and thus rules out the presence of unmeasured confounders \citep{gilbert2025consistency}.

Throughout this work, we assume that the error covariance matrix $\bm{\Sigma}$ is known and specified in advance. 
Our analysis is motivated by the observational study of Superfund remediation, but it extends to general linear regression models with arbitrary error covariance structures. 
We focus on the generalized least squares (GLS) estimator for $\tau$, denoted $\hat{\tau}_{GLS}$, which is given by
\vspace{-0.1in}
\begin{align}
    \begin{pmatrix} \hat{\bm{\beta}}_{GLS} \\
    \hat{\tau}_{GLS} \end{pmatrix} &= \bigg(\begin{pmatrix} \bm{X} & \bm{Z} \end{pmatrix}^T \bm{\Sigma}^{-1} \begin{pmatrix} \bm{X} & \bm{Z} \end{pmatrix} \bigg)^{-1} \begin{pmatrix} \bm{X} & \bm{Z} \end{pmatrix}^T \bm{\Sigma}^{-1} \bm{Y}. \label{eqn:gls}
\end{align}

\subsection{Random effects models}
\vspace{-0.1in}

As a first example, consider a random effects model that incorporates state-level random intercepts. 
Let $C_i$ denote the state indicator for site $i$, taking values $c \in \{1, \ldots, K\}$ with $K = 51$, so that $C_i = k$ indicates that Superfund site $i$ is in state $k$. 
The model is
\[
Y_i = \bm{\beta}^\top \bm{X}_i + \tau Z_i + \sum_{k=2}^K \gamma_k\, I(C_i = k) + \xi_i, \quad \gamma_k \stackrel{\text{i.i.d.}}{\sim} \mathcal{N}(0, \rho^2), \quad \xi_i \stackrel{\text{i.i.d.}}{\sim} \mathcal{N}(0, \sigma^2), \quad (\boldsymbol{\gamma}, \boldsymbol{\xi}) \indep (\bm{X}, \bm{Z}),
\]
where $\boldsymbol{\gamma} = (\gamma_1, \ldots, \gamma_K)$, $\boldsymbol{\xi} = (\xi_1, \ldots, \xi_n)$. 
In the general form \eqref{eqngeneral}, the spatial covariance matrix $\bm{S}$ is block diagonal, with the $k$th block given by $J_{n_k}$, an $n_k \times n_k$ matrix of ones, where $n_k = \sum_{i=1}^n I(C_i = k)$.

As $\rho^2 \rightarrow \infty$, the model enforces no shrinkage, and the coefficient estimates converge to those obtained from a model with fixed intercepts at the state level. 
As $\rho^2 \rightarrow 0$, the model imposes extreme shrinkage, forcing all random intercepts $\gamma_k$ to zero. 
In this case, the coefficient estimates converge to those obtained from a complete-pooling model, in which the data from all states are pooled together \citep{gelman2006data}.  
In our implementation of the random effects model in the Superfund data, we fix $\rho^2 = 10$ and $\sigma^2 = 1$.

\subsection{Conditional autoregressive models}
\vspace{-0.1in}

The second model we consider is the intrinsic conditional autoregressive (ICAR) model, 
\[
Y_i = \bm{\beta}^\top \bm{X}_i + \tau Z_i + \phi_i + \xi_i, \quad \bm{\phi} \sim \mathcal{N}\left(\bm{0},\, \rho^2 (\bm{D} - \bm{A})^{-1}\right),\quad \bm{\xi} \sim \mathcal{N}(\bm{0},\sigma^2 \bm{I}_n),\quad (\bm{\phi}, \bm{\xi}) \indep (\bm{X}, \bm{Z}),
\]
where $\bm{A}$ is a spatial weight matrix and $\bm{D}$ is a diagonal matrix with entries $D_{ii} = \sum_{j=1}^n A_{ij}$. 
A common choice for $\bm{A}$ is the spatial adjacency matrix, although distance-based matrices can also be used \citep{earnest2007evaluating}. 
While $\bm{D} - \bm{A}$ is singular and in a Bayesian sense leads to an improper prior for $\bm{\phi}$, this does not affect our theoretical results (see the Supplementary Material). In the general form of (\ref{eqngeneral}), we set $\bm{S} = (\bm{D} - \bm{A})^{-1}$. For the Superfund dataset, we implement the ICAR model by fixing $\rho^2 = 10$ and $\sigma^2 = 1$, and defining $A_{ij} = 1$ if site $i$ is among the five nearest neighbors of site $j$ or vice versa, and $A_{ij} = 0$ otherwise.

\subsection{Gaussian process models}
\vspace{-0.1in}
The third class is a Gaussian process (GP) model, specified as
\[
Y_i = \bm{\beta}^\top \bm{X}_i + \tau Z_i + \nu_i + \xi_i, \quad \bm{\nu} \sim \text{GP}(\bm{0}, \rho^2 \bm{K}), \quad \bm{\xi} \sim \mathcal{N}(\bm{0}, \sigma^2 \bm{I}_n), \quad (\bm{\nu}, \bm{\xi}) \indep (\bm{X}, \bm{Z}),
\]
where $\bm{K}$ is the kernel matrix of the Gaussian process. 
In the general form (\ref{eqngeneral}), this corresponds to setting $\bm{S} = \bm{K}$.

We apply the GP model to the Superfund dataset by setting $\rho^2 = 10$ and $\sigma^2 = 1$ and using a Mat\'ern correlation with hyperparameters $\kappa = 10,\phi=500/(2\sqrt{10})$ and the distance $d_{ij}$ between sites $i$ and $j$ is measured in meters. Specifically, 
$K_{ij} = 2^{1-\kappa}(\Gamma(\kappa))^{-1}(d_{ij}/\phi)^\kappa B_\kappa(d_{ij}/\phi)$
where $B_\kappa$ is the modified Bessel function of the second kind of order $\kappa$. 

\vspace{-0.1in}
\section{An encompassing weighting framework}
\label{sec:weighting}
\vspace{-0.1in}
This section presents our main results. First, we establish that the GLS estimator $\hat{\tau}_{GLS}$ in Equation \ref{eqn:gls} acts as a weighting estimator. 
Second, we demonstrate that it is specifically a minimal dispersion approximately balancing weights estimator: the weights minimize a measure of weight dispersion while achieving exact balance on the means of measured covariates and approximate balance on the means of a latent set of covariates. 
This property elucidates the exact nature of the adjustment induced by GLS, clarifying how it approximates key features of a hypothetical experiment and distinguishing it from other forms of regression adjustment.

While the following propositions apply to any error covariance matrix $\bm{\Sigma}$, we interpret them here in the context of spatial regression, where $\bm{\Sigma} = \sigma^2 \bm{I}_n + \rho^2\bm{S}$ and $\bm{S}$ encodes spatial autocorrelation. 
The following proposition shows that $\hat{\tau}_{GLS}$ admits a weighting representation, as it can be re-expressed as a difference of weighted means of the treated and control outcomes. 

\begin{proposition}
\label{thm:weighting}
For any positive semidefinite covariance matrix $\Sigma$, the GLS estimator of $\tau$ in Equation \ref{eqn:gls} can be expressed as $$\hat{\tau}_{GLS} = \sum_{i:Z_i = 1} w_i Y_i - \sum_{i:Z_i = 0} w_i Y_i$$ with weights of $\bm{w} = (w_1, \ldots, w_n)$ that admit the following closed-form representation
\begin{align}
    \bm{w} &= \bm{M} \frac{(\bm{I}_n -  \bm{\Sigma}^{-1}\bm{X}(\bm{X}^T \bm{\Sigma}^{-1} \bm{X})^{-1} \bm{X}^T)\bm{\Sigma}^{-1}\bm{Z}}{\bm{Z}^T\bm{\Sigma}^{-1}(\bm{I}_n- \bm{X}(\bm{X}^T \bm{\Sigma}^{-1} \bm{X})^{-1} \bm{X}^T\bm{\Sigma}^{-1})\bm{Z}}, \label{weights}
\end{align}
where $\bm{M}$ is the diagonal matrix with $(i,i)$ entry $M_{ii} = 2Z_i -1$. 
\end{proposition}
Proposition \ref{thm:weighting} shows that the implied unit-level data weights $w_1, \ldots, w_n$ depend only on the treatment indicators, covariates, and error covariance matrix $\bm{\Sigma}$. 
Therefore, the weights can be obtained as part of the design stage of the study as opposed to analysis stage. 
Moreover, the weights can equivalently be obtained by solving a quadratic programming problem. 

 \begin{proposition}
 \label{thm:minimal-weighting}

Suppose that $\bm{\Sigma} = \sigma^2 \bm{I}_n + \rho^2\bm{S}$, where $\bm{S}$ is an $n\times n$ positive semidefinite matrix. 
Let $\bm{v}_1, \ldots, \bm{v}_n$ be the eigenvectors of $\bm{S}$ with corresponding eigenvalues $\lambda_1 \geq \ldots \geq \lambda_n \geq 0$. 
Consider the following quadratic programming problem:
\begin{align}
    &\min_{\bm{w} \in \mathbb{R}^n} \bigg\{  \sigma^2 \sum_{i = 1}^n w_i^2 + \rho^2 \sum_{k = 1}^n \lambda_k (\sum_{i:Z_i = 1} w_i v_{ki} - \sum_{i:Z_i = 0} w_i v_{ki})^2 \bigg\}\label{dispersion}\\
    &\text{subject to } \sum_{i: Z_i = 1} w_i \bm{X}_i = \sum_{i:Z_i = 0} w_i \bm{X}_i, \sum_{i:Z_i=1} w_i =\sum_{i:Z_i=0} w_i= 1.
   \label{constraint}
   \end{align}

   The solution to this problem is the implied weights of $\hat{\tau}_{GLS}$,
   $$\bm{w} = \bm{M} \frac{(\bm{I}_n -  \bm{\Sigma}^{-1}\bm{X}(\bm{X}^T \bm{\Sigma}^{-1} \bm{X})^{-1} \bm{X}^T)\bm{\Sigma}^{-1}\bm{Z}}{\bm{Z}^T\bm{\Sigma}^{-1}(\bm{I}_n- \bm{X}(\bm{X}^T \bm{\Sigma}^{-1} \bm{X})^{-1} \bm{X}^T\bm{\Sigma}^{-1})\bm{Z}}.$$
     \end{proposition}

Proposition \ref{thm:minimal-weighting} demonstrates that $\hat{\tau}_{GLS}$ belongs to the class of minimal dispersion approximately balancing weights estimators, as it arises from solving a weighting problem that minimizes a measure of weight dispersion (\ref{dispersion}) subject to covariate balance requirements (\ref{constraint}). An extension to a more general $\bm{\Sigma}$, allowing any positive semidefinite structure, is given in Corollary \ref{corollary_general} in the Supplementary Material. This result builds on the existing literature that shows numerical equivalences between linear regression and balancing weights  \citep{bruns2023augmented, chattopadhyay2023implied, robins2007comment} and also relates to the connection between Gaussian process regression, kernel ridge regression, and weighting \citep{arbour2021using}. 
To our knowledge, the reformulation of linear regression models with general error structures as weighting estimators offers a new perspective into their finite-sample behavior.

The implied weights of GLS sum to one in each treatment arm and exactly balance the means of the measured covariates included in the regression, similar to both weighted least squares and augmented inverse probability weighting \citep{chattopadhyay2023implied}. 

A key distinction of the balancing weights problem for GLS, compared to weighted least squares and augmented inverse probability weighting, is its objective function $\{  \sigma^2 \sum_{i = 1}^n w_i^2 + \rho^2 \sum_{k = 1}^n \lambda_k (\sum_{i:Z_i = 1} w_i v_{ki} - \sum_{i:Z_i = 0} w_i v_{ki})^2 \}$. This objective simultaneously minimizes the variance of the weights while penalizing differences in the weighted means of a latent set of covariates---the eigenvectors $\bm{v}_1, \ldots, \bm{v}_n$ of $\bm{S}$---between the treated and control groups. These eigenvectors correspond to the canonical regressors identified by \cite{hodges2010adding}. When $\sigma^2$ is small relative to $\rho^2 \lambda_1, \ldots, \rho^2 \lambda_n$, the objective function is dominated by the mean differences of eigenvectors with the largest eigenvalues. As a result, GLS approximately balances the means of these eigenvectors across the treated and control groups---that is, $\sum_{i:Z_i = 1} w_i v_{ki} \approx \sum_{i:Z_i = 0} w_i v_{ki}$ for $\lambda_k$ large. 

Specifically, the weighting problem of Proposition \ref{thm:minimal-weighting} admits an equivalent formulation as a convex quadratic program with a soft‐balancing constraint on the sum of the mean eigenvector imbalances. In particular, one can instead solve $\min_{\bm{w}} \{  \sigma^2 \sum_{i = 1}^n w_i^2 \} $ subject to $ \sum_{i:Z_i=1} w_i = 1, \sum_{i:Z_i=0} w_i = 1
    \sum_{i: Z_i = 1} w_i \bm{X}_i = \sum_{i:Z_i = 0} w_i \bm{X}_i,$ and $ \sum_{k = 1}^n \lambda_k (\sum_{i:Z_i = 1} w_i v_{ki} - \sum_{i:Z_i = 0} w_i v_{ki})^2 \leq \Delta(\rho^2)
   $ for a specific function $\Delta$. The mapping $\rho^2 \mapsto \Delta(\rho^2)$ is derived in Supplementary Section \ref{sec:hyperparameter}.
   
The ratio $\sigma^2/\rho^2$ thus parametrizes a balance-dispersion trade-off. 
A larger ratio favors reducing the variance of the weights, whereas a smaller ratio favors balancing the means of eigenvectors with the largest eigenvalues. 
This result provides a new perspective on the bias-variance tradeoff in random effects and fixed effects models (see Supplementary Section~\ref{sec:re-fe-tradeoff}), and is of independent interest in the more general context of regression analysis.

Furthermore, the GLS estimator of $\tau$ can be equivalently obtained by augmenting a linear regression of $\bm{Y}$ on $\bm{X}$ and $\bm{Z}$ with the eigenvectors of $\bm{S}$, and applying ridge penalties to their coefficients. 
Specifically, minimizing the ridge regression loss
\[
\mathcal{L}(\bm{\beta}, \tau, \bm{\gamma}) = \sum_{i=1}^n \left(Y_i - \tau Z_i - \bm{\beta}^\top \bm{X}_i - \sum_{j=1}^n \gamma_j v_{ij}\right)^2 + \frac{\sigma^2}{\rho^2} \sum_{j=1}^n \frac{\gamma_j^2}{\lambda_j}
\]
with respect to $(\bm{\beta}, \tau, \bm{\gamma})$ yields $\hat{\tau}_{\text{ridge}} = \hat{\tau}_{\text{GLS}}$ (see Supplementary Section~\ref{sec:gls-ridge} for details). 
The ridge penalty on the coefficient for eigenvector $\bm{v}_k$ is $\sigma^2/(\rho^2 \lambda_k)$, so that eigenvectors with the largest eigenvalues are subject to the least shrinkage and effectively enter the linear model as covariates.
This perspective also relates to the connection between weighting and kernel ridge regression as discussed in \cite{ben2021balancing} and \cite{murray2024unifying}.

We now interpret this result to explain how spatial regression adjusts for confounding. 
Some intuition can be based on two key observations. 
First, as discussed above, spatial regression induces approximate mean balance on the eigenvectors of $\bm{S}$ with the largest eigenvalues. 
Second, these leading eigenvectors correspond to the highest Moran’s I statistics, which is a widely used measure of spatial autocorrelation, where higher values of Moran’s I indicate greater spatial correlation. 
More formally, the first eigenvector $\bm{v}_1$ of $\bm{S}$ maximizes  the ratio $\bm{x}^\top \bm{S} \bm{x} / \bm{x}^\top \bm{x}$, which coincides with Moran’s I after centering; $\bm{v}_2$ maximizes the ratio in the subspace orthogonal to $\bm{v}_1$, and so forth.
Thus, spatial regression implicitly balances the means of a latent set of unmeasured, spatially autocorrelated covariates. 
If an unmeasured confounder $\bm{U}$ is highly spatially autocorrelated, it will tend to lie within the span of these eigenvectors, and will therefore be nearly mean-balanced by the spatial regression adjustment. 
Section~\ref{sec:space-diagnostics} formalizes this intuition, providing both a mathematical analysis and visual illustration of the class of unmeasured confounders that spatial regression approximately balances.

\vspace{-0.1in}
\section{Understanding spatial regression contrasts across space}
\label{sec:space-diagnostics} 
Let $\bm{U} = (U_1, \ldots, U_n)$ denote an unmeasured confounder with sample mean $\bar{U} = \sum_{i=1}^n U_i/n$. 
In this section, we examine how the bias induced by $\bm{U}$ depends on its spatial autocorrelation structure.
Our target estimand is the average treatment effect on the treated (ATT):
\[
\tau_{\mathrm{ATT}} = \mathbb{E}\big[Y_i(1) - Y_i(0) \mid Z_i = 1\big].
\]
\subsection{Causal identification} 

To identify this quantity, we first posit the following assumptions.
\begin{assumption}[Consistency and no interference]\label{consistency}
    $Y_i = Y_i(z)$ when $Z_i = z$.
\end{assumption}
\begin{assumption}[Positivity of treatment]\label{positivity}
    $0 < \pp(Z_i = 1 | X_i, U_i) < 1$.
\end{assumption}
\begin{assumption}[Conditional ignorability given $X, U$]\label{ignorability}
    $\{Y_i(1), Y_i(0)\} \indep Z_i | X_i, U_i$.
\end{assumption}

Under these assumptions,
$
\tau_{\mathrm{ATT}} = \mathbb{E}\Big[
    \mathbb{E}(Y_i|Z_i=1, X_i, U_i) - \mathbb{E}(Y_i|Z_i=0, X_i, U_i)
    \ \Big|\, Z_i=1
\Big].$ However, $\bm{U}$ is unobserved and an analysis adjusting only for $\bm{X}$ will be subject to bias. In what follows, we show that when the outcome model is linear and $\bm{U}$ has a specific form of spatial autocorrelation, spatial regression methods can limit the magnitude of the resulting bias in finite samples.

\subsection{Spatial autocorrelation and unmeasured confounding bias}
We now relate the bias due to the unmeasured confounder $\bm{U}$ to the strength of its spatial autocorrelation. 
Moran's I statistic is a widely used metric for quantifying spatial autocorrelation. 
Let $\bm{S}$ be an $n \times n$ positive semidefinite spatial weights matrix with eigenvalues $\lambda_1 \geq \ldots \geq \lambda_n \geq 0$. The Moran's I statistic for $\bm{U}$ relative to $\bm{S}$ is defined as $\mathcal{I}(\bm{U};\bm{S}) = \sum_{i = 1}^n \sum_{j = 1}^n S_{ij} (U_i-\bar{U})(U_j-\bar{U})/(\lambda_1 \sum_{i = 1}^n (U_i-\bar{U})^2)$.
Although the conventional definition of Moran's I is $n\lambda_1 \mathcal{I}(\bm{U};\bm{S})/ \sum_{i = 1}^n \sum_{j=1}^n S_{ij}$, the version just presented is normalized to lie between 0 and 1 for ease of interpretation. 
A value of 1 indicates perfect spatial autocorrelation, while a value of 0 indicates no spatial autocorrelation.

 Moran's I can be interpreted as a scalar summary quantifying how strongly $\bm{U}$ aligns with the eigenvectors of $\bm{S}$ that have the largest eigenvalues. 
 To see this, without loss of generality, assume that $\bar{U} = 0$ and $\sum_{i = 1}^n U_i^2 = 1$. Then $\mathcal{I}(\bm{U};\bm{S}) = \sum_{i = 1}^n \lambda_i (\bm{v}_i^T \bm{U})^2$. 
In other words, Moran’s I is a weighted sum of the squared coefficients of $\bm{U}$ represented in the eigenbasis of $\bm{S}$, where each coefficient is weighted by the corresponding eigenvalue $\lambda_i$.
Section~\ref{sec:weighting} showed that a spatially autocorrelated $\bm{U}$ is largely supported on the leading eigenvectors of $\bm{S}$ with the largest eigenvalues; as a result, $\bm{U}$ is approximately mean-balanced by the weighting procedure of Proposition~\ref{thm:minimal-weighting}.
Therefore, spatial regression approximately adjusts for such $\bm{U}$ if the true outcome model is linear. 
The following proposition formalizes this intuition by relating the bias from unmeasured confounding due to $\bm{U}$ to its Moran's I statistic.

    \begin{proposition}[Relation between Moran's I and unmeasured confounding bias from $\bm{U}$]\label{thm:imbalance_spatialauto}
         Suppose that Assumptions \ref{consistency}--\ref{ignorability} hold. 
         Further assume that the outcome model is linear, so that $Y_i = \bm{\beta}^T \bm{X}_i + \tau Z_i + \gamma U_i + \epsilon_i$, where $\epsilon_i$ is an independent and identically distributed error term with $\E(\epsilon_i) = 0$, $\text{Var}(\epsilon_i) = \sigma^2$, and $\epsilon_i \indep (\bm{X}, \bm{Z}, \bm{U})$ for $i = 1, \ldots, n$. 
         Let $\bm{\Sigma} = \sigma^2 \bm{I}_n + \rho^2\bm{S}$ where $\bm{S}$ is an $n\times n$ positive semidefinite matrix with eigenvectors $\bm{v}_1, \ldots, \bm{v}_n$ and  corresponding eigenvalues $\lambda_1 \geq \ldots \geq \lambda_n \geq 0$. 
         Then, the bias of $\hat{\tau}_{GLS}$ for $\tau_{\mathrm{ATT}} = \tau$, conditional on $(\bm{X}, \bm{Z}, \bm{U})$, is bounded as follows:
         \begin{align*}
         \bigg\vert\E(\hat{\tau}_{GLS}|\bm{X}, \bm{Z}, \bm{U}) - \tau\bigg\vert &= \bigg\vert\gamma (\sum_{i:Z_i = 1} w_i U_i - \sum_{i:Z_i = 0} w_i U_i)\bigg\vert \\
         &\leq |\gamma| \sqrt{ \frac{c_0(2\sigma^2+\rho^2 \lambda_1 + \rho^2 \lambda_n)^2}{4(\sigma^2 + \rho^2 \lambda_1)(\sigma^2 + \rho^2 \lambda_n)(\sigma^2 + \rho^2 \lambda_1 \mathcal{I}(\bm{U};\bm{S}))}\sum_{i = 1}^n (U_i-\bar{U})^2} \end{align*}
        where $ c_0 = \sigma^2 \sum_{i = 1}^n w_i^2 + \rho^2 \sum_{k = 1}^n \lambda_k (\sum_{i:Z_i = 1} w_i v_{ki} - \sum_{i:Z_i = 0} w_i v_{ki})^2 $ and $\mathcal{I}(\bm{U}; \bm{S})$ denotes the Moran's I statistic of $\bm{U}$ as defined above. 
    \end{proposition}
Proposition~\ref{thm:imbalance_spatialauto} shows that the bias due to unmeasured confounding from $\bm{U}$ tends to decrease as its Moran's I statistic increases.
In other words, spatial regression reduces confounding bias arising from spatially autocorrelated unmeasured confounders.
Previous works have arrived at similar results, although through analytical or simulation-based analyses that are less generalizable or interpretable.
For example, \cite{paciorek2010a} assume that $\bm{U}$ and $\bm{Z}$ jointly follow a Gaussian process and use the range parameter of the Matérn correlation function to describe spatial autocorrelation, showing that bias decreases as the range of $\bm{U}$ increases. \cite{dupont2023demystifying} quantify spatial smoothness through the eigenvalues of $\bm{S}$ and find that confounding bias is small when $\bm{U}$ aligns with the eigenvectors of $\bm{S}$ corresponding to high eigenvalues. \cite{narcisi2024effect} define smoothness in terms of $\bm{U}^T \sqrt{\bm{\Sigma}^{-1}} (I-\bm{1} \bm{1}^T / n) \sqrt{\bm{\Sigma}^{-1}}  \bm{U}$, and show that bias decreases as this measure increases. 
By contrast, Proposition~\ref{thm:imbalance_spatialauto} avoids assumptions regarding the distribution or structure of $\bm{U}$ and $\bm{S}$, and makes use of Moran's I, which is a widely used measure of spatial autocorrelation.
  
\subsection{Visual diagnostics}
\vspace{-0.1in}

We now present three visual diagnostics that help clarify the spatial confounding adjustments made by spatial regression models, as well as the composition of the causal contrast they build across the spatial dataset.
Our visual diagnostics build on the weighting diagnostics introduced by \cite{chattopadhyay2023implied}. We apply them to the Superfund sites and birth outcomes dataset, examining each of the three spatial regression models described in Section~\ref{sec:threeexamples}.

The first two diagnostics depict the class of unmeasured confounders that spatial regression adjusts for. 
Figure~\ref{fig:ASMD_moran} in the Supplementary Material illustrates the first visual diagnostic, which plots the maximum absolute conditional bias arising from $\bm{U}$ against its Moran's I statistic and its coefficient in the outcome model.
For this diagnostic, we employ nonlinear programming because the analytical bound in Proposition \ref{thm:imbalance_spatialauto} can sometimes be quite conservative. 
For all three spatial regression models, this diagnostic reveals that the maximum absolute conditional bias decreases as the unmeasured confounder's Moran's I increases, and increases with the magnitude of $\gamma$, the coefficient of $\bm{U}$ in the outcome model. 
This pattern empirically corroborates Proposition~\ref{thm:imbalance_spatialauto}.
The second visual diagnostic, illustrated in Figure~\ref{fig:consolidated} of the Supplementary Material, displays three instances of $\bm{U}$ together with their respective Moran’s I statistics and mean imbalances. 
As expected, confounders characterized by higher Moran’s I exhibit reduced mean imbalance.
However, the three spatial regression models target unmeasured confounders with different forms of spatial autocorrelation. 
For the random effects model, spatially autocorrelated unmeasured confounders are nearly constant within states. 
The conditional autoregressive model targets spatially autocorrelated unmeasured confounders that are approximately constant over adjacencies; here, sites $i$ and $j$ are defined as adjacent if site $j$ is one of site $i$'s five nearest neighbors or vice versa.  
For the Gaussian process model, spatially autocorrelated unmeasured confounders appear to be smooth functions of spatial coordinates. 

\begin{figure}
    \centering
    \vspace{-1.75cm}
    \includegraphics[width=0.97\linewidth]{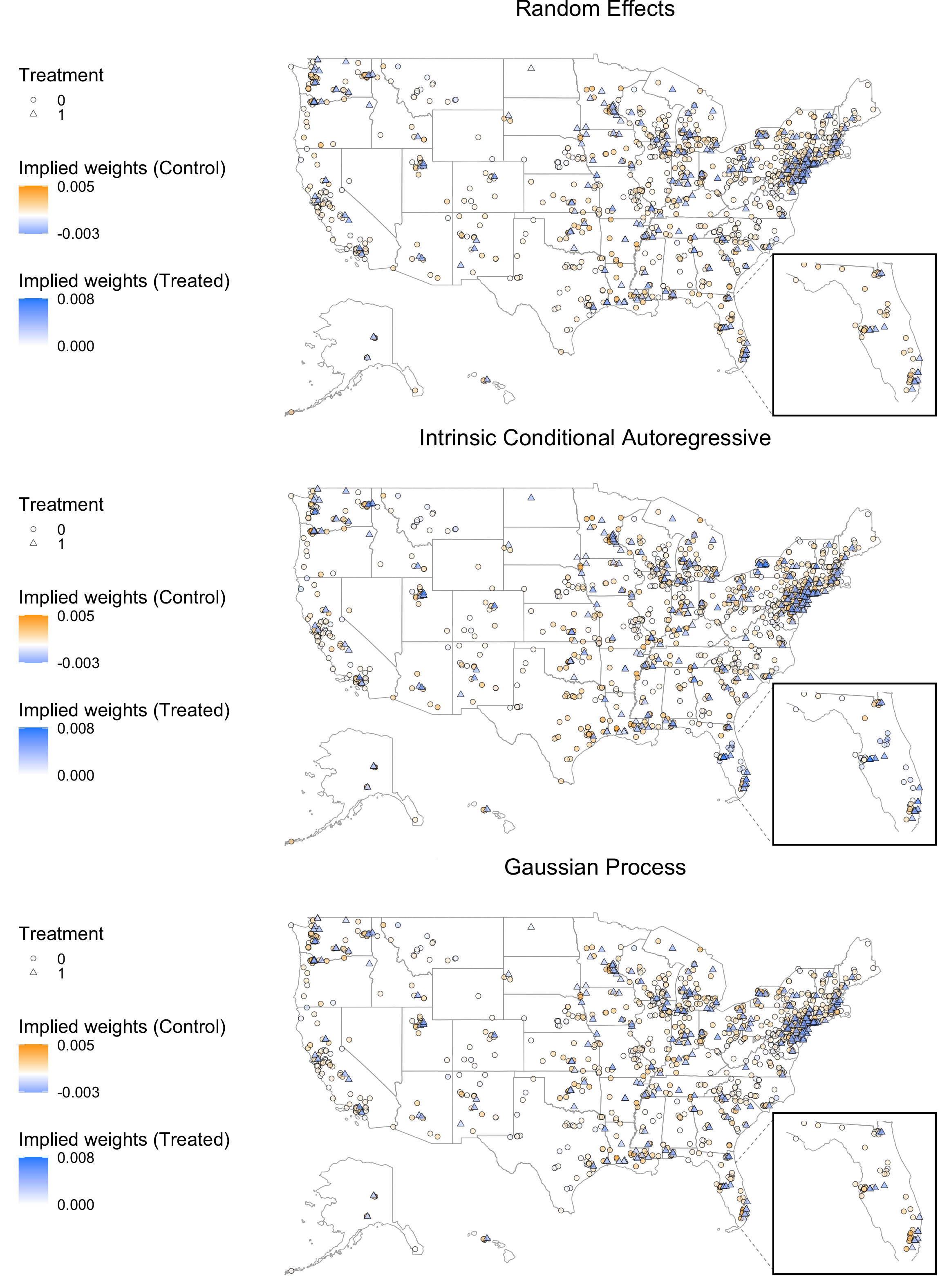}
    \caption{Implied weights $w_1, \ldots, w_n$ in three common spatial regression models. Our analysis enables a direct comparison of the spatial adjustments each model applies. 
    In Florida, the random effects model produces weights that are nearly uniform across the control sites in the state. The Gaussian process model assigns higher weights to control sites that are geographically closer to treated sites, leading to greater variation in control weights. The conditional autoregressive model places the largest weights on control sites adjacent to treated sites, and some control sites receive negative weights.
    }
    \label{fig:impliedweights}
\end{figure}

The third diagnostic, shown in Figure~\ref{fig:impliedweights}, displays the implied weights $w_1, \ldots, w_n$ from the three spatial regression models. 
For the random effects model, the weights are largest in magnitude for Superfund sites located in states with a high proportion of control sites, or for control sites in states with a high proportion of treated sites. 
For instance, sites in Montana and Wyoming, which contain only control sites, are assigned weights that are nearly zero. 
For the conditional autoregressive model, the largest weights correspond to treated sites bordering many control sites, or vice versa. 
For the Gaussian process model, the largest weights are observed for Superfund sites that are geographically proximate to sites of the opposite treatment status. Across all three models, some control sites receive negative weights, producing an estimator that is not sample bounded \citep{robins2007comment}.

In summary, Figure~\ref{fig:impliedweights} shows that spatial regression produces a weighting estimator in which the treated and control components are weighted to be spatially proximate, according to a model-specific definition of proximity. 
Further discussion of this spatially localized weighting is presented in Supplementary Section~\ref{sec:spatially-localized-weighting}. 
This approach is akin to distance-adjusted propensity score matching, which pairs treated and control units with similar propensity scores only if they are geographically close \citep{papadogeorgou2019a} and geographic regression discontinuity designs, which construct treated-control contrasts using observations situated on either side of a geographic boundary separating treatment groups \citep{keele2015a}. 
However, unlike these methods, spatial regression imposes a very strict requirement of linearity, and if there are deviations from this specification, its estimates can be subject to biases. 
See Supplementary Section \ref{sec:violations_linearity} for examples.
In the next section, we present an extension of spatial regression that is capable of addressing nonlinearities and heterogeneity of treatment effects in the outcome model.

\vspace{-0.1in}
\section{Extending spatial regression for robustness}
\label{sec:extension}
\vspace{-0.1in}

In this section, we extend spatial regression to estimate the average treatment effect on the treated, introducing a weighting estimator that adjusts for multiple forms of unmeasured spatial confounding, relaxes the assumptions of linearity and effect homogeneity, and can be constructed to be sample bounded.

Let $X^{\text{aug}} = (\bm{X}, \bm{V}) \in \mathbb R^{n \times (18 + J)}$ be the set of measured covariates, augmented with $J$ pre-selected eigenvectors $\bm{V} = (\bm{v}_1, \ldots, \bm{v}_J) \in \mathbb R^{n \times J}$ of any positive semidefinite matrix $\bm{S}$. Building on \cite{wang2020minimal} and \cite{zubizarreta2015stable}, we solve the following quadratic programming problem in the control group: 
\begin{align*}
     & \min_{\bm{w} \in \mathbb R^{n-n_t}} \sum_{i:Z_i = 0} w_i^2& \\
     & \text{subject to } \big \vert \sum_{i:Z_i = 0} w_i B_k (\bm{X}^{\text{aug}}_i) - \frac{1}{n_t}\sum_{i:Z_i = 1} B_k (\bm{X}^{\text{aug}}_i) \big \vert \leq \delta_k, k \in 1, \ldots, K, 
     \sum_{i:Z_i = 0} w_i = 1, w_i \geq 0 \text{ }\forall i
 \end{align*}
where the functions $B_k(\bm{X}^{\text{aug}})$, for $k = 1, \ldots, K$, represent a set of pre-specified basis functions, $\delta_k$ is the associated imbalance tolerance for the $k$th function, and $n_t = \sum_{i = 1}^n Z_i$.

The solution to this optimization problem yields a set of normalized, minimum-variance control weights that approximately balance the basis functions $B(\cdot) = \{B_1(\cdot), \ldots, B_K(\cdot)\}^T$ relative to the treated group. 
To ensure that the resulting estimator is sample bounded \citep{robins2007comment}, we also impose non-negativity constraints on the weights. 
The spatial weighting estimator is then defined as
$\hat{\tau}_{SW} = (n_t)^{-1}\sum_{i:Z_i = 1} Y_i - \sum_{i:Z_i = 0} w_i Y_i,$
where $\{w_i\}_{Z_i = 0}$ are the optimized weights obtained from the mathematical program described above.

Our weighting approach requires selecting $J$ augmenting eigenvectors $\bm{v}_1, \ldots, \bm{v}_J$. 
To reduce bias from the main classes of unmeasured spatial confounding implicitly targeted by the random effects, conditional autoregressive, and Gaussian process models, we include a large number of eigenvectors corresponding to the highest eigenvalues from each model.
Unlike methods that adjust for only a single confounder class, $\hat{\tau}_{SW}$ enables simultaneous adjustment for multiple types of unmeasured confounders, including those constant at the state level, those exhibiting smooth spatial variation, and those with minimal variation across adjacencies. 
Proposition~\ref{thm:robust-weighting-bias} characterizes the finite-sample bias and variance of the spatial weighting estimator as a function of the Moran’s I statistic of the unmeasured spatial confounder. 

In Supplementary Section \ref{sec:simulation}, we assess the performance of the spatial weighting estimator in simulation. First, we investigate whether it can mitigate confounding bias from three distinct classes of unmeasured confounders with different patterns of spatial autocorrelation. Second, we evaluate its performance while varying the complexity of the outcome model. The simulation results demonstrate that the spatial weighting estimator can flexibly accommodate complex nonlinearities and heterogeneity of treatment effects, while simultaneously adjusting for unmeasured confounders that are constant within states, smooth across adjacencies, or vary continuously over space.

\vspace{-0.1in}
\section{Effects of Superfund remediation on birth outcomes}
\label{sec:application} 
\vspace{-0.1in}

We now apply our spatial weighting approach to estimate the effect of Superfund site remediation during 1991--2015 on the rates of congenital anomalies and small vulnerable newborns during 2016--2018 among Medicaid-covered births in the United States. During this period, there were approximately $11.6$ million total U.S. births, of which 4,361,860 (37\%) were Medicaid-covered births identified in the inpatient claims data available to us. Of these, 266,604 occurred within Superfund site areas or their 2-kilometer buffers. For stable estimation, we excluded sites with fewer than $10$ Medicaid-covered live births within the site area and its 2-kilometer buffer, resulting in a final dataset of $n = 1079$ Superfund sites (see Supplementary Section \ref{suppsec:dataapp} for additional details). To our knowledge, this study represents the first national-scale linkage of Medicaid claims with EPA remediation data to evaluate health impacts of Superfund cleanup.
 
We report results for five methods: OLS, three spatial regression models (RE, CAR, GP), and the proposed spatial weighting (SW) approach, each with confidence intervals obtained via bootstrap with 500 replications. In the final specification of the spatial weighting estimator, we add the $10$ highest-eigenvalue eigenvectors from each of $\bm{S}^\text{RE}$, $\bm{S}^\text{CAR}$ and $\bm{S}^\text{GP}$ to the design matrix. We apply balancing thresholds $\delta_k$ of $0.001$ to the $17$ measured covariates and $0.01$ to the $30$ latent covariates after standardization.

Remediation appears to have a modest beneficial effect on the percentage of small vulnerable newborns, with ATT estimates ranging from $-0.3$\% to $-0.8$\% (relative to an overall rate of 14.2\% among Medicaid-covered births), but no discernible effect on the percentage of congenital anomalies. The spatial weighting approach yields ATT estimates of $-0.60$\% (95\% CI: $-1.15$\%, $-0.10$\%) for the rate of small vulnerable newborns and $0.07$\% (95\% CI: $-0.72$\%, $0.91$\%) for the rate of congenital anomalies, both with smaller confidence intervals than the other three spatial models. Estimates for all methods are reported in Table \ref{tab:dataapp-results}.

\begin{figure}[!ht]
    \centering
    \begin{subfigure}[t]{0.49\linewidth}
                \caption{}
\includegraphics[width=\linewidth]{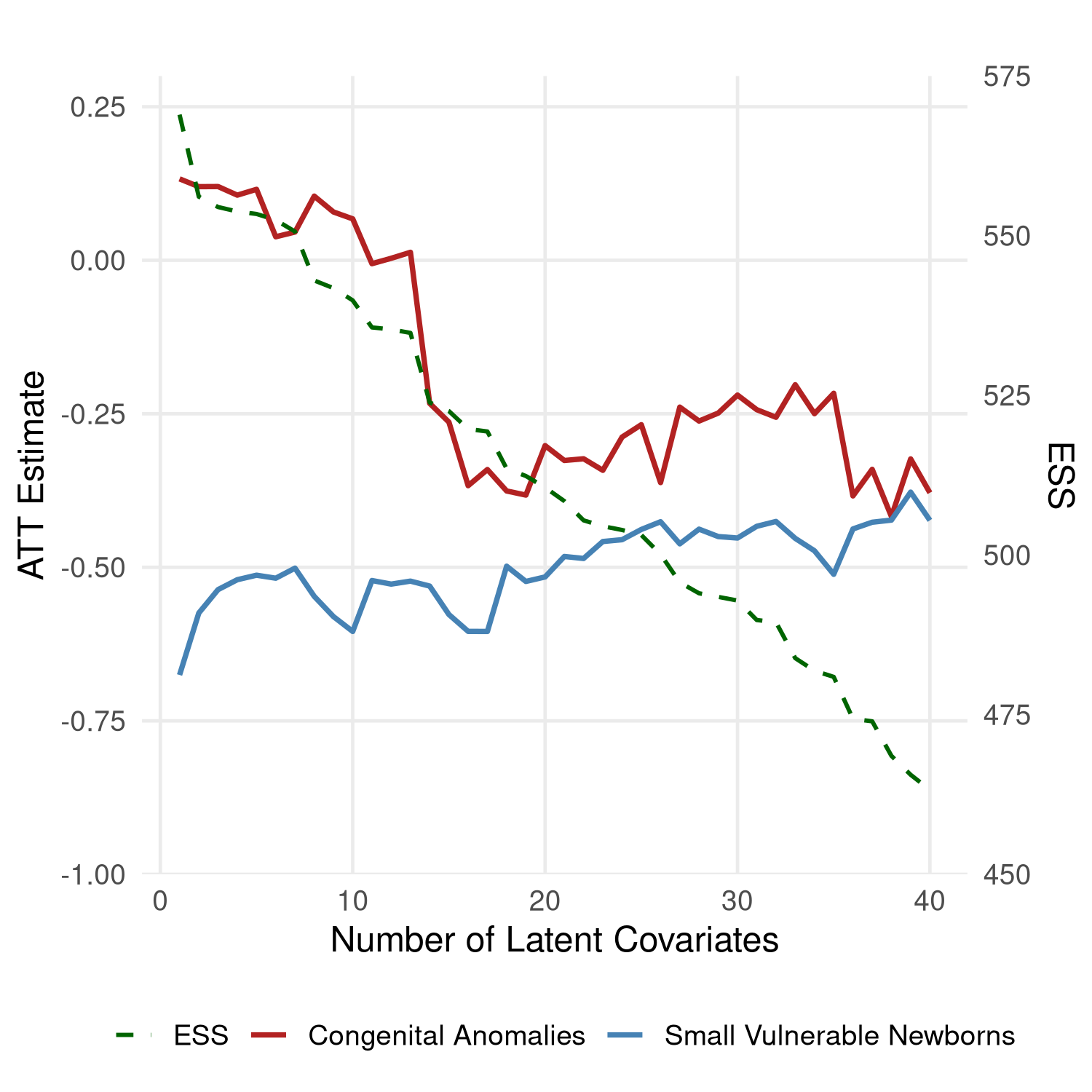}
        \label{fig:spatial_weighting_estimates}
    \end{subfigure}\hfill
    \begin{subfigure}[t]{0.49\linewidth}
                \caption{}
\includegraphics[width=\linewidth]{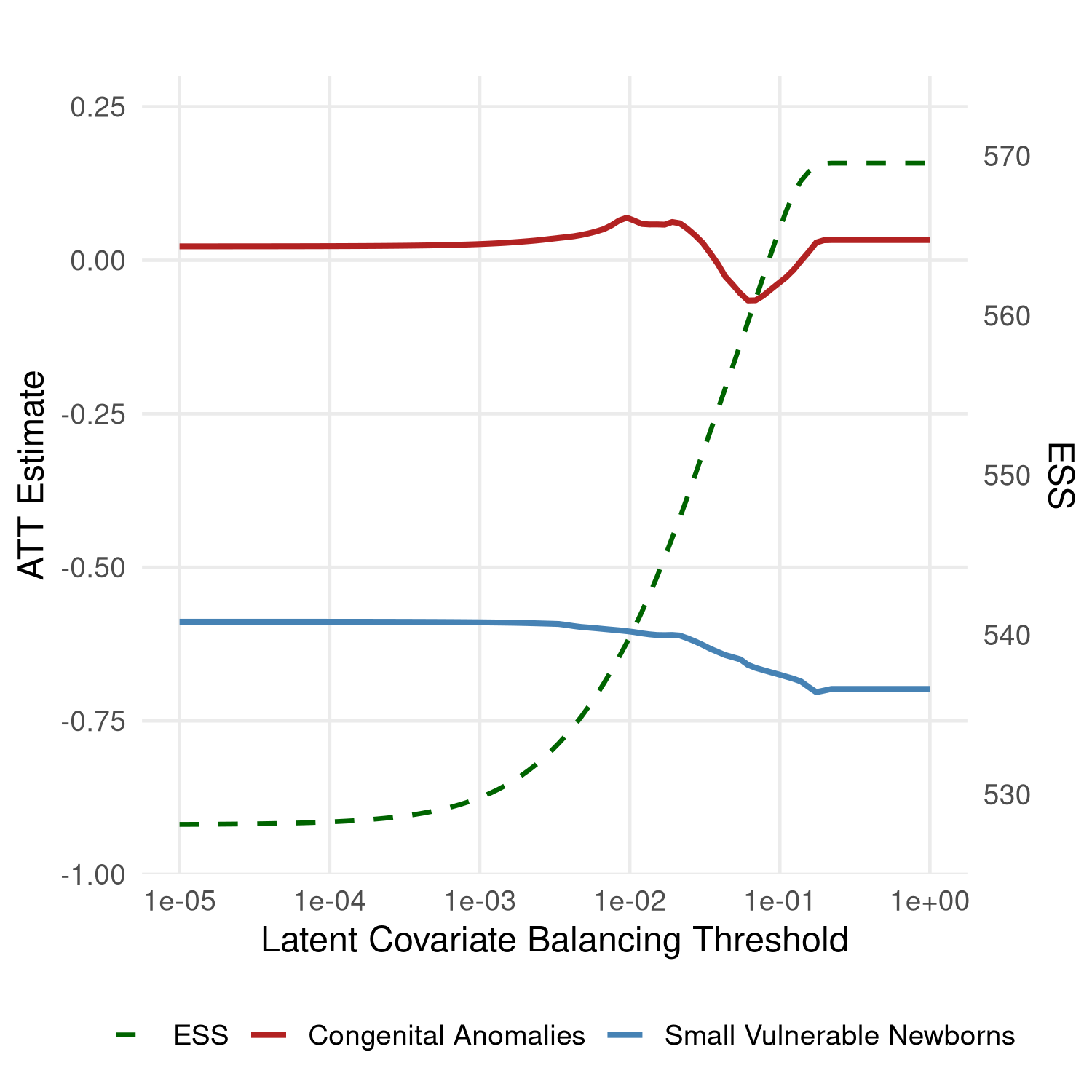}
        \label{fig:balancing_threshold}
    \end{subfigure}
\caption{Estimate of the ATT $\E(Y(1)-Y(0)\mid Z = 1)$ and effective sample size (ESS) 
$(\sum_{i = 1}^n |w_i|)^2/\sum_{i = 1}^n w_i^2$ obtained from the spatial weighting approach as a function of (a) the number of highest-eigenvalue eigenvectors from $\bm{S}^{\text{RE}}$, $\bm{S}^{\text{CAR}}$, and $\bm{S}^{\text{GP}}$ included as latent covariates, fixing the balancing threshold to 0.01, and (b) the balancing threshold applied to the latent covariates after standardization, fixing the number of latent covariates included from each model to $10$ ($30$ in total).}
        \label{fig:estimates-vs-eigenvectors}
    \end{figure}

Adjusting for unmeasured spatial confounding,  whether through spatial regression adjustment or our proposed weighting approach, attenuates the estimated effects of remediation on the percentage of small vulnerable newborns. This attenuation can be observed in (i) Table \ref{tab:dataapp-results}, where spatial weighting and spatial regression yield estimates closer to the null than OLS; (ii) Figure \ref{fig:spatial_weighting_estimates}, where including a larger number of latent covariates in the spatial weighting approach attenuates the ATT (blue) while reducing the effective sample size (green); and (iii) Figure \ref{fig:balancing_threshold}, where enforcing tighter balance on the 30 latent covariates similarly attenuates the ATT (blue) while reducing the effective sample size (green). 

\begin{table}
\centering
\begin{tabular}{lcc}
  \hline
 & \% Small Vulnerable & \% with Congenital Anomalies \\
  \hline
Ordinary least squares &  -0.74 (-1.24, -0.25) & 0.00 (-0.73, 0.80)\\
Random effects &  -0.52 (-1.37, 0.04) & -0.30 (-1.35, 1.09)\\
 Conditional autoregressive & -0.42 (-1.34, 0.12)  & 0.29 (-1.01, 1.50)\\
Gaussian process &  -0.36 (-1.53, 0.23) & -0.058 (-1.31, 1.12)\\
Spatial weighting & -0.60 (-1.15, -0.10) & 0.07 (-0.72, 0.91)\\
   \hline
\end{tabular}

\caption{Estimates of the ATT and 95\% confidence intervals from five methods, assessing the impact of Superfund site remediation (1991--2015) on two birth outcomes (2016--2018) among Medicaid-covered births within 2 kilometers of Superfund site areas. 
The two birth outcomes are: i) the percentage of newborns that are ``small vulnerable'', meaning preterm, low birth weight, or small for gestational age, and ii) percentage of newborns with congenital anomalies. }
\label{tab:dataapp-results}

\end{table}

Figure \ref{fig:diag-app} demonstrates how spatial weighting approximates a hypothetical experiment across space, providing visual diagnostics that link geography, confounding adjustment, and effect estimation. Panels (a) and (b) present covariate means within each treatment group before and after spatial weighting for the 17 measured and 30 latent covariates, respectively. The control group is reweighted so that its covariate means closely align with those of the treated group.  
The balance tables illustrate how the spatial weighting procedure achieves the specified balance constraints, giving the investigator explicit control over which covariates are balanced and to what degree. 

Panel (c) displays a histogram of the weights. The effective sample size for the spatial weighting approach is $539.9$, compared to $597.3$, $577.6$, $563.8$, and $569.4$ for OLS, RE, CAR, and GP, respectively; the nominal sample size is $n = 1079$. Although the effective sample size for the spatial weighting approach is lower than for the three spatial regression models, the  confidence interval is in fact narrower, suggesting that approximately balancing the observed covariates and simultaneously adjusting for three types of latent covariates 
may reduce variance by explaining residual spatial variation in outcome that is independent of observed covariates and treatment 
(see Supplementary Section \ref{sec:finite-sample-prop} and Figure \ref{fig:outcomemaps}).

Panel (d) plots the 30 latent covariates. This diagnostic reveals the latent covariates that spatial regression models balance implicitly and that the spatial weighting approach balances explicitly. By approximately balancing this set of latent covariates, the spatial weighting approach simultaneously adjusts for three classes of unmeasured spatial confounders: those that are constant within states, those that exhibit neighborhood-level similarity, and those that vary smoothly over space.

Panel (e) maps the weights in geographic space. By balancing on the 30 latent covariates, the spatial weighting approach constructs a geographically local contrast between treated and control sites. A substantial share of this contrast (17.7\% of the total weight) arises from the New York--New Jersey region, reflecting the high concentration of both treated and control Superfund sites there. 
There are 114 control sites that receive zero weight, with many located far from any treated site. Although spatial weighting reduces the effective sample size, it ensures that inference is restricted to treated and control sites that are comparable: control sites far from treated sites may differ with respect to unmeasured spatial confounders and should therefore be excluded from the comparison.

\begin{figure}
\vspace{-4em}
    \caption{Diagnostic plots summarizing the spatial weighting procedure.}
    \label{fig:diag-app}
    \centering
        \begin{minipage}[t]{0.4\textwidth}

  \begin{subtable}[t]{\textwidth}
      \vspace{0pt}
        \centering
        \scriptsize
        \renewcommand*{\arraystretch}{1.1}
        \begin{tabular}{p{1.4in}|ll|l}
  & \multicolumn{2}{|c|}{Control} & Treated \\
    \hline
Covariate & \makecell[l]{Before \\weighting} & \makecell[l]{After\\ weighting} & \\ 
  \hline
\makecell[l]{Population density\\ (ppl/mi$^2$)} & 1730 & 1610 & 1610 \\ 
  \makecell[l]{Proportion of Hispanic\\ residents} & 0.0681 & 0.0548 & 0.0549 \\ 
  \makecell[l]{Proportion of Black \\residents} & 0.104 & 0.118 & 0.118 \\ 
  \makecell[l]{Proportion of Asian\\ residents} & 0.0238 & 0.0189 & 0.0189 \\ 
  \makecell[l]{Proportion of Indigenous\\ residents} & 0.00814 & 0.00821 & 0.00818 \\ 
\makecell[l]{Proportion of renter\\ occupied housing} & 0.325 & 0.296 & 0.295\\ 
  \makecell[l]{Median household\\ income (\$)} & 32,500 & 31,700 & 31,700 \\ 
  Median house value (\$) & 94,400 & 87,100 & 87,000 \\ 
  \makecell[l]{Proportion of residents\\ in poverty} & 0.115 & 0.124 & 0.124 \\ 
  \makecell[l]{Proportion of residents\\ that graduated high school} & 0.742 & 0.737 & 0.737 \\
  \makecell[l]{Median year of housing \\construction} & 1960 & 1960 & 1960 \\ 
  Site score & 43.7 & 38.4 & 38.4 \\ 
  Metal present & 0.744 & 0.753 & 0.754 \\ 
  VOC present & 0.751 & 0.534 & 0.533 \\ 
  Contamination through soil & 0.845 & 0.816 & 0.815 \\ 
  \makecell[l]{Contamination through \\liquid medium} & 0.896 & 0.800 & 0.800 \\ 
  Contamination through gas & 0.145 & 0.0670 & 0.0667 \\
  \end{tabular}
\caption{Balance and representativeness: measured covariates.}

  \end{subtable}
    \end{minipage}
    \hfill
    \begin{minipage}[t]{0.45\textwidth}
        \begin{subtable}[t]{\textwidth}
      \vspace{0pt}
        \centering
        \scriptsize
        \renewcommand*{\arraystretch}{1.1}
        \begin{tabular}{p{0.5in}|ll|l}
  & \multicolumn{2}{|c|}{Control} & Treated \\
    \hline
Covariate & \makecell[l]{Before \\weighting} & \makecell[l]{After\\ weighting} & \\ 
  \hline
  $\bm{v}_1^{\text{RE}}$ &-0.00887 & -0.00786 & -0.00815 \\ 
  $\bm{v}_2^{\text{RE}}$ &-0.00775 & -0.01100 & -0.01130 \\  
  $\bm{v}_3^{\text{RE}}$ & -0.00858 & -0.00609 & -0.00639 \\ 
  $\bm{v}_4^{\text{RE}}$ & -0.00909 & -0.00427 & -0.00406 \\  
  $\bm{v}_5^{\text{RE}}$ & -0.00745 & -0.00625 & -0.00596 \\ 
    $\bm{v}_6^{\text{RE}}$ & -0.00648 & -0.00972 & -0.01000 \\  
  $\bm{v}_7^{\text{RE}}$ & -0.00596 & -0.01010 & -0.00996 \\  
  $\bm{v}_8^{\text{RE}}$ & -0.00673 & -0.00343 & -0.00313 \\ 
  $\bm{v}_9^{\text{RE}}$ & -0.00547 & -0.00568 & -0.00598 \\  
  $\bm{v}_{10}^{\text{RE}}$ & -0.00622 & -0.00286 & -0.00256 \\ 
  $\bm{v}_1^{\text{CAR}}$ & -0.000497 & 0.00202 & 0.00225 \\ 
  $\bm{v}_2^{\text{CAR}}$ & 0.000674 & -0.00275 & -0.00306 \\ 
  $\bm{v}_3^{\text{CAR}}$ & 0.000991 & -0.00424 & -0.00449 \\ 
  $\bm{v}_4^{\text{CAR}}$ & -0.000320 & 0.00175 & 0.00145 \\
  $\bm{v}_5^{\text{CAR}}$ & -0.000353 & 0.00144 & 0.00160 \\  
    $\bm{v}_6^{\text{CAR}}$ & -0.000915 & 0.00436 & 0.00415 \\
  $\bm{v}_7^{\text{CAR}}$ & 0.000147 & -0.000364 & -0.000669 \\  
  $\bm{v}_8^{\text{CAR}}$ & 0.0000614 & -0.000583 & -0.000278 \\  
  $\bm{v}_9^{\text{CAR}}$ & -0.000451 & 0.00174 & 0.00205 \\ 
  $\bm{v}_{10}^{\text{CAR}}$ &  0.000986 & -0.004170 & -0.004470 \\  
 $\bm{v}_1^{\text{GP}}$ & -0.0200 & -0.0178 & -0.0175 \\ 
  $\bm{v}_2^{\text{GP}}$ & -0.0140 & -0.0117 & -0.0114 \\ 
  $\bm{v}_3^{\text{GP}}$ & 0.00635 & 0.00377 & 0.00347 \\ 
  $\bm{v}_4^{\text{GP}}$ & -0.00901 & -0.0150 & -0.0148 \\ 
  $\bm{v}_5^{\text{GP}}$ & -0.0115 & -0.00942 & -0.00914 \\
  $\bm{v}_6^{\text{GP}}$ & 0.00187 & 0.00787 & 0.00818 \\
  $\bm{v}_7^{\text{GP}}$ & -0.00104 & 0.000352 & 0.000657 \\
  $\bm{v}_8^{\text{GP}}$ &  0.000827 & 0.000184 & -0.000120 \\
  $\bm{v}_9^{\text{GP}}$ & 0.00366 & 0.00838 & 0.00869 \\  
  $\bm{v}_{10}^{\text{GP}}$ & -0.00301 & -0.00223 & -0.00254 \\ 
   \hline
\end{tabular}
        \caption{Balance and representativeness: latent covariates.}
    \end{subtable}
    \end{minipage}

\begin{subfigure}[t]{0.25\textwidth}
    \vspace{0pt}
    \centering
    \includegraphics[height=2.1in]{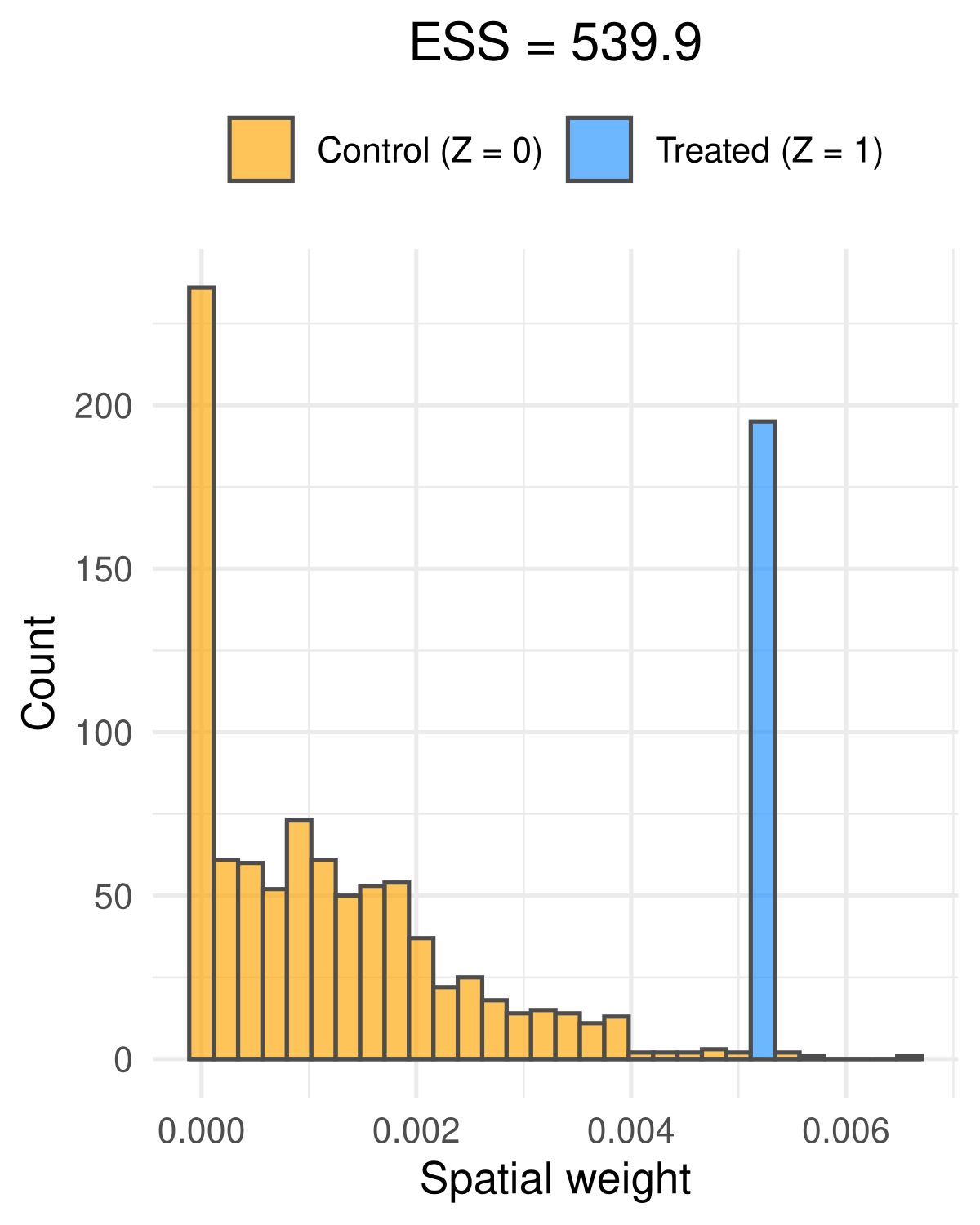}
    \caption{Weight dispersion.}
    \label{fig:weights_hist}
\end{subfigure}
\hfill
\begin{subfigure}[t]{0.7\textwidth}
    \vspace{0pt}
    \centering
    \includegraphics[height=1.9in]{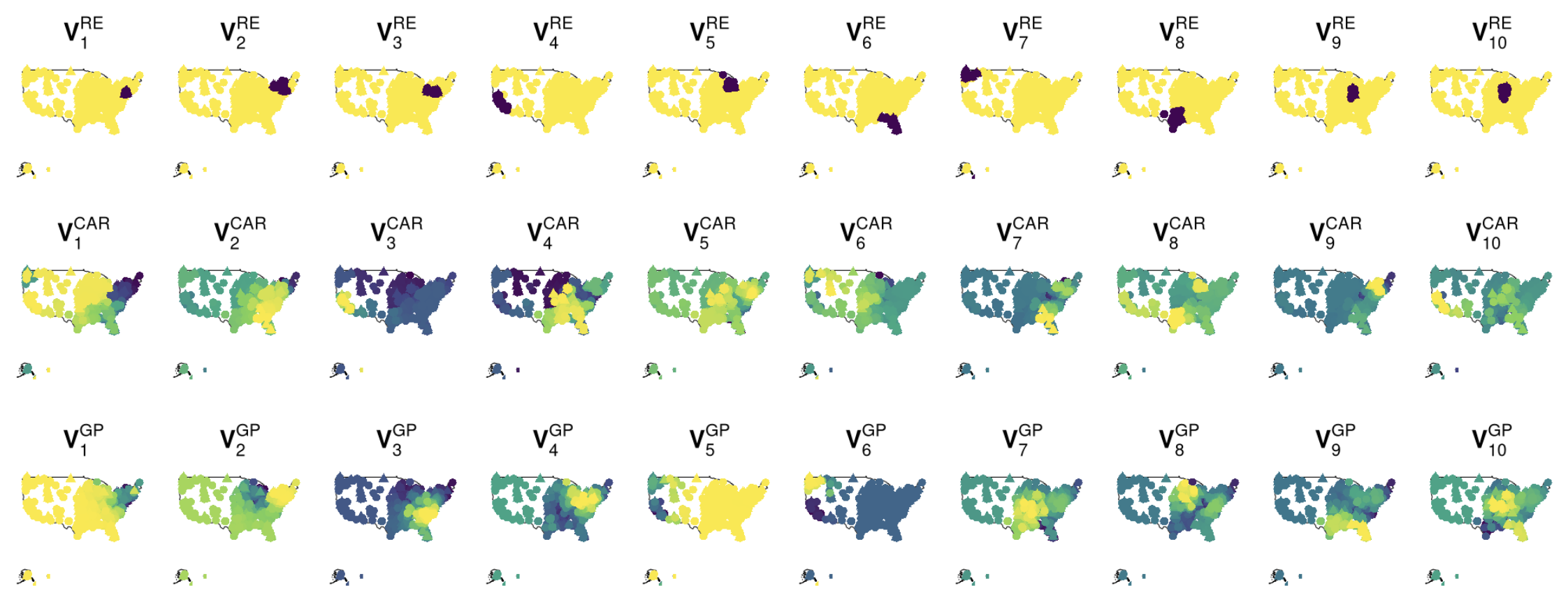}
    \caption{Thirty latent covariates—the ten highest-eigenvalue eigenvectors from $\bm{S}^{\text{RE}}$, $\bm{S}^{\text{CAR}}$, and $\bm{S}^{\text{GP}}$.}
    \label{fig:latent_covs}
\end{subfigure}

\vspace{0.4em}

\begin{subfigure}[t]{0.2\textwidth}
  \vspace{50pt}
  \captionsetup{justification=raggedright,singlelinecheck=false}
  \caption{Spatial weights $w_1,\ldots,w_n$.} 
  \label{fig:spatial_weights_map}
\end{subfigure}\hfill \begin{subfigure}[t]{0.75\textwidth}
  \vspace{0pt}
  \includegraphics[width=\linewidth]{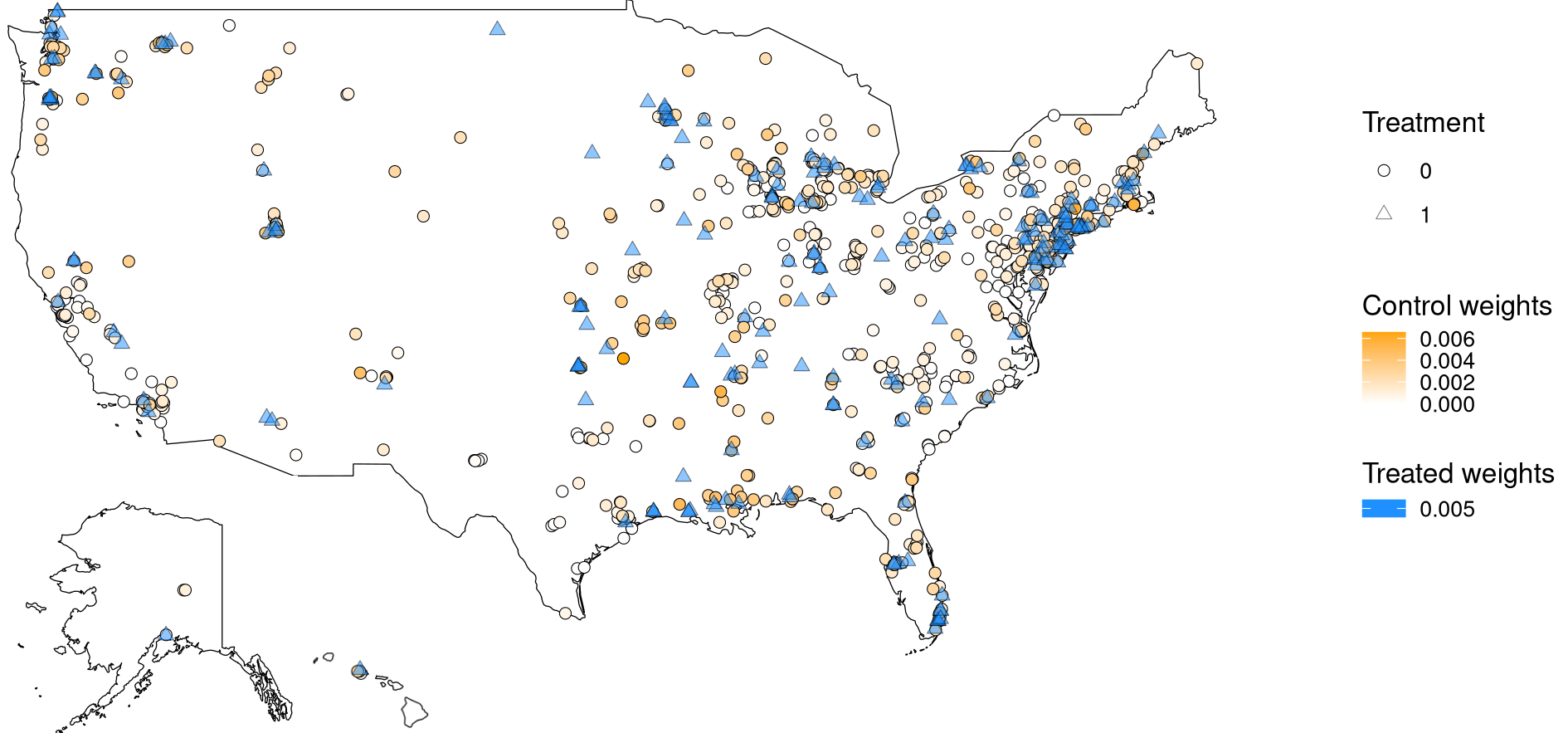}
\end{subfigure}

    \end{figure}

\vspace{-0.1in}
\section{Practical considerations}
\label{sec:guidance}
\vspace{-0.1in}
In general, we recommend that researchers examine the weights of the treated and control components of the spatial estimator to understand how information is used across space and to assess whether this pattern aligns with subject-matter knowledge or intuition. To support this, we conclude by describing how borrowing information across space involves navigating a set of interrelated tradeoffs.

First, borrowing strength across space reflects a fundamental bias–variance tradeoff. At one extreme, imposing tight balancing constraints on the latent covariates restricts  comparisons to the nearest treated and control units, so information is borrowed only locally. At the other extreme, relaxing these constraints effectively disregards the spatial adjustment performed by the latent covariates and borrows information globally. The former approach reduces bias from the targeted form of unmeasured spatial confounding at the cost of increased variance, whereas the latter approach lowers variance but is more susceptible to unmeasured spatial confounding bias. This result mirrors the analytical result in Supplementary Section \ref{sec:balance-dispersion} and aligns with the intuition behind donor pool selection in synthetic control analyses \citep{abadie2021using}. We therefore recommend that investigators (i) carefully examine the distribution of the weights in space and (ii) assess the sensitivity of point estimates and standard errors across a plausible range of balancing thresholds for the latent covariates.

Second, borrowing strength across space introduces a tension between adjusting for unmeasured spatial confounding bias and avoiding bias from spatial interference. When interference is present, using the most proximal control units as comparisons may be inappropriate, as these units could themselves be affected by the treatment received by nearby treated units. Previous studies have proposed selecting control units immediately outside a spatial buffer, so they are not treated to the treatment but still have similar values of the unmeasured spatial confounder \citep{butts2021difference}. In settings where interference is plausible, we recommend considering such buffered comparisons to reduce potential contamination. 

A third tradeoff introduced by borrowing strength across space relates to the positivity assumption and the target population. When borrowing strength only locally, some treated units may be too distant from any control units to permit meaningful comparisons. In such cases, we argue that these treated units should be excluded from the target population, as there is no justification for assuming the existence of control units with similar values of the unmeasured confounder (a violation of positivity). Thus, imposing strict balancing constraints on the latent covariates may require redefining the target population. In light of this, we recommend that researchers carefully consider the choice of target population before proceeding with any analysis.
All of these considerations are in the spirit of matching and weighting methods for observational studies \citep{zubizarreta2023handbook}. Embedding spatial-contrast considerations into routine practice can strengthen both the credibility and interpretability of spatial causal analyses. 

\vspace{-0.1in}
\section{Discussion}
\label{sec:discussion}
\vspace{-0.1in}

This paper unifies three classes of spatial regression models within a common weighting framework for causal inference. Our exact, finite-sample analysis reveals that spatial regression adjusts for unmeasured spatial confounding by approximately balancing a latent set of spatially autocorrelated covariates determined by the assumed spatial error structure. Furthermore, we show that the confounding bias induced by an unmeasured confounder $\bm{U}$ is small whenever $\bm{U}$ exhibits strong spatial autocorrelation as measured by Moran's I. We propose visual diagnostics to better understand this result, as well as a more general spatial weighting estimator that simultaneously adjusts for unmeasured confounders with different forms of spatial autocorrelation, which are typically considered separately (i.e., one at a time) in spatial regression. In an observational study of Superfund site remediation and Medicaid-covered birth outcomes, we find that remediation between 1991--2015 potentially decreased the rate of small vulnerable newborns in nearby areas during 2016--2018 (ATT estimate: $-0.60$\%).

Our main result demonstrates that incorporating a spatial error term is equivalent to augmenting the linear model with penalized eigenvectors of the covariance matrix. 
This formulation implicitly induces approximate mean balance on these latent covariates. 
In this sense, ``modeling the error'' and ``softly balancing eigenvectors of the covariance matrix'' are two sides of the same coin. 
While we develop these results in the context of spatial regression, they broadly apply to linear regression models with generalized error structures. 
In other words, under a linear model, the error covariance structure can be re-expressed as additional covariate structure with appropriate regularization, making these implicit adjustments explicit and more interpretable.

We highlight two avenues for future work. 
First, although we provide a finite-sample analysis that informs large-sample considerations, the development of a comprehensive asymptotic theory remains open. 
\citet{gilbert2025consistency} established consistency for Gaussian process regression, and \citet{datta2025consistent} derived consistency results for settings where the treatment, outcome, and unmeasured confounder are generated by Gaussian random fields, both under infill asymptotics. 
However, extending this theory to random effects or conditional autoregressive models presents additional challenges that warrant further investigation. 
As the sample size grows, new units may join existing clusters and form new adjacencies, so pre-existing entries of the covariance matrix change with sample size. 
Accounting for this evolving geometry is essential to any asymptotic treatment of these models.

Second, practitioners typically select the spatial hyperparameter $\rho^2$ in a Bayesian framework using outcome information. 
Here we treated $\rho^2$ and $\bm{\Sigma}$ as fixed and known a priori. {Our results elucidate how any choice of $\bm{\Sigma}$ adjusts for confounders in finite samples.} Supplementary Section \ref{sec:balance-dispersion} additionally establishes the role of $\rho^2$ in a balance-dispersion (hence bias-variance) tradeoff. Further investigation is needed to understand how outcome-driven hyperparameter selection, including approaches that leverage multiple outcomes \citep{dellavigna2025using}, can inform confounding adjustment in a frequentist framework. This includes the selection of $\rho^2$ and $\bm{\Sigma}$, or equivalently, eigenvectors and balance tolerances.

 \bibliography{main}

\begin{thebibliography}{50}
\newcommand{\enquote}[1]{``#1''}
\expandafter\ifx\csname natexlab\endcsname\relax\def\natexlab#1{#1}\fi

\bibitem[{Abadie(2021)}]{abadie2021using}
Abadie, A. (2021), \enquote{Using synthetic controls: Feasibility, data requirements, and methodological aspects,} \textit{Journal of Economic Literature}, 59, 391--425.

\bibitem[{Anselin(2002)}]{anselin2002under}
Anselin, L. (2002), \enquote{Under the hood issues in the specification and interpretation of spatial regression models,} \textit{Agricultural Economics}, 27, 247--267.

\bibitem[{Arbour et~al.(2021)Arbour, Ben-Michael, Feller, Franks, and Raphael}]{arbour2021using}
Arbour, D., Ben-Michael, E., Feller, A., Franks, A., and Raphael, S. (2021), \enquote{Using multitask gaussian processes to estimate the effect of a targeted effort to remove firearms,} \textit{arXiv preprint arXiv:2110.07006}.

\bibitem[{Ashorn et~al.(2023)Ashorn, Ashorn, Muthiani, Aboubaker, Askari, Bahl, Black, Dalmiya, Duggan, Hofmeyr, et~al.}]{ashorn2023small}
Ashorn, P., Ashorn, U., Muthiani, Y., Aboubaker, S., Askari, S., Bahl, R., Black, R.~E., Dalmiya, N., Duggan, C.~P., Hofmeyr, G.~J., et~al. (2023), \enquote{Small vulnerable newborns—big potential for impact,} \textit{The Lancet}, 401, 1692--1706.

\bibitem[{Ben-Michael et~al.(2021)Ben-Michael, Feller, Hirshberg, and Zubizarreta}]{ben2021balancing}
Ben-Michael, E., Feller, A., Hirshberg, D.~A., and Zubizarreta, J.~R. (2021), \enquote{The balancing act in causal inference,} \textit{arXiv preprint arXiv:2110.14831}.

\bibitem[{Brender et~al.(2011)Brender, Maantay, and Chakraborty}]{brender2011residential}
Brender, J.~D., Maantay, J.~A., and Chakraborty, J. (2011), \enquote{Residential proximity to environmental hazards and adverse health outcomes,} \textit{American Journal of Public Health}, 101, S37--S52.

\bibitem[{Bruns-Smith et~al.(2023)Bruns-Smith, Dukes, Feller, and Ogburn}]{bruns2023augmented}
Bruns-Smith, D., Dukes, O., Feller, A., and Ogburn, E.~L. (2023), \enquote{Augmented balancing weights as linear regression,} \textit{arXiv preprint arXiv:2304.14545}.

\bibitem[{Butts(2021)}]{butts2021difference}
Butts, K. (2021), \enquote{Difference-in-differences estimation with spatial spillovers,} \textit{arXiv preprint arXiv:2105.03737}.

\bibitem[{{Centers for Medicare \& Medicaid Services (CMS)}(2022)}]{cms_taf_ip}
{Centers for Medicare \& Medicaid Services (CMS)} (2022), \enquote{T-MSIS Analytic Files (TAF) Inpatient File,} \url{https://resdac.org/cms-data/files/taf-ip}.

\bibitem[{Chattopadhyay and Zubizarreta(2023)}]{chattopadhyay2023implied}
Chattopadhyay, A. and Zubizarreta, J.~R. (2023), \enquote{On the implied weights of linear regression for causal inference,} \textit{Biometrika}, 110, 615--629.

\bibitem[{Currie et~al.(2011)Currie, Greenstone, and Moretti}]{currie2011superfund}
Currie, J., Greenstone, M., and Moretti, E. (2011), \enquote{Superfund cleanups and infant health,} \textit{American Economic Review}, 101, 435--441.

\bibitem[{Datta and Stein(2025)}]{datta2025consistent}
Datta, A. and Stein, M.~L. (2025), \enquote{Consistent Infill Estimability of the Regression Slope Between Gaussian Random Fields Under Spatial Confounding,} \textit{arXiv preprint arXiv:2506.09267}.

\bibitem[{DellaVigna et~al.(2025)DellaVigna, Imbens, Kim, and Ritzwoller}]{dellavigna2025using}
DellaVigna, S., Imbens, G., Kim, W., and Ritzwoller, D.~M. (2025), \enquote{Using Multiple Outcomes to Adjust Standard Errors for Spatial Correlation,} Tech. rep., National Bureau of Economic Research.

\bibitem[{Dupont et~al.(2023)Dupont, Marques, and Kneib}]{dupont2023demystifying}
Dupont, E., Marques, I., and Kneib, T. (2023), \enquote{Demystifying Spatial Confounding,} \textit{arXiv preprint arXiv:2309.16861}.

\bibitem[{Earnest et~al.(2007)Earnest, Morgan, Mengersen, Ryan, Summerhayes, and Beard}]{earnest2007evaluating}
Earnest, A., Morgan, G., Mengersen, K., Ryan, L., Summerhayes, R., and Beard, J. (2007), \enquote{Evaluating the effect of neighbourhood weight matrices on smoothing properties of Conditional Autoregressive (CAR) models,} \textit{International Journal of Health Geographics}, 6, 1--12.

\bibitem[{{E.P.A.}(2022)}]{EPA2021Superfund}
{E.P.A.} (2022), \enquote{Population Surrounding 1,877 Superfund Sites,} Tech. rep., {US EPA, Office of Land and Emergency Management}, accessed: 2025-02-14.

\bibitem[{Gelman and Hill(2006)}]{gelman2006data}
Gelman, A. and Hill, J. (2006), \textit{Data analysis using regression and multilevel/hierarchical models}, Cambridge university press.

\bibitem[{Gilbert et~al.(2021)Gilbert, Datta, Casey, and Ogburn}]{gilbert2021causal}
Gilbert, B., Datta, A., Casey, J.~A., and Ogburn, E.~L. (2021), \enquote{A causal inference framework for spatial confounding,} \textit{arXiv preprint arXiv:2112.14946}.

\bibitem[{Gilbert et~al.(2025)Gilbert, Ogburn, and Datta}]{gilbert2025consistency}
Gilbert, B., Ogburn, E.~L., and Datta, A. (2025), \enquote{Consistency of common spatial estimators under spatial confounding,} \textit{Biometrika}, 112, asae070.

\bibitem[{Hamilton and Viscusi(1999)}]{hamilton1999costly}
Hamilton, J.~T. and Viscusi, W.~K. (1999), \enquote{How costly is “clean”? An analysis of the benefits and costs of Superfund site remediations,} \textit{Journal of Policy Analysis and Management}, 18, 2--27.

\bibitem[{Hanks et~al.(2015)Hanks, Schliep, Hooten, and Hoeting}]{hanks2015restricted}
Hanks, E.~M., Schliep, E.~M., Hooten, M.~B., and Hoeting, J.~A. (2015), \enquote{Restricted spatial regression in practice: geostatistical models, confounding, and robustness under model misspecification,} \textit{Environmetrics}, 26, 243--254.

\bibitem[{Henn et~al.(2016)Henn, Ettinger, Hopkins, Jim, Amarasiriwardena, Christiani, Coull, Bellinger, and Wright}]{henn2016prenatal}
Henn, B.~C., Ettinger, A.~S., Hopkins, M.~R., Jim, R., Amarasiriwardena, C., Christiani, D.~C., Coull, B.~A., Bellinger, D.~C., and Wright, R.~O. (2016), \enquote{Prenatal arsenic exposure and birth outcomes among a population residing near a mining-related superfund site,} \textit{Environmental Health Perspectives}, 124, 1308--1315.

\bibitem[{Hodges and Reich(2010)}]{hodges2010adding}
Hodges, J.~S. and Reich, B.~J. (2010), \enquote{Adding spatially-correlated errors can mess up the fixed effect you love,} \textit{The American Statistician}, 64, 325--334.

\bibitem[{Keele et~al.(2015)Keele, Titiunik, and Zubizarreta}]{keele2015a}
Keele, L., Titiunik, R., and Zubizarreta, J. (2015), \enquote{Enhancing a geographic regression discontinuity design through matching to estimate the effect of Ballot initiatives on voter turnout,} \textit{Journal of Royal Statistical Society A}, 178, 223–23.

\bibitem[{Keeler et~al.(2023)Keeler, Luben, Forestieri, Olshan, and Desrosiers}]{keeler2023residential}
Keeler, C., Luben, T.~J., Forestieri, N., Olshan, A.~F., and Desrosiers, T.~A. (2023), \enquote{Is residential proximity to polluted sites during pregnancy associated with preterm birth or low birth weight? Results from an integrated exposure database in North Carolina (2003--2015),} \textit{Journal of Exposure Science \& Environmental Epidemiology}, 33, 229--236.

\bibitem[{Khan and Berrett(2023)}]{khan2023re}
Khan, K. and Berrett, C. (2023), \enquote{Re-thinking Spatial Confounding in Spatial Linear Mixed Models,} \textit{arXiv preprint arXiv:2301.05743}.

\bibitem[{Kiel and Zabel(2001)}]{kiel2001estimating}
Kiel, K. and Zabel, J. (2001), \enquote{Estimating the economic benefits of cleaning up Superfund sites: The case of Woburn, Massachusetts,} \textit{The Journal of Real Estate Finance and Economics}, 22, 163--184.

\bibitem[{Kihal-Talantikite et~al.(2017)Kihal-Talantikite, Zmirou-Navier, Padilla, and Deguen}]{kihal2017systematic}
Kihal-Talantikite, W., Zmirou-Navier, D., Padilla, C., and Deguen, S. (2017), \enquote{Systematic literature review of reproductive outcome associated with residential proximity to polluted sites,} \textit{International Journal of Health Geographics}, 16, 1--39.

\bibitem[{Klemick et~al.(2020)Klemick, Mason, and Sullivan}]{klemick2020superfund}
Klemick, H., Mason, H., and Sullivan, K. (2020), \enquote{Superfund cleanups and children’s lead exposure,} \textit{Journal of Environmental Economics and Management}, 100, 102289.

\bibitem[{Langlois et~al.(2009)Langlois, Brender, Suarez, Zhan, Mistry, Scheuerle, and Moody}]{langlois2009maternal}
Langlois, P.~H., Brender, J.~D., Suarez, L., Zhan, F.~B., Mistry, J.~H., Scheuerle, A., and Moody, K. (2009), \enquote{Maternal residential proximity to waste sites and industrial facilities and conotruncal heart defects in offspring,} \textit{Paediatric and Perinatal Epidemiology}, 23, 321--331.

\bibitem[{Miller(2004)}]{miller2004tobler}
Miller, H.~J. (2004), \enquote{Tobler's first law and spatial analysis,} \textit{Annals of the Association of American Geographers}, 94, 284--289.

\bibitem[{Murray and Feller(2024)}]{murray2024unifying}
Murray, J.~S. and Feller, A. (2024), \enquote{A Unifying Weighting Perspective on Causal Machine Learning: Kernel Methods, Gaussian Processes, and Bayesian Tree Models,} YouTube video.

\bibitem[{Narcisi et~al.(2024)Narcisi, Greco, and Trivisano}]{narcisi2024effect}
Narcisi, M., Greco, F., and Trivisano, C. (2024), \enquote{On the effect of confounding in linear regression models: an approach based on the theory of quadratic forms,} \textit{Environmental and Ecological Statistics}, 31, 433--461.

\bibitem[{Nobre et~al.(2021)Nobre, Schmidt, and Pereira}]{nobre2021effects}
Nobre, W.~S., Schmidt, A.~M., and Pereira, J.~B. (2021), \enquote{On the effects of spatial confounding in hierarchical models,} \textit{International Statistical Review}, 89, 302--322.

\bibitem[{Ouidir et~al.(2020)Ouidir, Louis, Kanner, Grantz, Zhang, Sundaram, Rahman, Lee, Kannan, Tekola-Ayele, et~al.}]{ouidir2020association}
Ouidir, M., Louis, G. M.~B., Kanner, J., Grantz, K.~L., Zhang, C., Sundaram, R., Rahman, M.~L., Lee, S., Kannan, K., Tekola-Ayele, F., et~al. (2020), \enquote{Association of maternal exposure to persistent organic pollutants in early pregnancy with fetal growth,} \textit{JAMA Pediatrics}, 174, 149--161.

\bibitem[{Paciorek(2010)}]{paciorek2010a}
Paciorek, C. (2010), \enquote{The Importance of Scale for Spatial-Confounding Bias and Precision of Spatial Regression Estimators,} \textit{Statistical Science}, 25, 107 – 125.

\bibitem[{Page et~al.(2017)Page, Liu, He, and Sun}]{page2017estimation}
Page, G.~L., Liu, Y., He, Z., and Sun, D. (2017), \enquote{Estimation and prediction in the presence of spatial confounding for spatial linear models,} \textit{Scandinavian Journal of Statistics}, 44, 780--797.

\bibitem[{Papadogeorgou et~al.(2019)Papadogeorgou, Choirat, and Zigler}]{papadogeorgou2019a}
Papadogeorgou, G., Choirat, C., and Zigler, C. (2019), \enquote{Adjusting for unmeasured spatial confounding with distance adjusted propensity score matching,} \textit{Biostatistics}, 20, 256–272.

\bibitem[{Perera et~al.(2003)Perera, Rauh, Tsai, Kinney, Camann, Barr, Bernert, Garfinkel, Tu, Diaz, et~al.}]{perera2003effects}
Perera, F.~P., Rauh, V., Tsai, W.-Y., Kinney, P., Camann, D., Barr, D., Bernert, T., Garfinkel, R., Tu, Y.-H., Diaz, D., et~al. (2003), \enquote{Effects of transplacental exposure to environmental pollutants on birth outcomes in a multiethnic population.} \textit{Environmental Health Perspectives}, 111, 201--205.

\bibitem[{Petrie(2006)}]{petrie2006environmental}
Petrie, M. (2006), \enquote{Environmental justice in the south: An analysis of the determinants and consequences of community involvement in superfund,} \textit{Sociological Spectrum}, 26, 471--489.

\bibitem[{Porpora et~al.(2019)Porpora, Piacenti, Scaramuzzino, Masciullo, Rech, and Benedetti~Panici}]{porpora2019environmental}
Porpora, M.~G., Piacenti, I., Scaramuzzino, S., Masciullo, L., Rech, F., and Benedetti~Panici, P. (2019), \enquote{Environmental contaminants exposure and preterm birth: a systematic review,} \textit{Toxics}, 7, 11.

\bibitem[{Reich et~al.(2021)Reich, Yang, Guan, Giffin, Miller, and Rappold}]{reich2021a}
Reich, B., Yang, S., Guan, Y., Giffin, A., Miller, M., and Rappold, A. (2021), \enquote{A review of spatial causal inference methods for environmental and epidemiological applications,} \textit{International Statistical Review}, 89, 605–634.

\bibitem[{Robins et~al.(2007)Robins, Sued, Lei-Gomez, and Rotnitzky}]{robins2007comment}
Robins, J., Sued, M., Lei-Gomez, Q., and Rotnitzky, A. (2007), \enquote{Comment: Performance of double-robust estimators when" inverse probability" weights are highly variable,} \textit{Statistical Science}, 22, 544--559.

\bibitem[{Schnell and Papadogeorgou(2020)}]{schnell2020a}
Schnell, P. and Papadogeorgou, G. (2020), \enquote{Mitigating unobserved spatial confounding when estimating the effect of supermarket access on cardiovascular disease deaths,} \textit{The Annals of Applied Statistics}, 14, 2069–2095.

\bibitem[{Schroeder et~al.(2025)Schroeder, Van~Riper, Manson, Knowles, Kugler, Roberts, and Ruggles}]{Schroeder2025NHGIS}
Schroeder, J., Van~Riper, D., Manson, S., Knowles, K., Kugler, T., Roberts, F., and Ruggles, S. (2025), \enquote{IPUMS National Historical Geographic Information System: Version 20.0,} Accessed: 2025-10-31.

\bibitem[{Stillerman et~al.(2008)Stillerman, Mattison, Giudice, and Woodruff}]{stillerman2008environmental}
Stillerman, K.~P., Mattison, D.~R., Giudice, L.~C., and Woodruff, T.~J. (2008), \enquote{Environmental exposures and adverse pregnancy outcomes: a review of the science,} \textit{Reproductive Sciences}, 15, 631--650.

\bibitem[{{U.S. Census Bureau}(1991)}]{USCensus1990}
{U.S. Census Bureau} (1991), \enquote{1990 Decennial Census, Summary Tape Files 1 and 3A (STF 1 and 3A),} .

\bibitem[{Wang and Zubizarreta(2020)}]{wang2020minimal}
Wang, Y. and Zubizarreta, J.~R. (2020), \enquote{Minimal dispersion approximately balancing weights: asymptotic properties and practical considerations,} \textit{Biometrika}, 107, 93--105.

\bibitem[{Zubizarreta(2015)}]{zubizarreta2015stable}
Zubizarreta, J.~R. (2015), \enquote{Stable weights that balance covariates for estimation with incomplete outcome data,} \textit{Journal of the American Statistical Association}, 110, 910--922.

\bibitem[{Zubizarreta et~al.(2023)Zubizarreta, Stuart, Small, and Rosenbaum}]{zubizarreta2023handbook}
Zubizarreta, J.~R., Stuart, E.~A., Small, D.~S., and Rosenbaum, P.~R. (2023), \textit{Handbook of matching and weighting adjustments for causal inference}, CRC Press.

\end{thebibliography}


\begin{thebibliography}{18}
\newcommand{\enquote}[1]{``#1''}
\expandafter\ifx\csname natexlab\endcsname\relax\def\natexlab#1{#1}\fi

\bibitem[{Auty et~al.(2024)Auty, Daw, Admon, and Gordon}]{auty2024comparing}
Auty, S.~G., Daw, J.~R., Admon, L.~K., and Gordon, S.~H. (2024), \enquote{Comparing approaches to identify live births using the Transformed Medicaid Statistical Information System,} \textit{Health Services Research}, 59, e14233.

\bibitem[{{Centers for Medicare \& Medicaid Services (CMS)}(2022)}]{cms_taf_ip}
{Centers for Medicare \& Medicaid Services (CMS)} (2022), \enquote{T-MSIS Analytic Files (TAF) Inpatient File,} \url{https://resdac.org/cms-data/files/taf-ip}.

\bibitem[{Chattopadhyay and Zubizarreta(2023)}]{chattopadhyay2023implied}
Chattopadhyay, A. and Zubizarreta, J.~R. (2023), \enquote{On the implied weights of linear regression for causal inference,} \textit{Biometrika}, 110, 615--629.

\bibitem[{Clark and Linzer(2015)}]{clark2015should}
Clark, T.~S. and Linzer, D.~A. (2015), \enquote{Should I use fixed or random effects?} \textit{Political Science Research and Methods}, 3, 399--408.

\bibitem[{Gelman and Hill(2006)}]{gelman2006data}
Gelman, A. and Hill, J. (2006), \textit{Data analysis using regression and multilevel/hierarchical models}, Cambridge university press.

\bibitem[{Gilbert et~al.(2021)Gilbert, Datta, Casey, and Ogburn}]{gilbert2021causal}
Gilbert, B., Datta, A., Casey, J.~A., and Ogburn, E.~L. (2021), \enquote{A causal inference framework for spatial confounding,} \textit{arXiv preprint arXiv:2112.14946}.

\bibitem[{Lawn et~al.(2023)Lawn, Ohuma, Bradley, Idueta, Hazel, Okwaraji, Erchick, Yargawa, Katz, Lee, et~al.}]{lawn2023small}
Lawn, J.~E., Ohuma, E.~O., Bradley, E., Idueta, L.~S., Hazel, E., Okwaraji, Y.~B., Erchick, D.~J., Yargawa, J., Katz, J., Lee, A.~C., et~al. (2023), \enquote{Small babies, big risks: global estimates of prevalence and mortality for vulnerable newborns to accelerate change and improve counting,} \textit{The Lancet}, 401, 1707--1719.

\bibitem[{Marden et~al.(1964)Marden, Smith, and McDonald}]{marden1964congenital}
Marden, P.~M., Smith, D.~W., and McDonald, M.~J. (1964), \enquote{Congenital anomalies in the newborninfant, including minor variations: A study of 4,412 babies by surface examination for anomalies and buccal smear for sex chromatin,} \textit{The Journal of Pediatrics}, 64, 357--371.

\bibitem[{Mercatanti and Li(2014)}]{mercatanti2014debit}
Mercatanti, A. and Li, F. (2014), \enquote{Do debit cards increase household spending? Evidence from a semiparametric causal analysis of a survey,} .

\bibitem[{Moodie et~al.(2018)Moodie, Saarela, and Stephens}]{moodie2018doubly}
Moodie, E.~E., Saarela, O., and Stephens, D.~A. (2018), \enquote{A doubly robust weighting estimator of the average treatment effect on the treated,} \textit{Stat}, 7, e205.

\bibitem[{Schroeder et~al.(2025)Schroeder, Van~Riper, Manson, Knowles, Kugler, Roberts, and Ruggles}]{Schroeder2025NHGIS}
Schroeder, J., Van~Riper, D., Manson, S., Knowles, K., Kugler, T., Roberts, F., and Ruggles, S. (2025), \enquote{IPUMS National Historical Geographic Information System: Version 20.0,} Accessed: 2025-10-31.

\bibitem[{{U.S. Census Bureau}(1991)}]{USCensus1990}
{U.S. Census Bureau} (1991), \enquote{1990 Decennial Census, Summary Tape Files 1 and 3A (STF 1 and 3A),} .

\bibitem[{{US EPA}(2025{\natexlab{a}})}]{EPA2025COCSpreadsheet}
{US EPA} (2025{\natexlab{a}}), \enquote{Contaminant of Concern Data for Decision Documents by Media, FYs 1981-2024 (Final NPL, Deleted NPL, and Superfund Alternative Approach Sites),} \url{https://www.epa.gov/superfund/superfund-data-and-reports}.

\bibitem[{{US EPA}(2025{\natexlab{b}})}]{EPA2025SearchSitesWhereYouLive}
--- (2025{\natexlab{b}}), \enquote{Search for Superfund Sites Where You Live,} \url{https://www.epa.gov/superfund/search-superfund-sites-where-you-live}.

\bibitem[{{US EPA}(2025{\natexlab{c}})}]{EPA2025cleanup}
--- (2025{\natexlab{c}}), \enquote{Superfund Cleanup Process,} \url{https://www.epa.gov/superfund/superfund-cleanup-process}, accessed: 2025-10-31.

\bibitem[{{US EPA, Office of Mission Support}(2025)}]{EPA2025NPLBoundaries}
{US EPA, Office of Mission Support} (2025), \enquote{NPL Superfund Site Boundaries (EPA Public),} \url{https://www.arcgis.com/home/item.html?id=d6e1591d9a424f1fa6d95a02095a06d7}.

\bibitem[{Wang and Zubizarreta(2020)}]{wang2020minimal}
Wang, Y. and Zubizarreta, J.~R. (2020), \enquote{Minimal dispersion approximately balancing weights: asymptotic properties and practical considerations,} \textit{Biometrika}, 107, 93--105.

\bibitem[{Zubizarreta(2015)}]{zubizarreta2015stable}
Zubizarreta, J.~R. (2015), \enquote{Stable weights that balance covariates for estimation with incomplete outcome data,} \textit{Journal of the American Statistical Association}, 110, 910--922.

\end{thebibliography}
 \bibliographystyle{asa}

\newpage

\bigskip

\singlespacing 
\begin{center}
\Large 
    \textbf{Supplementary Materials for} ``Understanding Spatial Regression Models from a Weighting Perspective in an Observational Study of Superfund Remediation''
\end{center}
\doublespacing
\vspace{1em}

\setcounter{section}{0}

\renewcommand{\thesection}{\Alph{section}}
\renewcommand{\thesubsection}{\thesection.\arabic{subsection}}

\makeatletter
\renewcommand{\theHsection}{supp.\Alph{section}}
\renewcommand{\theHsubsection}{supp.\Alph{section}.\arabic{subsection}}
\makeatother

\section{Proof of Propositions \ref{thm:weighting}, \ref{thm:minimal-weighting}, \ref{thm:imbalance_spatialauto}}
\label{sec:weighting-proof}
\subsection*{Proof of Proposition \ref{thm:weighting}}

\begin{proof}
    By the Frisch-Waugh-Lovell Theorem, or block matrix inversion,
\begin{align*}
    \hat{\tau}_{GLS} &= \frac{\bm{Z}^T\bm{\Sigma}^{-1}(I_n -  \bm{X}(\bm{X}^T \bm{\Sigma}^{-1} \bm{X})^{-1} \bm{X}^T\bm{\Sigma}^{-1})}{\bm{Z}^T\bm{\Sigma}^{-1}(I_n- \bm{X}(\bm{X}^T \bm{\Sigma}^{-1} \bm{X})^{-1} \bm{X}^T\bm{\Sigma}^{-1})\bm{Z}}\bm{Y}\\
    &=  \sum_{i:Z_i = 1} w_i Y_i - \sum_{i:Z_i = 0} w_i Y_i
\end{align*}for 
\begin{align*}
    (w_1, \ldots, w_n)^T &= \bm{M} \frac{(\bm{I}_n -  \bm{\Sigma}^{-1}\bm{X}(\bm{X}^T \bm{\Sigma}^{-1} \bm{X})^{-1} \bm{X}^T)\bm{\Sigma}^{-1}\bm{Z}}{\bm{Z}^T\bm{\Sigma}^{-1}(\bm{I}_n- \bm{X}(\bm{X}^T \bm{\Sigma}^{-1} \bm{X})^{-1} \bm{X}^T\bm{\Sigma}^{-1})\bm{Z}}.
\end{align*}
\end{proof}

\subsection*{Proof of Proposition \ref{thm:minimal-weighting}}
\begin{proof}
    Let $\bm{M}$ be the diagonal matrix with $(i,i)$ entry equal to $M_{ii} = 2Z_i - 1$. Since $\bm{\Sigma} = \sigma^2 \bm{I}_n + \rho^2 \bm{S}$, the quadratic programming is equivalent to:
    \begin{align*}
        &\min_{\bm{w}} \bm{w}^T \bm{M} \bm{\Sigma} \bm{M} \bm{w}\\
        &\text{subject to} \begin{cases} \sum_{i:Z_i = 1} w_i = \sum_{i:Z_i = 0} w_i = 1\\
        \sum_{i:Z_i = 1} w_i \bm{X}_i = \sum_{i:Z_i = 0} w_i \bm{X}_i
        \end{cases}
    \end{align*}
    Let $\bm{l} = \bm{M} \bm{w}$. 
    In terms of $\bm{l}$, 
    the quadratic programming problem takes the form 
    \begin{center}
    $\min_{\bm{l}} 
    \bm{l}^T\bm{\Sigma}  \bm{l}$ \\
     
    \text{subject to } $\begin{cases} \bm{Z}^T \bm{l} = 1\\
    \bm{X}^T \bm{l} = \bm{0}_p
   \end{cases}$. 
   \end{center}Note that the constraint $\sum_{i:Z_i = 0} w_i = -(\bm{1}-\bm{Z})^T \bm{l} = 1$ is implied by the two constraints since the design matrix $\bm{X}$ includes an intercept. 

The Lagrangian of this quadratic programming problem is 
\begin{align*}
    \mathcal{L}(\bm{l}, \lambda_1, \bm{\lambda}_2) &= \bm{l}^T \bm{\Sigma}   \bm{l} + \lambda_1(\bm{Z}^T \bm{l} - 1) - \bm{\lambda}_2^T (\bm{X}^T \bm{l}).
\end{align*}Computing the partial derivatives $\frac{\partial \mathcal{L}}{d\bm{l}}, \frac{\partial \mathcal{L}}{d\bm{\lambda_1}}, \frac{\partial \mathcal{L}}{d \lambda_2}$ and equating them to $0$, we obtain the three equations \begin{align}
    \bm{0}_n &= 2\Sigma \bm{l}  + \lambda_1 \bm{Z} - \bm{X} \lambda_2\\
    0 &= \bm{Z}^T \bm{l} - 1\\
    \bm{0}_p &= \bm{X}^T \bm{l}.
\end{align}
The first equation can be rewritten as $$\bm{l} = \frac{\bm{\Sigma}^{-1} \bm{X}\lambda_2 - \lambda_1 \bm{\Sigma}^{-1} \bm Z}{2}.$$
For ease of notation, define $a = \bm{Z}^T \bm{\Sigma}^{-1} \bm{Z}$ and $b = \bm{Z}^T \bm{\Sigma}^{-1} \bm{X} (\bm{X}^T \bm{\Sigma}^{-1} \bm{X})^{-1} \bm{X}^T \bm{\Sigma}^{-1} \bm{Z}$. Combining the first equation with the second we obtain
\begin{align*}
    \frac{\bm{Z}^T \bm{\Sigma}^{-1} \bm{X} \lambda_2 - \lambda_1 \bm{Z}^T \bm{\Sigma}^{-1} \bm{Z}}{2} &= 1\\
    \implies \lambda_1 &= \frac{\bm{Z}^T \bm{\Sigma}^{-1} \bm{X}\lambda_2 - 2}{a}.
\end{align*}
Combining the first equation with the third, and substituting $\lambda_1$ with the expression above, we obtain
\begin{align*}
    \bm{X}^T \bm{\Sigma}^{-1} \bm{X} \lambda_2 - \lambda_1 \bm{X}^T \bm{\Sigma}^{-1} \bm{Z} &= 0\\
    \bm{X}^T \bm{\Sigma}^{-1} \bm{X} \lambda_2 - \bigg(\frac{\bm{Z}^T \bm{\Sigma}^{-1} \bm{X}\lambda_2 - 2}{a}\bigg) \bm{X}^T \bm{\Sigma}^{-1} \bm{Z} &= 0\\
    \lambda_2 &= \frac{2}{a} \bm{W}^{-1} \bm{X}^T \bm{\Sigma}^{-1} \bm{Z},
\end{align*}where $$\bm{W} = \frac{\bm{X}^T \bm{\Sigma}^{-1} \bm{Z} \bm{Z}^T \bm{\Sigma}^{-1} \bm{X}}{a} - \bm{X}^T \bm{\Sigma}^{-1} \bm{X}.$$
By the Woodbury identity, 
\begin{align*}
    \bm{W}^{-1} &= -(\bm{X}^T \bm{\Sigma}^{-1} \bm{X})^{-1} -\frac{(\bm{X}^T \bm{\Sigma}^{-1} \bm{X})^{-1} \bm{X}^T \bm{\Sigma}^{-1} \bm{Z} \bm{Z}^T \bm{\Sigma}^{-1} \bm{X}(\bm{X}^T \bm{\Sigma}^{-1} \bm{X})^{-1}}{\bm{Z}^T\bm{\Sigma}^{-1}(I_n- \bm{X}(\bm{X}^T \bm{\Sigma}^{-1} \bm{X})^{-1} \bm{X}^T\bm{\Sigma}^{-1})\bm{Z}}\\
    &= -(\bm{X}^T \bm{\Sigma}^{-1} \bm{X})^{-1} -\frac{(\bm{X}^T \bm{\Sigma}^{-1} \bm{X})^{-1} \bm{X}^T \bm{\Sigma}^{-1} \bm{Z} \bm{Z}^T \bm{\Sigma}^{-1} \bm{X}(\bm{X}^T \bm{\Sigma}^{-1} \bm{X})^{-1}}{a-b}.
\end{align*}
Substituting $\lambda_1, \lambda_2$ in the first equation with the expressions above, we obtain \begin{align*}
    \bm{l} &= \frac{\bm{\Sigma}^{-1} \bm{X}\lambda_2 - \lambda_1 \bm{\Sigma}^{-1} \bm Z}{2}\\
    &= \frac{\bm{\Sigma}^{-1} \bm{X} \bigg(\frac{2}{a} \bm{W}^{-1} \bm{X}^T \bm{\Sigma}^{-1} \bm{Z}\bigg)}{2} - \frac{(\bm{Z}^T \bm{\Sigma}^{-1} \bm{X}\bigg(\frac{2}{a} \bm{W}^{-1} \bm{X}^T \bm{\Sigma}^{-1} \bm{Z}\bigg) - 2)\bm{\Sigma}^{-1} \bm{Z}}{a}\\
    &= -\frac{\bm{\Sigma}^{-1} \bm{X}(\bm{X}^T \bm{\Sigma}^{-1} \bm{X})^{-1} \bm{X}^T \bm{\Sigma}^{-1} \bm{Z}}{a} - \frac{\bm{\Sigma}^{-1} \bm{X}(\bm{X}^T \bm{\Sigma}^{-1} \bm{X})^{-1} \bm{X}^T \bm{\Sigma}^{-1} \bm{Z} \bm{Z}^T \bm{\Sigma}^{-1} \bm{X}(\bm{X}^T \bm{\Sigma}^{-1} \bm{X})^{-1} \bm{X}^T \bm{\Sigma}^{-1} \bm{Z}}{a(a-b)}\\
    &\hspace{0.5cm} + \frac{\bm{Z}^T \bm{\Sigma}^{-1} \bm{X}(\bm{X}^T \bm{\Sigma}^{-1} \bm{X})^{-1} \bm{X}^T\bm{\Sigma}^{-1} \bm{Z}} {a^2}\bm{\Sigma}^{-1} \bm{Z} \\
    &\hspace{1cm}+ \frac{\bm{Z}^T \bm{\Sigma}^{-1} \bm{X}(\bm{X}^T \bm{\Sigma}^{-1} \bm{X})^{-1} \bm{X}^T \bm{\Sigma}^{-1} \bm{Z} \bm{Z}^T \bm{\Sigma}^{-1} \bm{X}(\bm{X}^T \bm{\Sigma}^{-1} \bm{X})^{-1} \bm{X}^T \bm{\Sigma}^{-1} \bm{Z}}{a^2(a-b)}\bm{\Sigma}^{-1} \bm{Z} + \frac{\bm{\Sigma}^{-1} \bm{Z}}{a}\\
    &= -\frac{\bm{\Sigma}^{-1} \bm{X}(\bm{X}^T \bm{\Sigma}^{-1} \bm{X})^{-1} \bm{X}^T \bm{\Sigma}^{-1} \bm{Z}}{a} - \frac{\bm{\Sigma}^{-1} \bm{X}(\bm{X}^T \bm{\Sigma}^{-1} \bm{X})^{-1} \bm{X}^T \bm{\Sigma}^{-1} \bm{Z} b}{a(a-b)} + \frac{b} {a^2}\bm{\Sigma}^{-1} \bm{Z} \\
    &\hspace{0.5in}+ \frac{b^2}{a^2(a-b)}\bm{\Sigma}^{-1} \bm{Z} + \frac{\bm{\Sigma}^{-1} \bm{Z}}{a}\\
    &= \frac{\bm{\Sigma}^{-1} \bm{Z} - \bm{\Sigma}^{-1} \bm{X}(\bm{X}^T \bm{\Sigma}^{-1} \bm{X})^{-1} \bm{X}^T \bm{\Sigma}^{-1} \bm{Z}}{a-b},
\end{align*}
Left-multiplying the expression in line (17) by $\bm{M}$ yields the implied weights of the generalized least squares, completing the proof. 
\end{proof}

\begin{corollary}[General $\bm{\Sigma}$]
\label{corollary_general}
When $\bm{\Sigma}$ does not have a natural decomposition $\bm{\Sigma}= \sigma^2 \bm{I}_n + \rho^2 \bm{S}$ and is instead an arbitrary positive semidefinite matrix, the weights of the GLS estimator correspond to the solution of the following quadratic programming problem:
\begin{align*}
    &\min_{\bm{w}} \bigg\{  \sum_{k = 1}^n \lambda_k \bigg(\sum_{i:Z_i = 1} w_i v_{ki} - \sum_{i:Z_i = 0} w_i v_{ki}\bigg)^2 \bigg\}\\
    &\text{subject to } \begin{cases} \sum_{i:Z_i=1} w_i = 1, \sum_{i:Z_i=0} w_i = 1\\
    \sum_{i: Z_i = 1} w_i \bm{X}_i = \sum_{i:Z_i = 0} w_i \bm{X}_i
   \end{cases} 
   \end{align*}
   where $\bm{v}_1, \ldots, \bm{v}_n$ are the eigenvectors of $\bm{\Sigma}$ with corresponding eigenvalues $\lambda_1 \geq \ldots \geq \lambda_n \geq 0$. 
\end{corollary}

\subsection*{Proof of Proposition \ref{thm:imbalance_spatialauto}}
\begin{proof}
    
Assume the same notation as in the Proof of Proposition \ref{thm:minimal-weighting}, i.e. $\bm{l} = \bm{M} \bm{w}$ where $\bm{M}$ is the diagonal matrix with $(i,i)$ entry equal to $M_{ii} = 2Z_i - 1$. 

The Global Moran's I statistic of $\bm{U}$ relative to the spatial covariance matrix $\bm{S}$ can be rewritten as $$\mathcal{I}(\bm{U};\bm{S}) = \frac{(\bm{U}-\bar{U}\bm{1})^T \bm{S} (\bm{U}-\bar{U}\bm{1})}{\lambda_1 (\bm{U}-\bar{U}\bm{1})^T (\bm{U}-\bar{U}\bm{1})}.$$

The absolute mean imbalance for $\bm{U}$ is \begin{align*}
 &\sqrt{(\sum_{i:Z_i = 1} w_i U_i - \sum_{i:Z_i = 0} w_i U_i)^2}\\
    &= \sqrt{(\bm{l}^T \bm{U})^2} \\
    &=  \sqrt{(\bm{l}^T (\bm{U}-\bar{U}\bm{1}))^2} \\
    &= \sqrt{((\bm{\sqrt{\bm{\Sigma}}l})^T \sqrt{\bm{\Sigma}^{-1}}(\bm{U}-\bar{U}\bm{1}))^2} \\
    &\leq \sqrt{(\bm{l}^T \bm{\Sigma} \bm{l}) (\bm{U}-\bar{U}\bm{1})^T \bm{\Sigma}^{-1} (\bm{U}-\bar{U}\bm{1})}\\
    &= \sqrt{c_0 (\bm{U}-\bar{U}\bm{1})^T \bm{\Sigma}^{-1} (\bm{U}-\bar{U}\bm{1})}\\
    &\leq \sqrt{c_0 \frac{((\bm{U}-\bar{U}\bm{1})^T (\bm{U}-\bar{U}\bm{1}))^2}{(\bm{U}-\bar{U}\bm{1})^T\bm{\Sigma}(\bm{U}-\bar{U}\bm{1})} \frac{(\lambda_{\text{min}}(\bm{\Sigma}) + \lambda_{\text{max}}(\bm{\Sigma}))^2}{4 \lambda_{\text{min}}(\bm{\Sigma})\lambda_{\text{max}}(\bm{\Sigma})}}\\
    &= \sqrt{c_0 \frac{((\bm{U}-\bar{U}\bm{1})^T (\bm{U}-\bar{U}\bm{1}))^2}{\sigma^2 (\bm{U}-\bar{U}\bm{1})^T(\bm{U}-\bar{U}\bm{1}) + \rho^2 (\bm{U}-\bar{U}\bm{1})^T \bm{S} (\bm{U}-\bar{U}\bm{1})} \frac{(2\sigma^2 + \rho^2 \lambda_1 + \rho^2 \lambda_n)^2}{4 (\sigma^2 + \rho^2 \lambda_1)(\sigma^2 + \rho^2 \lambda_n)}}\\
    &= \sqrt{c_0 \frac{1}{\sigma^2  + \rho^2 \lambda_1 \mathcal{I}(\bm{U};\bm{S})} \frac{(2\sigma^2 + \rho^2 \lambda_1 + \rho^2 \lambda_n)^2}{4 (\sigma^2 + \rho^2 \lambda_1)(\sigma^2 + \rho^2 \lambda_n)}}\sqrt{(\bm{U}-\bar{U}\bm{1})^T (\bm{U}-\bar{U}\bm{1})}
\end{align*}by the Cauchy-Schwarz inequality (line 5) and the Kantorovich inequality (line 7). 

By Assumptions \ref{consistency}--\ref{ignorability}, and the assumption of a linear outcome model, the conditional bias is \begin{align*}
    \E(\hat{\tau}_\text{GLS}|\bm{Z}, \bm{X}, \bm{U}) - \tau &= \sum_{i:Z_i = 1} w_i Y_i - \sum_{i:Z_i = 0} w_i Y_i\\
    &= \sum_{i:Z_i = 1} w_i \E(Y_i|\bm{Z}, \bm{X}, \bm{U}) - \sum_{i:Z_i = 0} w_i \E(Y_i|\bm{Z}, \bm{X}, \bm{U}) - \tau\\
    &= \sum_{i:Z_i = 1} w_i (\bm{\beta}^T \bm{X}_i + \tau + \gamma U_i) - \sum_{i:Z_i = 0} w_i (\bm{\beta}^T \bm{X}_i + \gamma U_i) - \tau\\
    &= \gamma \bigg(\sum_{i:Z_i = 1} w_i U_i - \sum_{i:Z_i = 0} w_i U_i\bigg)
\end{align*}since the weights exactly balance the measured covariates in mean and sum to $1$ in each treatment group (Proposition \ref{thm:minimal-weighting}). 

Therefore, the absolute conditional bias is \begin{align*}
    \bigg\vert \E(\hat{\tau}_\text{GLS}|\bm{Z}, \bm{X}, \bm{U}) - \tau  \bigg \vert &= \bigg \vert \gamma \bigg(\sum_{i:Z_i = 1} w_i U_i - \sum_{i:Z_i = 0} w_i U_i\bigg) \bigg \vert\\
    &\leq |\gamma| \sqrt{c_0 \frac{(2\sigma^2 + \rho^2\lambda_1 + \rho^2 \lambda_n)^2}{4(\sigma^2 + \rho^2 \lambda_1 \mathcal{I}(\bm{U};\bm{S}) (\sigma^2 + \rho^2 \lambda_1)(\sigma^2 + \rho^2 \lambda_n)}}\sqrt{\sum_{i = 1}^n (U_i - \bar{U})^2}.
\end{align*}
\end{proof}

\section{Impact of the spatial hyperparameter $\rho^2$ on causal inference}
\label{sec:hyperparameter}
\subsection*{A balance-dispersion tradeoff}
\label{sec:balance-dispersion}
The quadratic programming problem of Proposition \ref{thm:minimal-weighting} admits an equivalent formulation with a soft-balancing constraint on the sum of mean eigenvector imbalances. By duality, the problem is equivalent to \begin{align*}
    &\min_{\bm{w}} \sigma^2 \sum_{i = 1}^n w_i^2 \\
    &\text{subject to } \begin{cases} 
    \sum_{k = 1}^n \lambda_k \bigg(\sum_{i:Z_i = 1} w_i v_{ki} - \sum_{i:Z_i = 0} w_i v_{ki}\bigg)^2 \leq \Delta(\rho^2), \\
    \sum_{i:Z_i = 1} w_i = 1, \sum_{i:Z_i = 0} w_i = 1,\\
    \sum_{i:Z_i = 1} w_i \bm{X}_i = \sum_{i:Z_i = 0} w_i \bm{X}_i\end{cases}
\end{align*}where $$\Delta(\rho^2) = \frac{\bm{Z}^T \bm{\Sigma}^{-1} (\bm{I}_n - \bm{X} (\bm{X}^T \bm{\Sigma}^{-1} \bm{X})^{-1} \bm{X}^T \bm{\Sigma}^{-1})\bm{S} (\bm{I}_n -  \bm{\Sigma}^{-1}\bm{X}(\bm{X}^T \bm{\Sigma}^{-1} \bm{X})^{-1} \bm{X}^T)\bm{\Sigma}^{-1}\bm{Z}}{(\bm{Z}^T\bm{\Sigma}^{-1}(\bm{I}_n- \bm{X}(\bm{X}^T \bm{\Sigma}^{-1} \bm{X})^{-1} \bm{X}^T\bm{\Sigma}^{-1})\bm{Z})^2}$$where $\Delta(\rho^2)$ depends on $\rho^2$ through $\bm{\Sigma}^{-1} = (\sigma^2 \bm{I}_n + \rho^2 \bm{S})^{-1}$. 

Furthermore, define weight dispersion as \begin{align*}
    D(\rho^2) = \sum_{i = 1}^n w_i^2 &= \frac{\bm{Z}^T \bm{\Sigma}^{-1} (\bm{I}_n - \bm{X}(\bm{X}^T \bm{\Sigma}^{-1} \bm{X})^{-1}\bm{X}^T \bm{\Sigma}^{-1})(\bm{I}_n - \bm{\Sigma}^{-1}\bm{X}(\bm{X}^T \bm{\Sigma}^{-1} \bm{X})^{-1}\bm{X}^T )\bm{\Sigma}^{-1} \bm{Z}}{(\bm{Z}^T \bm{\Sigma}^{-1} (\bm{I}_n - \bm{X}(\bm{X}^T \bm{\Sigma}^{-1} \bm{X})^{-1}\bm{X}^T \bm{\Sigma}^{-1}) \bm{Z})^2}\\
    &= (\bm{Z}^T \bm{\Sigma}^{-1} (\bm{I}_n - \bm{X}(\bm{X}^T \bm{\Sigma}^{-1} \bm{X})^{-1}\bm{X}^T \bm{\Sigma}^{-1}) \bm{Z})^{-1}.
\end{align*}
Weight dispersion is linearly related to the sample variance of the weights, since 
    $D(\rho^2) = \sum_{i = 1}^n (w_i - \bar{w})^2  + \frac{4}{n}.$

\begin{proposition}[Balance and weight dispersion as functions of $\rho^2$]
\label{thm:balance-dispersion-tradeoff}
    $\Delta(\rho^2)$ is non-increasing in $\rho^2$ and $D(\rho^2)$ is non-decreasing in $\rho^2$. 
\end{proposition}

\begin{proof}
Let $\bm{V}$ be a matrix whose rows form a $\bm{\Sigma}^{-1}$-orthogonal basis for $\text{col}(\bm{X})^\perp$. We can rewrite $$\bm{\Sigma}^{-1}(\bm{I}_n - \bm{X} (\bm{X}^T \bm{\Sigma}^{-1} \bm{X})^{-1} \bm{X}^T \bm{\Sigma}^{-1}) = \bm{V}^T (\bm{V} \bm{\Sigma} \bm{V}^T)^{-1} \bm{V}.$$
For ease of notation, let $\bm{M} = \bm{V}^T (\bm{V} \bm{\Sigma} \bm{V}^T)^{-1} \bm{V}$ so that $\Delta(\rho^2) = \frac{\bm{Z}^T \bm{M} \bm{S} \bm{M} \bm{Z}}{(\bm{Z}^T \bm{M} \bm{Z})^2}$ and $D(\rho^2) = (\bm{Z}^T \bm{M} \bm{Z})^{-1}$. 

Then \begin{align*}
    \frac{\partial \Delta(\rho^2)}{\partial \rho^2} &= \frac{-2(\bm{Z}^T \bm{M} \bm{Z})^2 \bm{Z}^T \bm{M} \bm{S} \bm{M} \bm{S}\bm{M} \bm{Z} + 2(\bm{Z}^T \bm{M} \bm{Z}) (\bm{Z}^T \bm{M} \bm{S} \bm{M} \bm{Z})^2}{(\bm{Z}^T \bm{M} \bm{Z})^4}\\
    &= \frac{2(\bm{Z}^T \bm{M} \bm{Z})[(\bm{Z}^T \bm{M} \bm{S} \bm{M} \bm{Z})^2 - (\bm{Z}^T \bm{M} \bm{Z})(\bm{Z}^T \bm{M} \bm{S} \bm{M} \bm{S}\bm{M} \bm{Z})]  }{(\bm{Z}^T \bm{M} \bm{Z})^4}.
\end{align*}By the Cauchy-Schwarz inequality, \begin{align*}
    (\bm{Z}^T \bm{M} \bm{S} \bm{M} \bm{Z})^2 = (\bm{Z}^T \sqrt{\bm{M}}) (\sqrt{\bm{M}} \bm{S} \bm{M} \bm{Z}) \leq (\bm{Z}^T \bm{M} \bm{Z}) (\bm{Z}^T \bm{M} \bm{S} \bm{M} \bm{S} \bm{M} \bm{Z})
\end{align*}implying that $\frac{\partial \Delta(\rho^2)}{\partial \rho^2} \leq 0$. 
Additionally, 
\begin{align*}
    \frac{\partial D(\rho^2)}{\partial \rho^2 } &= -(\bm{Z}^T \bm{M} \bm{Z})^{-2} \bm{Z}^T (-\bm{M} \bm{S} \bm{M}) \bm{Z}\\
    &= \frac{\bm{Z}^T \bm{M} \bm{S} \bm{M} \bm{Z}}{(\bm{Z}^T \bm{M} \bm{Z})^2}
\end{align*} implying that $ \frac{\partial D(\rho^2)}{\partial \rho^2 } \geq 0$.
\end{proof}

\subsection*{Random effects and fixed effects models: a new perspective on the bias-variance tradeoff}
\label{sec:re-fe-tradeoff}
We now interpret the results of Section \ref{sec:balance-dispersion} in the context of random effects, fixed effects, and complete pooling models \citepsupp{gelman2006data}. In the weighting problem for the random effects model, eigenvectors of $\bm{S}$ corresponding to nonzero eigenvalues are the binary state-level indicators: $$\bm{v}_k = (I(C_1 = k), \ldots, I(C_n = k))^T \text{ with } \lambda_k = n_k = \sum_{i = 1}^n I(C_i = k) \text{ for }k = 1, \ldots, K = 51.$$ Hence, the latent covariates that the random effects model approximately balances are precisely the state-level indicators.

By Proposition \ref{thm:balance-dispersion-tradeoff}, the weighted sum of imbalances over state indicators is non-increasing in $\rho^2$ while weight dispersion is non-decreasing in $\rho^2$. This has the following implications. At the fixed effects limit $(\rho^2 = \infty)$, the state-level indicators are perfectly balanced in mean at the cost of higher weight dispersion. At the complete pooling limit ($\rho^2 = 0$), balance on state-level indicators is ignored and weight dispersion is minimized. Random effects models, with $\rho^2 \in (0, \infty)$ occupy the continuum between these two extremes. 
 Figure \ref{fig:re-fe} illustrates this phenomenon in the Superfund data application.

\begin{figure}[!ht]
    \centering
    \includegraphics[width=\linewidth]{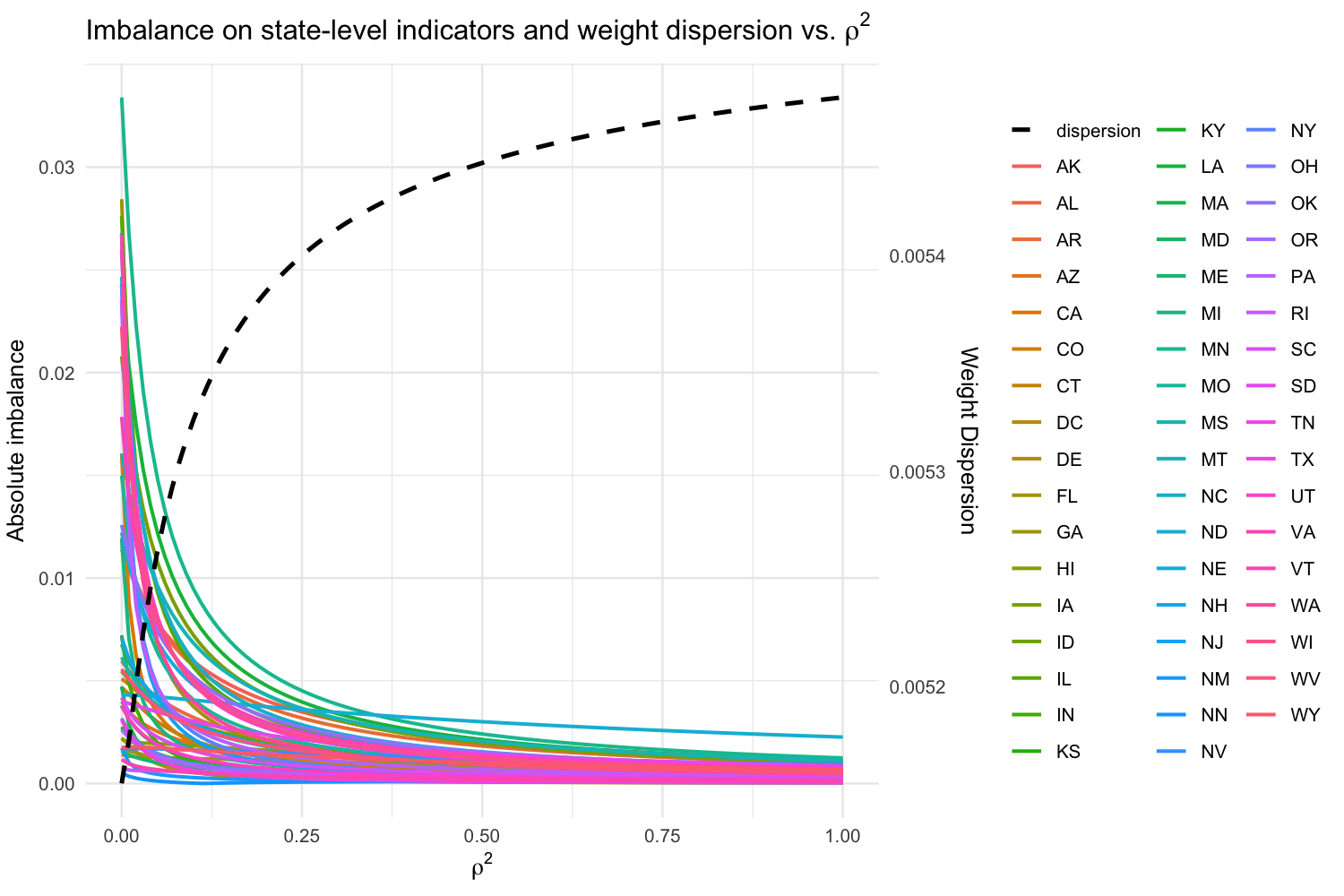}
    \caption{The balance-dispersion tradeoff presented by the variance of the random intercepts, $\rho^2$, in random effects models. We fix $\sigma^2 = 1$.}
    \label{fig:re-fe}
\end{figure}

It is well-documented that the variance of the random intercepts $\rho^2$ presents a bias-variance tradeoff \citepsupp{clark2015should}. We re-establish this result by demonstrating that our balance-dispersion result gives rise to the same tradeoff. Suppose: 
\begin{enumerate}
    \item the unmeasured confounder $\bm{U}$ is constant within counties, i.e. $U_i = U_j$ for $C_i = C_j$.
    \item the outcome model is linear, $ Y_i = \bm{\beta}^T \bm{X}_i + \tau Z_i + \gamma U_i + \epsilon_i$ for $i = 1, \ldots, n$ where $\epsilon_i$ is independent and identically distributed with $\E(\epsilon_i) = 0$, $\text{Var}(\epsilon_i) = \sigma^2$, and $\epsilon_i \indep (\bm{X}, \bm{Z}, \bm{U})$ $\forall i$.
\end{enumerate}By 1., $\bm{U}$ can be re-expressed as $U_i = \sum_{k = 1}^K \delta_k I(C_i = k)$ or equivalently, $$\bm{U} = \sum_{k = 1}^K \delta_k \bm{v}_k$$for some constants $\delta_1, \ldots, \delta_K \in \mathbb R$.

The conditional bias of the GLS estimate of the treatment coefficient derived from a random effects model $\hat{\tau}_{GLS}$ in estimating $\tau_{ATT} = \tau$ is
\begin{align*}
    |\E(\hat{\tau}_{GLS}|\bm{X}, \bm{Z}, \bm{U}) - \tau| &= \bigg\vert \gamma (\sum_{i:Z_i = 1} w_i U_i - \sum_{i:Z_i = 0} w_i U_i)\bigg\vert  \\
    &= \bigg\vert \gamma \sum_{k = 1}^K \delta_k (\sum_{i:Z_i =1} w_i v_{ki} - \sum_{i:Z_i =0} w_i v_{ki})\bigg\vert\\
    &= \bigg\vert \gamma \sum_{k = 1}^K \frac{\delta_k}{\sqrt{\lambda_k}} \sqrt{\lambda_k}(\sum_{i:Z_i =1} w_i v_{ki} - \sum_{i:Z_i =0} w_i v_{ki})\bigg\vert\\
    &\leq |\gamma|\sqrt{\sum_{k = 1}^K \frac{\delta_k^2}{\lambda_k}\Delta(\rho^2)}
\end{align*}by the Cauchy-Schwarz inequality. Meanwhile, the conditional variance of $\hat{\tau}_{GLS}$ is \begin{align*}
    \text{Var}(\hat{\tau}_{GLS}|\bm{X}, \bm{Z}, \bm{U}) &= \sigma^2 \sum_{i = 1}^n w_i^2,
\end{align*}which is proportional to the weight dispersion. By Proposition \ref{thm:balance-dispersion-tradeoff}, we conclude that the conditional bias of $\hat{\tau}_{GLS}$ tends to decrease with $\rho^2$ while the conditional variance increases, confirming the well-known bias-variance tradeoff. 

\section{Equivalence of GLS with ridge regression}
\label{sec:gls-ridge}
\begin{proposition}[Equivalence of Spatial Regression with Augmented Ridge Regression]
Suppose that $\bm{\Sigma} = \sigma^2 \bm{I}_n + \rho^2 \bm{S}$, where $\bm{S}$ is an $n\times n$ positive semidefinite matrix. Let $\bm{v}_1, \ldots, \bm{v}_n$ be the eigenvectors of $\bm{S}$ with corresponding eigenvalues $\lambda_1 \geq \ldots \geq \lambda_n \geq 0$. 

Consider a ridge regression of $\bm{Y}$ on measured covariates $\bm{X}$, binary treatment $\bm{Z}$, and the eigenvectors $\bm{v}_1, \ldots, \bm{v}_n$, with corresponding ridge penalties of $\bm{0}_p$, $0$, and $\sigma^2/(\rho^2 \lambda_1), \ldots, \sigma^2 /(\rho^2 \lambda_n)$ respectively. 

More specifically, denote the ridge regression loss as
$$\mathcal{L}(\bm{\beta}, \tau, \bm{\gamma}) = \sum_{i = 1}^n \bigg(Y_i - \tau Z_i -\bm{\beta}^T \bm{X}_i - \sum_{j = 1}^n \gamma_j v_{ij}\bigg)^2 + \frac{\sigma^2}{\rho^2}\sum_{j = 1}^n \frac{\gamma_j^2}{\lambda_j}. $$

Then $\hat{\tau}_{\text{ridge}}$, the coefficient estimate of $\tau$ obtained from minimizing $\mathcal{L}(\bm{\beta}, \tau, \bm{\gamma})$, is equal to the GLS estimator of $\tau$:
$$\hat{\tau}_{\text{GLS}} = \hat{\tau}_{\text{ridge}}.$$
\end{proposition}

\begin{proof} Let $\bm{V} = \begin{pmatrix} \bm{v}_1 & \ldots & \bm{v}_n\end{pmatrix}$ and let $\bm{\Lambda}$ denote the $n \times n$ diagonal matrix with $k$th entry $\lambda_k$.
    The ridge regression estimate is equal to 
    \begin{align*}
        \begin{pmatrix} \hat{\bm{\beta}}_{ridge} \\ \hat{\tau}_{ridge} \\
        \hat{\bm{\gamma}}_{ridge} \end{pmatrix} &= \bigg[\begin{pmatrix} \bm{X} &\bm{Z} & \bm{V} \end{pmatrix}^T \begin{pmatrix} \bm{X} & \bm{Z} & \bm{V} \end{pmatrix} + \begin{pmatrix} \bm{0}_{(p +1) \times (p+1)} & \bm{0}_{(p +1) \times n}   \\ \bm{0}_{n \times (p+1)}  & \frac{\sigma^2}{\rho^2}\bm{\Lambda}^{-1} \end{pmatrix} \bigg]^{-1}\begin{pmatrix} \bm{X} & \bm{Z} & \bm{V} \end{pmatrix}^T \bm{Y}.
    \end{align*}

    By the Woodbury formula, \begin{align*}
        \bm{\Sigma}^{-1} = \frac{\rho^2}{\sigma^2}(\bm{I}_n-\bm{V}(\bm{I}_n + \frac{\sigma^2}{\rho^2} \bm{\Lambda}^{-1})^{-1} \bm{V}^T).
    \end{align*}The GLS estimator of $(\beta, \tau)$ is therefore
    \begin{align*}
    \begin{pmatrix} \hat{\beta}_{GLS} \\ \hat{\tau}_{GLS} \end{pmatrix} &= \bigg[\begin{pmatrix} \bm{X} & \bm{Z} \end{pmatrix}^T\begin{pmatrix} \bm{X} & \bm{Z} \end{pmatrix} - \begin{pmatrix} \bm{X} & \bm{Z} \end{pmatrix}^T\bm{V}(\bm{I}_n + \frac{\sigma^2}{\rho^2} \bm{\Lambda}^{-1})^{-1} \bm{V}^T \begin{pmatrix} \bm{X} & \bm{Z} \end{pmatrix}\bigg]^{-1} \\
    &\hspace{2.5in}\bigg[\begin{pmatrix} \bm{X} & \bm{Z} \end{pmatrix}^T \bm{Y} - \begin{pmatrix} \bm{X} & \bm{Z} \end{pmatrix}^T \bm{V}(\bm{I}_n + \frac{\sigma^2}{\rho^2} \bm{\Lambda}^{-1})^{-1} \bm{V}^T \bm{Y}\bigg]. \end{align*}

    The ridge regression estimate is \begin{align*}
        \hat{\beta}_{ridge} &= \bigg[\begin{pmatrix} \bm{X} & \bm{Z} \end{pmatrix}^T \begin{pmatrix} \bm{X} & \bm{Z} \end{pmatrix} - \begin{pmatrix} \bm{X} & \bm{Z} \end{pmatrix}^T \bm{V} (\bm{I}_n + \frac{\sigma^2}{\rho^2} \bm{\Lambda}^{-1})^{-1} \bm{V}^T \begin{pmatrix} \bm{X} & \bm{Z} \end{pmatrix}\bigg]^{-1} \begin{pmatrix} \bm{X} & \bm{Z} \end{pmatrix}^T \bm{Y}\\
        &\hspace{0.25cm}-\bigg[\begin{pmatrix} \bm{X} & \bm{Z} \end{pmatrix}^T \begin{pmatrix} \bm{X} & \bm{Z} \end{pmatrix} - \begin{pmatrix} \bm{X} & \bm{Z} \end{pmatrix}^T \bm{V}(\bm{I}_n + \frac{\sigma^2}{\rho^2} \bm{\Lambda}^{-1})^{-1} \bm{V}^T \begin{pmatrix} \bm{X} & \bm{Z} \end{pmatrix}\bigg]^{-1} \\
        &\hspace{3.5in}\begin{pmatrix} \bm{X} & \bm{Z} \end{pmatrix}^T \bm{V} (\bm{I}_n + \frac{\sigma^2}{\rho^2} \bm{\Lambda}^{-1})^{-1} \bm{V}^T \bm{Y}\\
        &= \hat{\beta}_{GLS}.
    \end{align*}

    In summary, augmenting the regression of $\bm{Y}$ on $\bm{X},\bm{Z}$ with the eigenvectors of $\bm{S}$ and applying ridge penalties to their coefficients---penalizing each eigenvector coefficient inversely to its eigenvalue---yields the GLS estimate of $(\beta, \tau)$. The eigenvectors that are least penalized correspond to the highest eigenvalues. 
    
    As a corollary, the coefficients of the eigenvectors equivalently produced by the ridge regression, $\hat{\bm{\gamma}}_\text{ridge}$, are \begin{align*}
        \hat{\bm{\gamma}}_\text{ridge} &= \bigg(\frac{\sigma^2}{\rho^2} \bm{\Lambda}^{-1} + \bm{V}^T \bm{P}^\perp_{\bm{X}, \bm{Z}} \bm{V}\bigg)^{-1} \bm{V}^T  \bm{P}^\perp_{\bm{X}, \bm{Z}} \bm{Y}
    \end{align*}where the $\bm{P}^\perp_{\bm{X}, \bm{Z}} \bm{V} = \bm{I}_n - \begin{pmatrix} \bm{X} & \bm{Z} \end{pmatrix} \bigg(\begin{pmatrix} \bm{X} & \bm{Z} \end{pmatrix}^T \begin{pmatrix} \bm{X} & \bm{Z} \end{pmatrix}\bigg)^{-1} \begin{pmatrix} \bm{X} & \bm{Z} \end{pmatrix}^T.$ As $\rho^2 \rightarrow \infty$, $\hat{\bm{\gamma}}_\text{ridge}$ converges to $(\bm{V}^T \bm{P}^\perp_{\bm{X}, \bm{Z}} \bm{V})^{-1} \bm{V}^T  \bm{P}^\perp_{\bm{X}, \bm{Z}} \bm{Y}$, the coefficients of $\bm{V}$ that would be obtained from the overspecified ordinary least squares regression of $\bm{Y}$ on $\bm{X}, \bm{Z}, \bm{V}$. As $\rho^2 \rightarrow 0$, $\hat{\bm{\gamma}}_\text{ridge} \rightarrow \bm{0}$.
\end{proof}

\section{Visual diagnostics}
Figure \ref{fig:ASMD_moran} plots the maximum absolute conditional bias arising from $\bm{U}$ against its Moran's I statistic and its coefficient in the outcome model. Figure \ref{fig:consolidated} displays three instances of $\bm{U}$ together with their respective Moran’s I statistics and mean imbalances. 
\begin{SCfigure}[][!ht]
\includegraphics[width=0.4\linewidth]{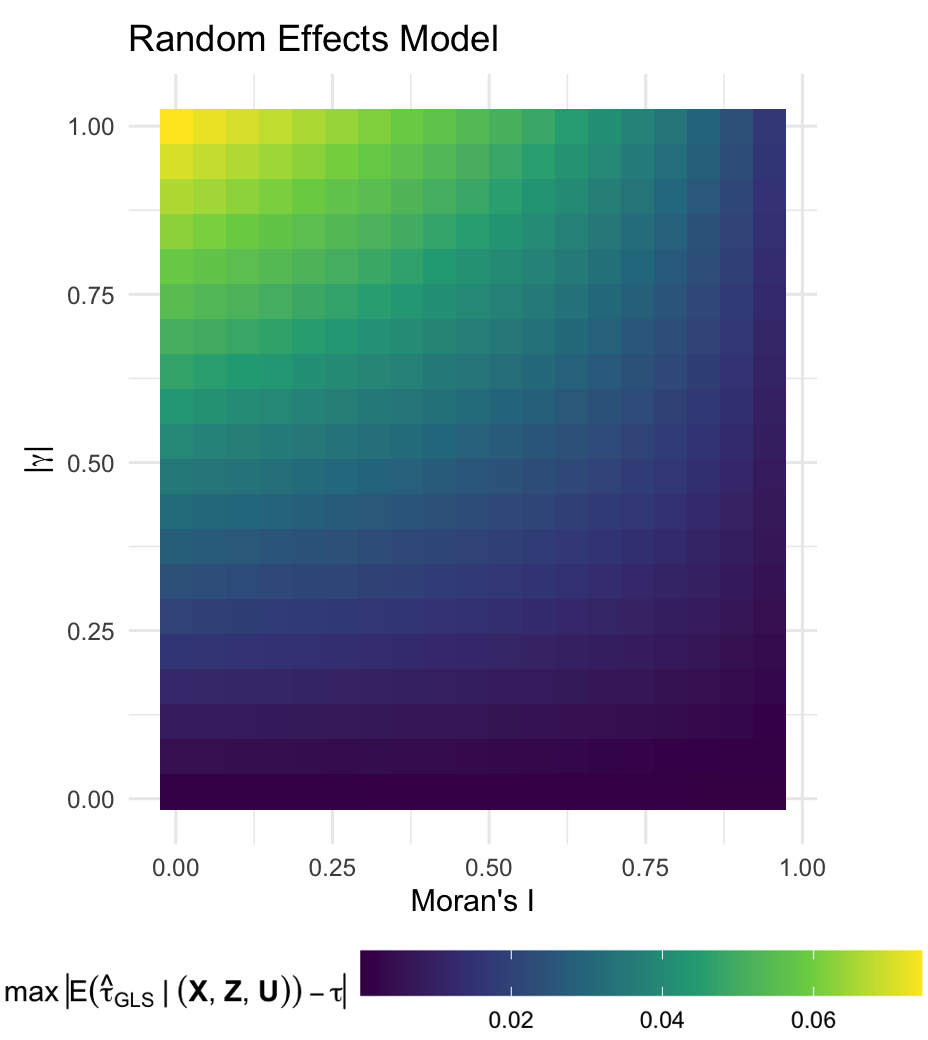}
    \captionsetup{singlelinecheck=off, justification=raggedright}
    \caption{Maximum absolute conditional bias as a function of Moran's I 
      and absolute regression coefficient $|\gamma|$. For a grid of equally-spaced values $c_1 \in (0,1)$ and $c_2 \in (0,1)$, the nonlinear program computes $\max_{\bm{U}} |\E(\hat{\tau}_{GLS}|(\bm{X}, \bm{Z}, \bm{U}) - \tau)|$ subject to $|\gamma| = c_1$ and $\mathcal{I}(\bm{U};\bm{S}) = c_2$. Results for the random effects, conditional autoregressive, and Gaussian process specifications are nearly identical.}
    \label{fig:ASMD_moran}
\end{SCfigure}

\begin{figure}[!ht]
    \centering
    \begin{subfigure}[t]{0.32\textwidth}
        \centering
        \includegraphics[height=1.7in]{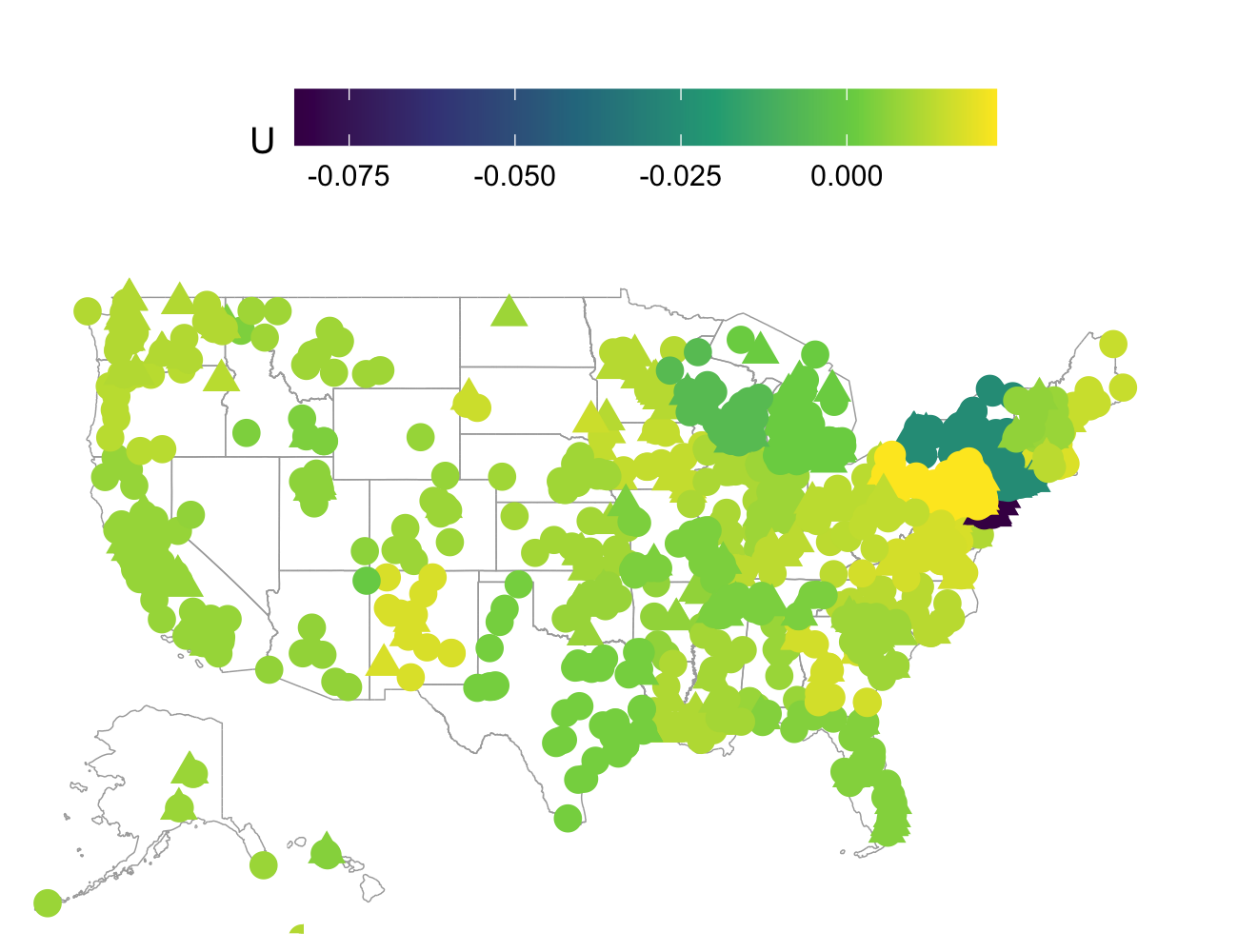}
        \caption{$\mathcal{I}(\bm{U};\bm{S}) = 0.90$, $\sum_{i:Z_i = 1} w_i U_i - \sum_{i:Z_i = 0} w_i U_i = -2.64 \times 10^{-6}$.}
    \end{subfigure}
    \hfill
    \begin{subfigure}[t]{0.32\textwidth}
        \centering
        \includegraphics[height=1.7in]{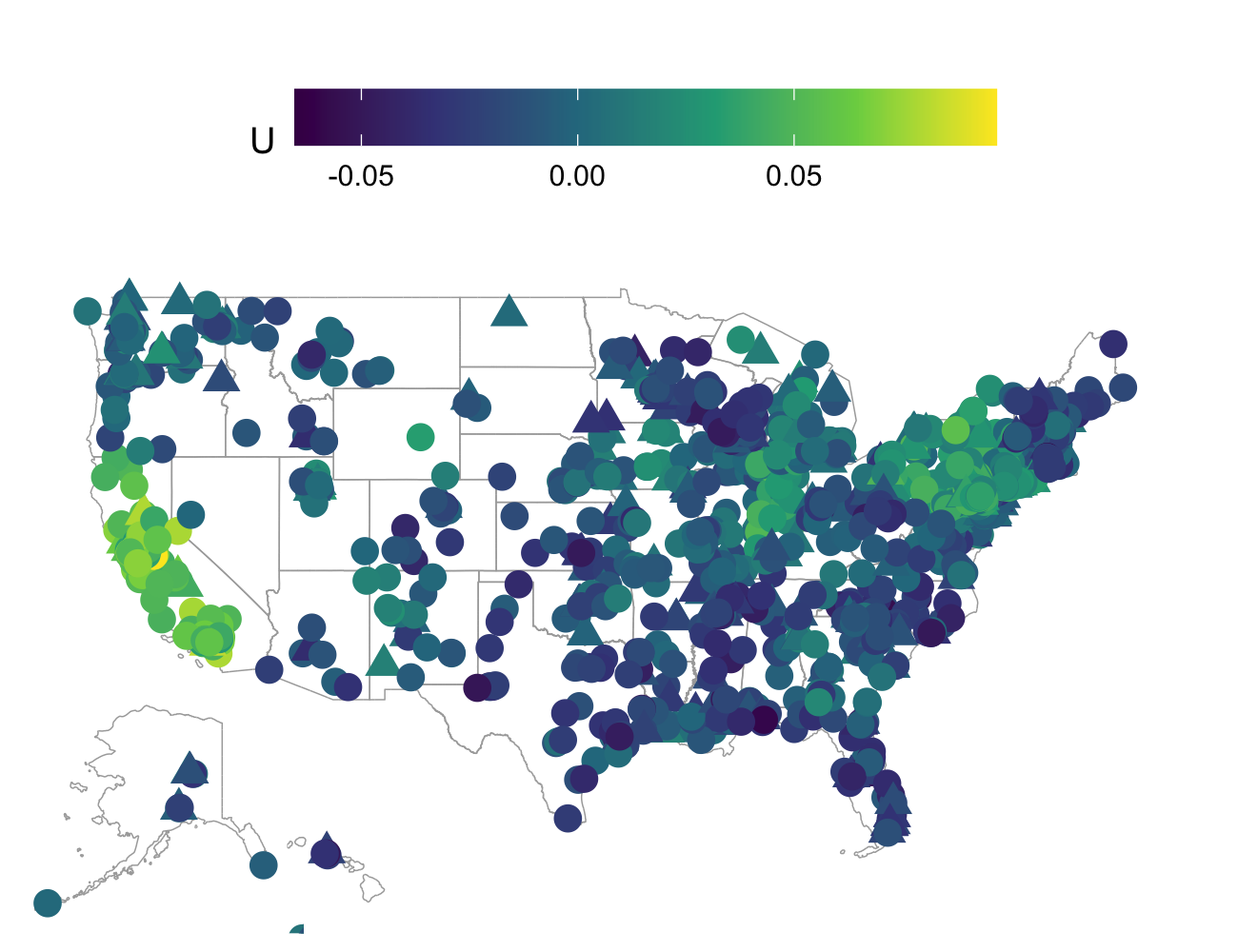}
        \caption{$\mathcal{I}(\bm{U};\bm{S}) = 0.49$, $\sum_{i:Z_i = 1} w_i U_i - \sum_{i:Z_i = 0} w_i U_i = 1.61  \times 10^{-3}$.}
    \end{subfigure}
    \hfill
    \begin{subfigure}[t]{0.32\textwidth}
        \centering
        \includegraphics[height=1.7in]{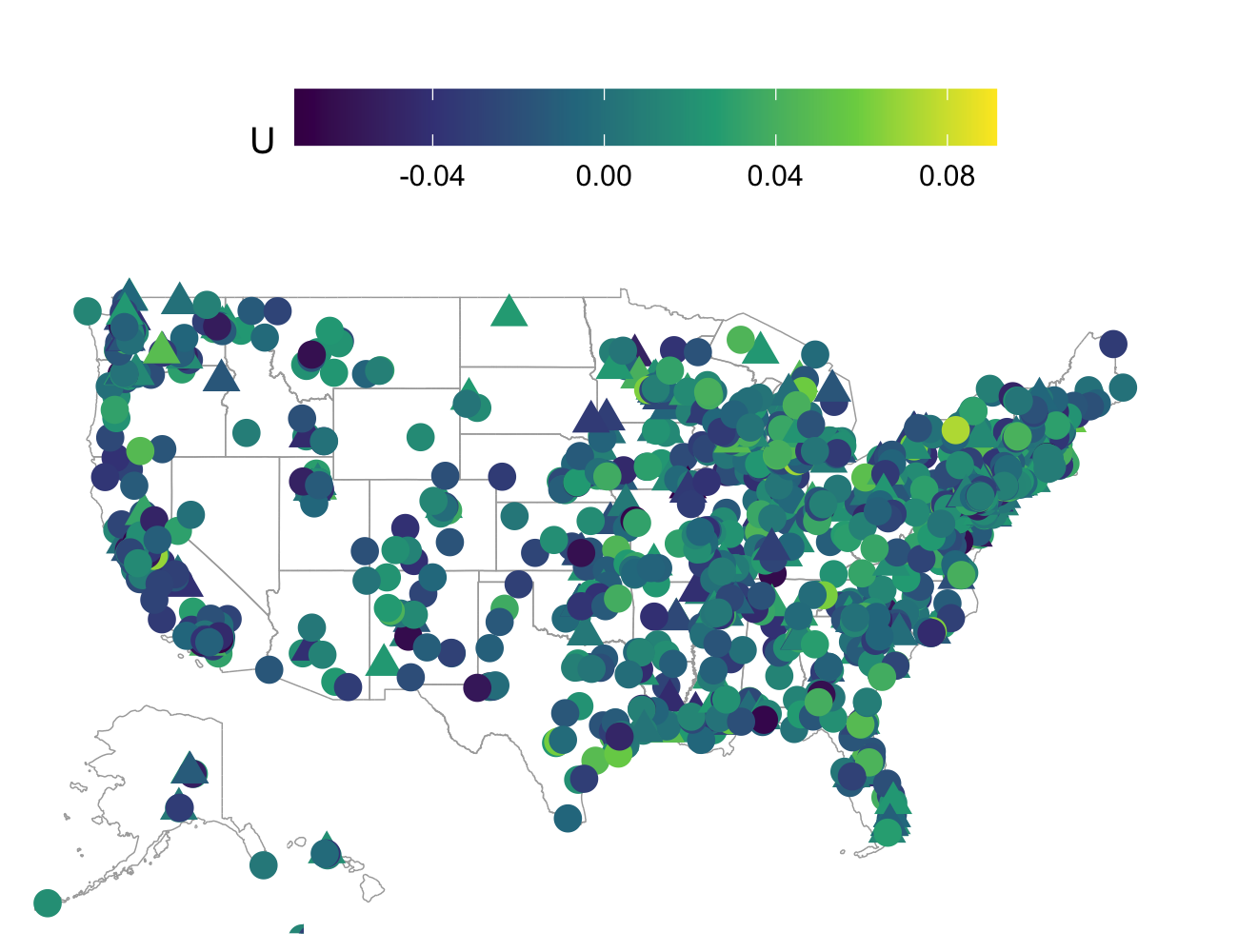}
        \caption{$\mathcal{I}(\bm{U};\bm{S}) = 0.03$, $\sum_{i:Z_i = 1} w_i U_i - \sum_{i:Z_i = 0} w_i U_i = 3.27 \times 10^{-3}.$}
    \end{subfigure}
        
    \begin{subfigure}[t]{0.32\textwidth}
        \centering
        \includegraphics[height=1.7in]{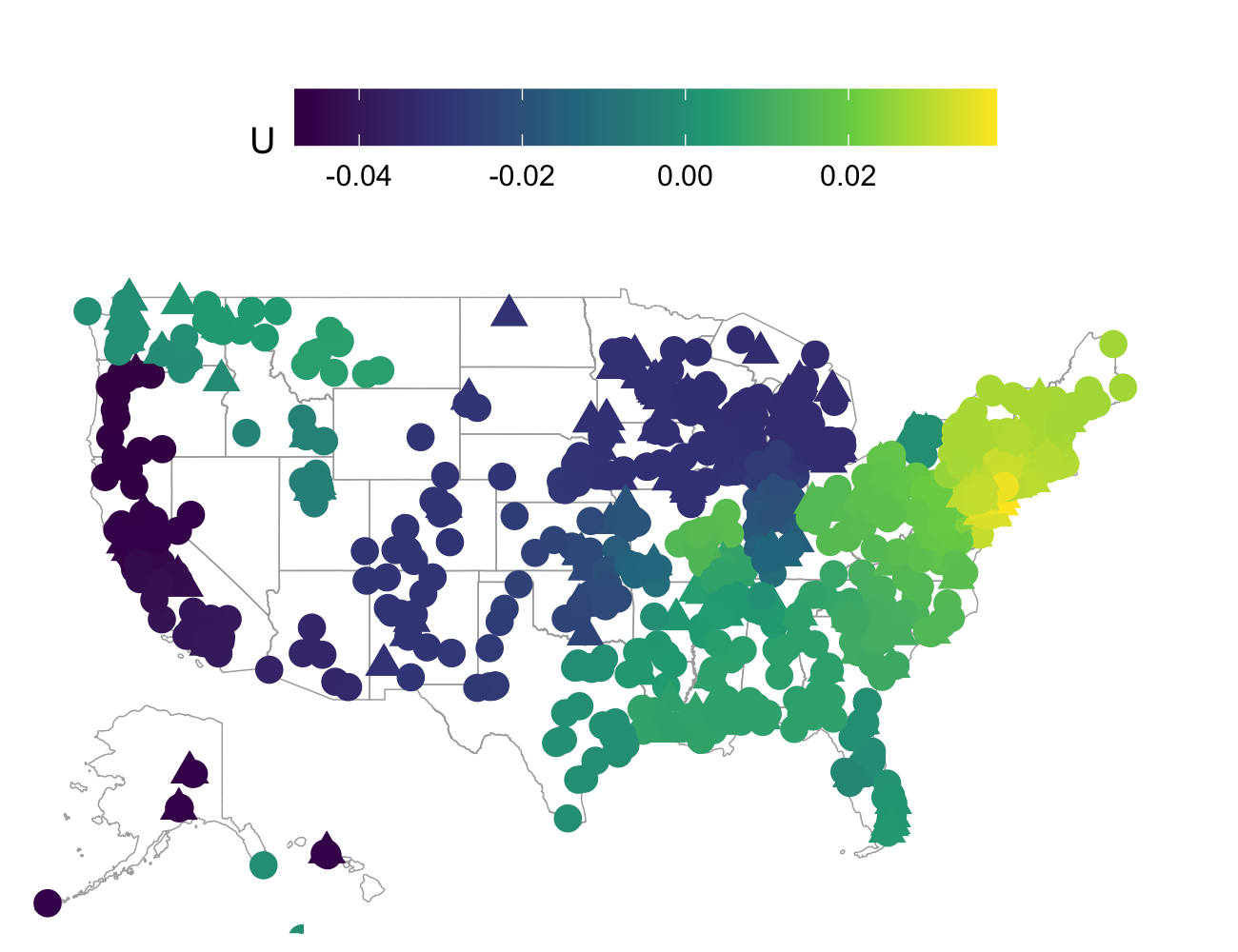}
        \caption{$\mathcal{I}(\bm{U};\bm{S}) = 0.95$, $\sum_{i:Z_i = 1} w_i U_i - \sum_{i:Z_i = 0} w_i U_i = 2.15 \times 10^{-6}$.}
    \end{subfigure}
    \hfill
    \begin{subfigure}[t]{0.32\textwidth}
        \centering
        \includegraphics[height=1.7in]{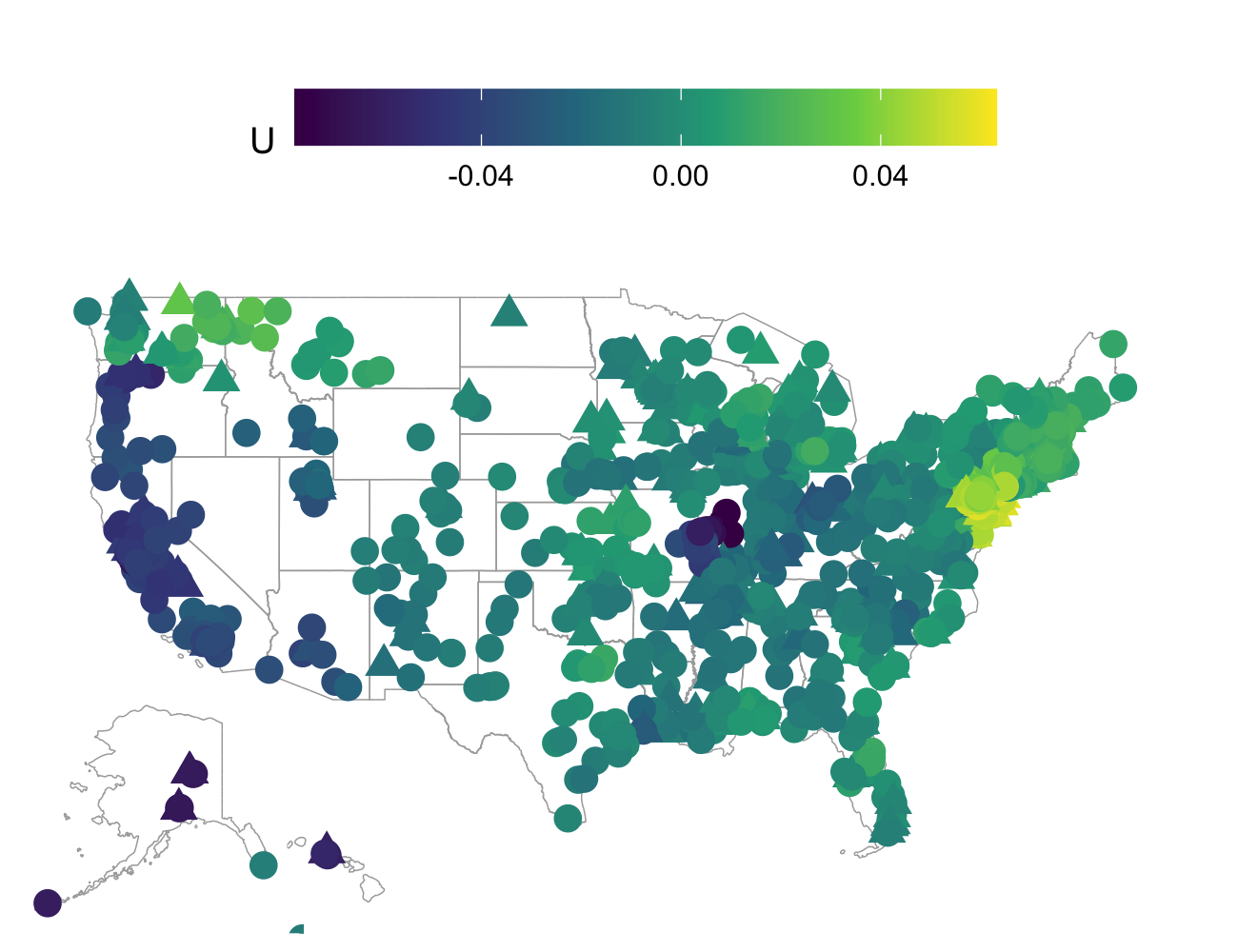}
        \caption{$\mathcal{I}(\bm{U};\bm{S}) = 0.53$, $\sum_{i:Z_i = 1} w_i U_i - \sum_{i:Z_i = 0} w_i U_i = 2.12 \times 10^{-4}$.}
    \end{subfigure}
    \hfill
    \begin{subfigure}[t]{0.32\textwidth}
        \centering
        \includegraphics[height=1.7in]{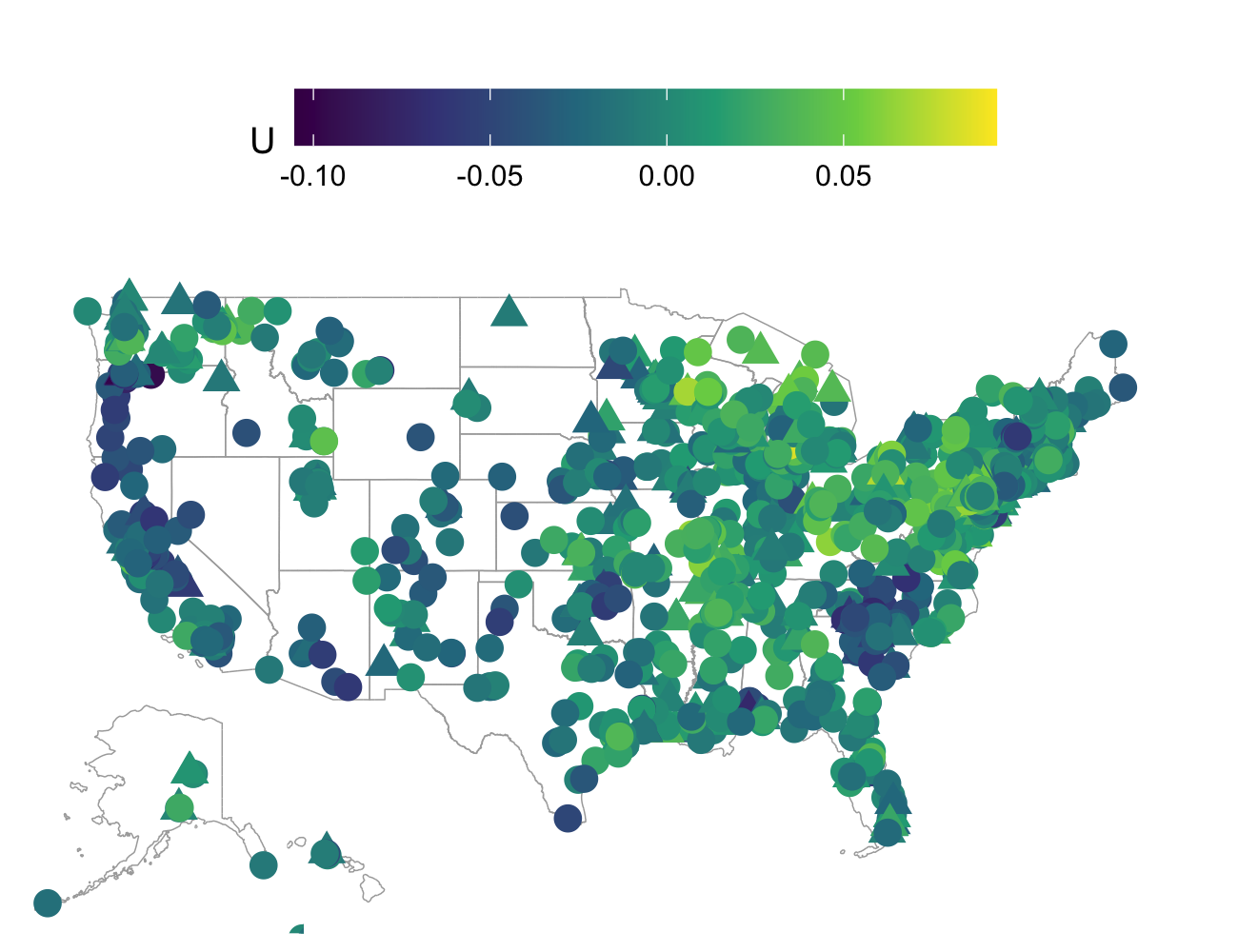}
        \caption{$\mathcal{I}(\bm{U};\bm{S}) = 0.09$, $\sum_{i:Z_i = 1} w_i U_i - \sum_{i:Z_i = 0} w_i U_i = 9.57\times 10^{-4}$.}
    \end{subfigure}

    \begin{subfigure}[t]{0.32\textwidth}
        \centering
        \includegraphics[height=1.7in]{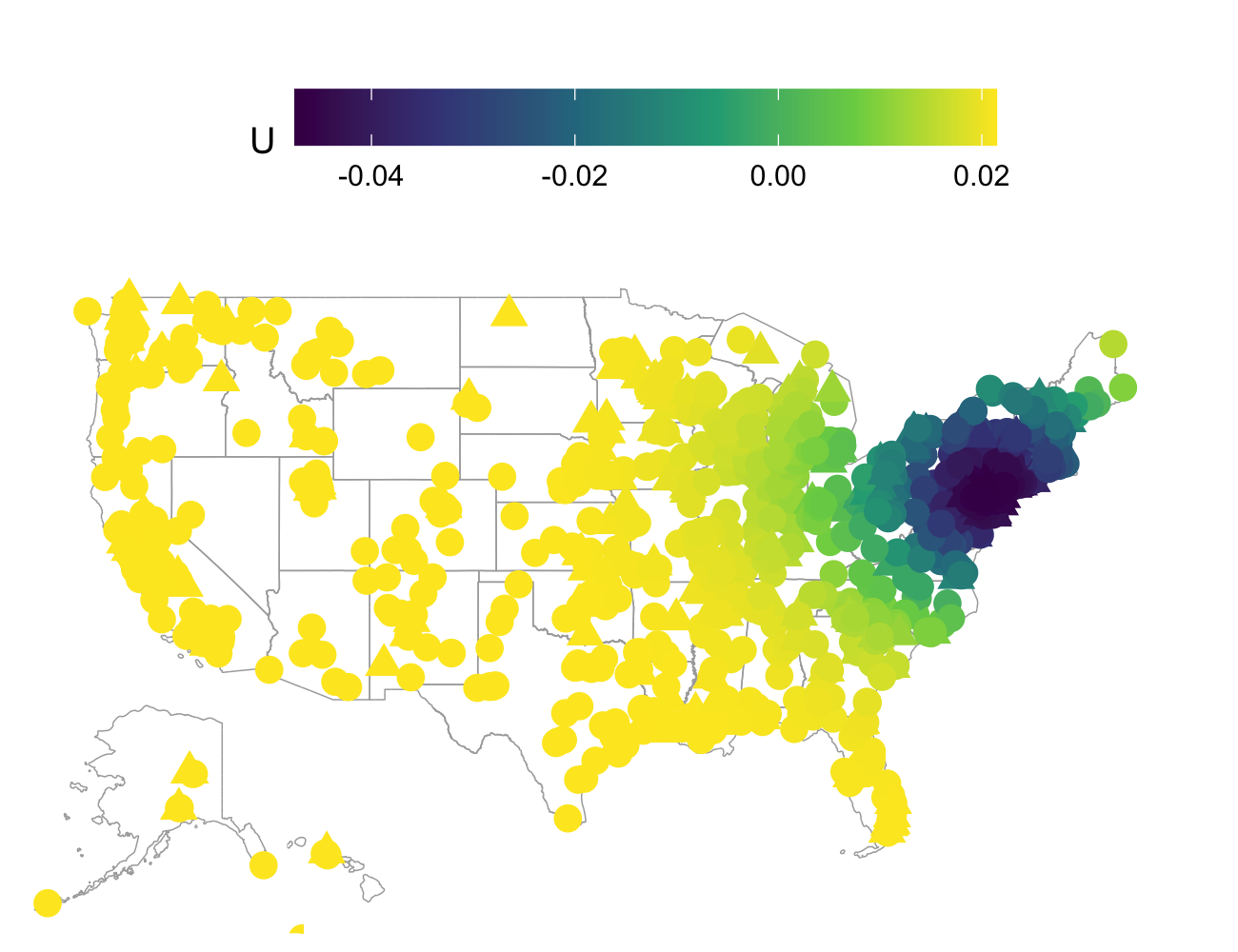}
        \caption{$\mathcal{I}(\bm{U};\bm{S}) = 0.72$, $\sum_{i:Z_i = 1} w_i U_i - \sum_{i:Z_i = 0} w_i U_i = 1.42 \times 10^{-6}.$}
    \end{subfigure}
    \hfill
    \begin{subfigure}[t]{0.32\textwidth}
        \centering
        \includegraphics[height=1.7in]{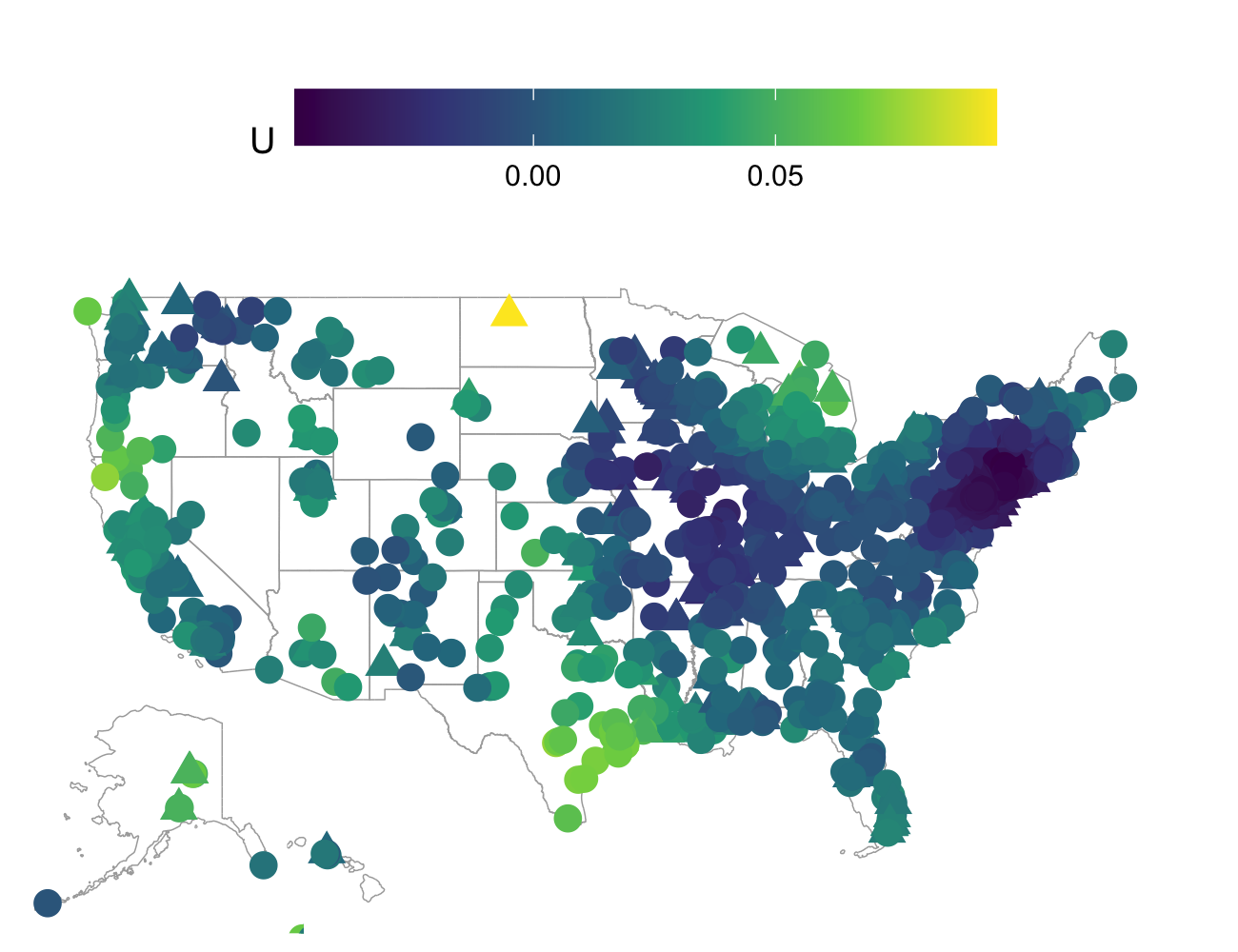}
        \caption{$\mathcal{I}(\bm{U};\bm{S}) = 0.46$, $\sum_{i:Z_i = 1} w_i U_i - \sum_{i:Z_i = 0} w_i U_i = 1.28 \times 10^{-4}.$}
    \end{subfigure}
    \hfill
    \begin{subfigure}[t]{0.32\textwidth}
        \centering
        \includegraphics[height=1.7in]{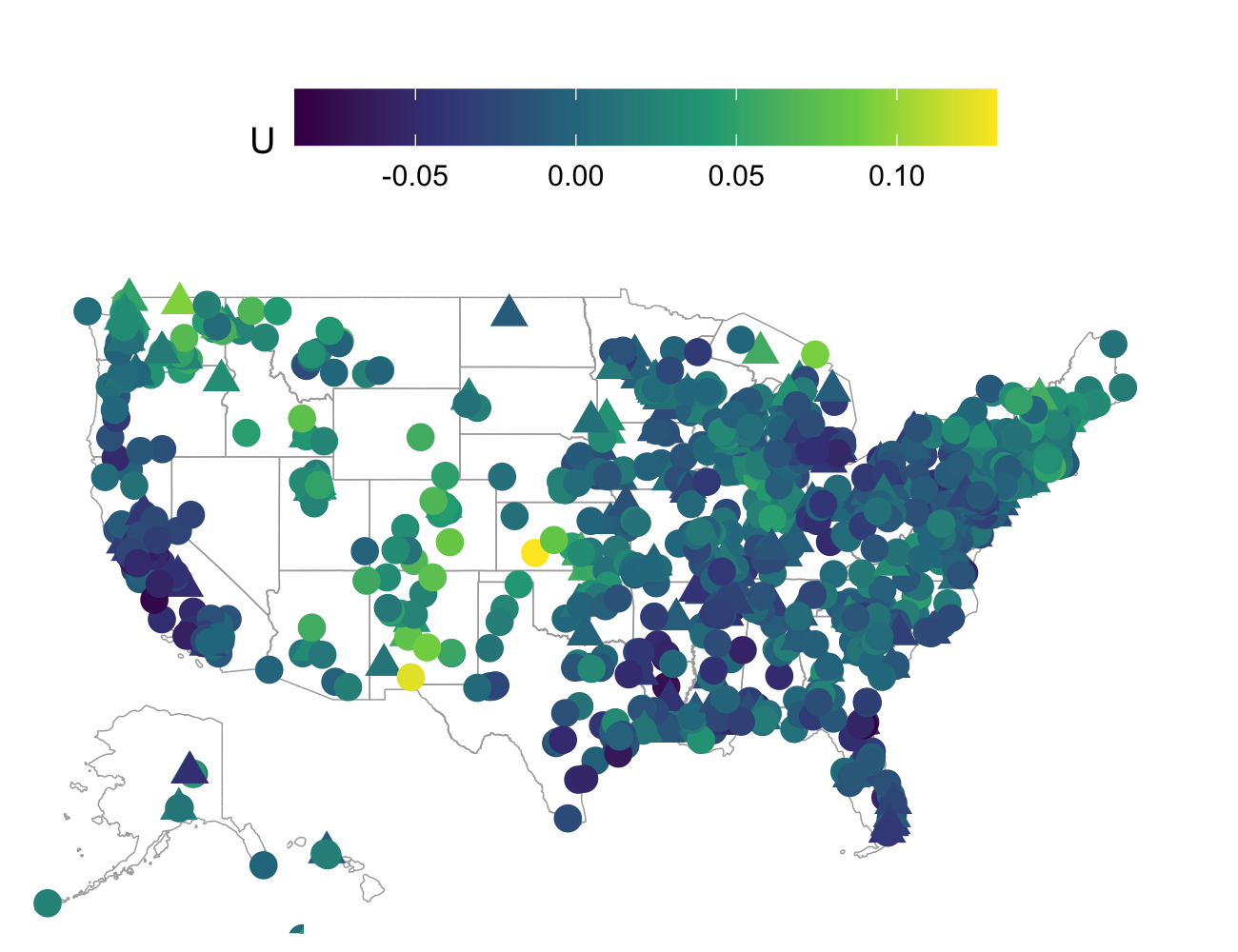}
        \caption{$\mathcal{I}(\bm{U};\bm{S}) = 0.074$, $\sum_{i:Z_i = 1} w_i U_i - \sum_{i:Z_i = 0} w_i U_i = -8.46 \times 10^{-4}$.}
    \end{subfigure}
    
    \caption{Three examples of unmeasured confounders with Moran's I and their mean imbalances. 
    \textbf{Top row} shows three examples from the random effects model, 
    \textbf{middle row} shows three examples from the conditional autoregressive model, and 
    \textbf{bottom row} shows three examples from the Gaussian process model. By our construction  $\bar{U} = 0$ and $\sum_{i = 1}^n U_i^2 = 1$.}
    \label{fig:consolidated}
\end{figure}

\clearpage

\section{Spatially localized weighting}
\label{sec:spatially-localized-weighting}
Using data from the Superfund and birth outcomes study, Figure \ref{fig:spatially-localized-weights} plots the absolute implied weight for each site ($n = 1429$) against a measure of proximity to sites with the opposite treatment status. As illustrated in Figures \ref{fig:impliedweights} and \ref{fig:spatially-localized-weights}, the implied weights are largest for remediated sites most proximal to non-remediated sites, and similarly, for non-remediated sites most proximal to remediated ones. We refer to this property as \textit{spatial localization.}

We now describe how the weighting problem in Proposition \ref{thm:minimal-weighting} gives rise to a spatially localized contrast. Spatial localization follows from the following reformulation of the weighting problem:
\begin{align*}
        &\min_{\bm{l} \in \mathbb R} \bigg(\sum_{i = 1}^n l_i^2\bigg)\bigg[\sigma^2 + \lambda_1 \rho^2 \mathcal{I}(\bm{l};\bm{S})\bigg]\\
        &\text{subject to} \sum_{i:Z_i = 1} w_i = \sum_{i:Z_i = 0} w_i = 1, \sum_{i:Z_i = 1} w_i \bm{X}_i = \sum_{i:Z_i = 0} w_i \bm{X}_i
    \end{align*}
    where $l_i = w_i$ if $Z_i = 1$ and $l_i = -w_i$ if $Z_i = 0$. We denote $\bm{l} = (l_1, \ldots, l_n)^T$. For a fixed value of the weight dispersion $\sum_{i = 1}^n l_i^2 = \sum_{i = 1}^n w_i^2$, the minimization problem reduces to minimizing the Moran's I of the transformed weights $\bm{l}$. In other words, $\bm{l}$ is constructed to exhibit minimal spatial autocorrelation, and sharp, localized variation in $\bm{l}$ is observed across space. Consequently, locations with large negative values of $\bm{l}$ (typically control units with large positive weight) tend to be situated near locations with large positive values of $\bm{l}$ (typically treated units with large positive weight). This yields a spatially localized weighted contrast that prioritizes comparisons between the most proximal treated and control units.

\begin{figure}[!ht]
    \centering
    \includegraphics[width=0.57\linewidth]{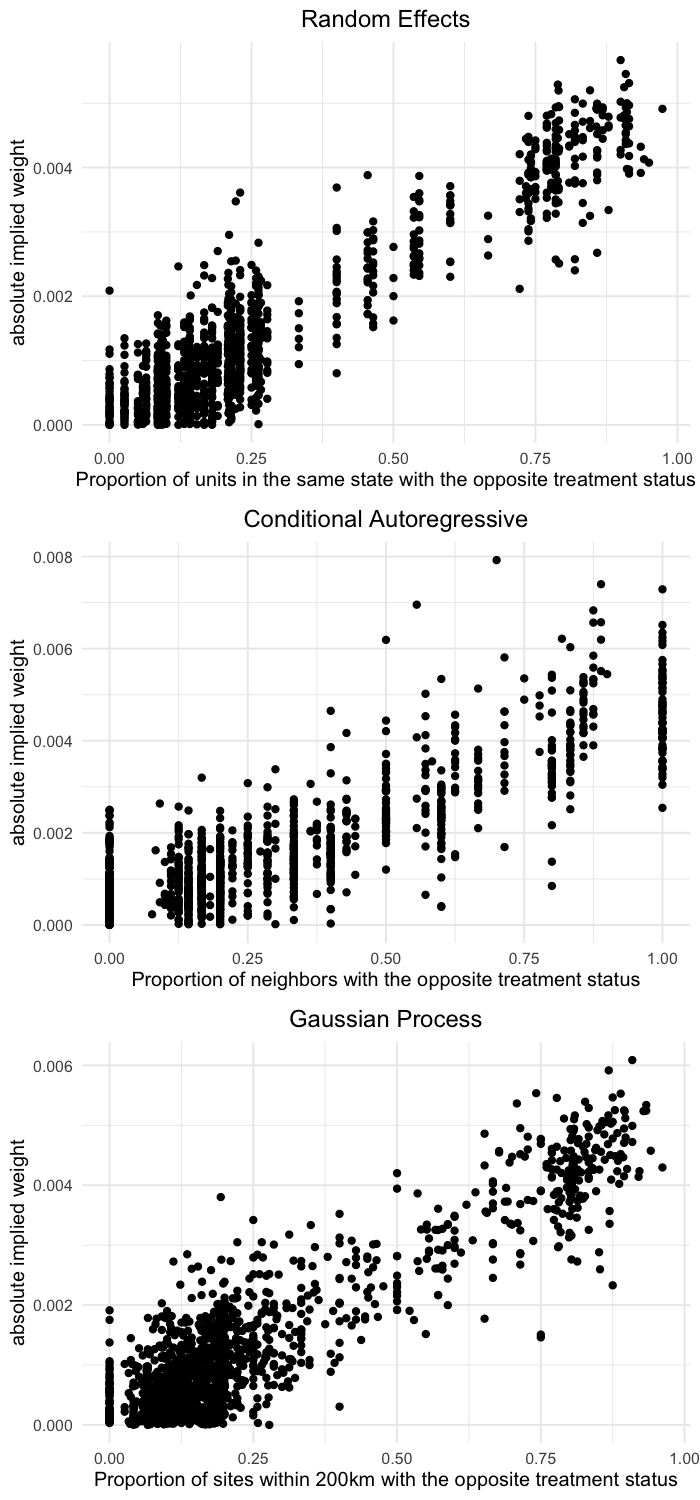}
    \caption{Absolute implied weight plotted against a measure of spatial proximity to sites with the opposite treatment status ($n = 1429$).}
    \label{fig:spatially-localized-weights}
\end{figure}

\clearpage 
\section{Bias of spatial regression under violations of linearity}
\label{sec:violations_linearity}
If linearity does not hold, i.e. the outcome model is not $Y_i = \bm{\beta}^T \bm{X}_i + \tau Z_i + \gamma U_i + \epsilon_i$ for $i = 1, \ldots, n$, then $\hat{\tau}_{GLS}$ can exhibit substantial bias in estimating the ATT. We provide two simple examples. 
\begin{enumerate}
    \item Suppose the outcome mean model $m(Z_i, X_i, U_i) = \E(Y_i|Z_i, X_i, U_i)$ includes an interaction between a measured confounder and the unmeasured confounder. Let $X_{ij}$ denote the $j$th covariate of $\bm{X}_i$, and suppose 
    $$m(Z_i, X_i, U_i) = \bm{\beta}^T \bm{X}_i + \tau Z_i + \gamma U_i + \delta X_{ij} U_i.$$
    Then $$\E(\hat{\tau}_{GLS}|\bm{X}, \bm{Z}, \bm{U}) - \tau_{ATT} = \gamma(\sum_{i:Z_i = 1} w_i U_i - \sum_{i:Z_i = 0} w_i U_i) + \delta (\sum_{i:Z_i = 1} w_i X_{ij} U_i - \sum_{i:Z_i = 0} w_i X_{ij} U_i). $$
    If $(X_{1j}U_1, \ldots, X_{nj}U_n)^T$ lies in $\text{col}(\bm{X})$ or has a high Moran's I value, then, by an argument parallel to Proposition \ref{thm:imbalance_spatialauto}, the second term is small and the resulting bias is bounded below a small threshold. Otherwise, the term $\delta (\sum_{i:Z_i = 1} w_i X_{ij} U_i - \sum_{i:Z_i = 0} w_i X_{ij} U_i)$ can be arbitrarily large. 

    \item Suppose the outcome mean model includes an interaction between treatment and a measured confounder. Let $X_{ij}$ denote the $j$th covariate of $\bm{X}_i$, and suppose
    $$m(Z_i, X_i, U_i) = \bm{\beta}^T \bm{X}_i + \tau Z_i + \gamma U_i + \delta Z_i X_{ij}.$$ Then $$\E(\hat{\tau}_{GLS}|\bm{X}, \bm{Z}, \bm{U}) - \tau_{ATT} = \gamma(\sum_{i:Z_i = 1} w_i U_i - \sum_{i:Z_i = 0} w_i U_i) + \delta (\sum_{i:Z_i = 1} w_i X_{ij} - \E(X_{ij}|Z_i = 1)). $$
    The second term may be viewed as a form of generalizability bias. That is, the weights $w_i$ do not balance the covariates towards any commonly seen target covariate profile---for example, $\sum_{i:Z_i = 1} w_i X_{ij} = \sum_{i:Z_i = 0} w_i X_{ij} \neq \frac{1}{n}\sum_{i = 1}^n X_{ij}$ and $\sum_{i:Z_i = 1} w_i X_{ij} =  \sum_{i:Z_i = 0} w_i X_{ij}\neq \frac{1}{n_t}\sum_{i:Z_i = 1} X_{ij}$---so the magnitude of this term can be large. For a discussion on the target population implied by regression, see \citesupp{chattopadhyay2023implied}.  
    
\end{enumerate}

\section{Finite-sample properties of the spatial weighting estimator}
\label{sec:finite-sample-prop}
The following proposition relates the conditional bias of $\hat{\tau}_{SW}$ to the spatial autocorrelation of the unmeasured confounder. 
\begin{proposition}[Finite-sample bias and variance of the spatial weighting estimator]
\label{thm:robust-weighting-bias}
    Suppose that Assumptions \ref{consistency}--\ref{ignorability} hold. Let $m(Z_i, X_i, U_i) = \E(Y_i|Z_i, X_i, U_i)$ denote the outcome mean model and $\text{Var}(Y_i|Z_i = k, X_i, U_i) = \sigma^2_k$ $\forall i$. Without loss of generality, let $\bar{U} = 0$ and $\bm{U}^T \bm{U} = 1$. Further suppose that the following conditions hold: 
    \begin{enumerate}
        \item The Moran's I statistic for $\bm{U}$, relative to the spatial covariance matrix $\bm{S}$, is $\mathcal{I}(\bm{U};\bm{S})$.
        \item The augmenting eigenvectors used to construct $\hat{\tau}_{SW}$ include $H$ eigenvectors of $\bm{S}$ corresponding to the highest eigenvalues. Let $\bm{V}_H \in \mathbb R^{n \times H}$ be the matrix containing these eigenvectors as its columns. 
        \item The outcome mean model $m$ is Lipschitz continuous in its third argument, i.e. $\exists \epsilon_1 > 0$ such that $$|m(0, X_i, U_i) - m(0, X_i, U_i')| \leq \epsilon_1 |U_i - U_i'| \text{ for all }i.$$
        \item The basis functions $B_k$ are rich enough to approximate the outcome mean model when $\bm{U}$ is substituted with its projection onto $\bm{V}_H$, i.e. $\exists \epsilon_2 > 0$ such that $$\max_i |m(0, X_i,  {\bm{V}_H}_i\bm{c}) - \bm{B}(\bm{X}_i^\text{aug})^T \bm{\xi}| < \epsilon_2,$$ where ${\bm{V}_H}_i \in \mathbb R^{1 \times H}$ is the $i$th row of $\bm{V}_H$, $\bm{c} = \bm{V}_H^T \bm{U} \in \mathbb R^{H \times 1}$, and $\bm{\xi} = (\xi_1, \ldots, \xi_K)$ are basis coefficients.  
\end{enumerate}
    Then the absolute bias of the spatial weighting estimator $\hat{\tau}_{SW}$ in estimating $\tau_{ATT}$, conditional on $(\bm{X}, \bm{Z}, \bm{U})$, is: \begin{align*}
|\E(\hat{\tau}_{SW}|\bm{X}, \bm{Z}, \bm{U}) -\tau_{ATT}|&= \bigg\vert\frac{1}{n_t}\sum_{Z_i = 1} m(0, X_i, U_i) - \sum_{Z_i = 0} w_i m(0, X_i, U_i)\bigg\vert\\
    &\leq \epsilon_3 + 2\bigg[\epsilon_2 +  \epsilon_1 \sqrt{\frac{\lambda_1 (1- \mathcal{I}(\bm{U};\bm{S}))}{\lambda_1 - \lambda_{H+1}}}\bigg]
\end{align*}where $\epsilon_3 = \sum_{k = 1}^K |\xi_k| \delta_k$ and the variance of the spatial weighting estimator, conditional on $(\bm{X}, \bm{Z}, \bm{U})$, is: \begin{align*}
    \text{Var}(\hat{\tau}_{SW}|\bm{X}, \bm{Z}, \bm{U}) &= \frac{\sigma^2_1}{n_t} + \sigma^2_0\sum_{Z_i = 0} w_i^2.
\end{align*} 
\end{proposition}

Proposition \ref{thm:robust-weighting-bias} shows that, under mild conditions, the spatial weighting estimator incurs only a small bias when the unmeasured confounder $\bm{U}$ is spatially autocorrelated, i.e. $\mathcal{I}(\bm{U};\bm{S})$ is close to 1. We briefly describe some intuition behind this result. If $\mathcal{I}(\bm{U};\bm{S})$ is large, then it is closely approximated by a linear combination of the $H$ eigenvectors included in the augmenting set. Since the outcome mean function 
$m$ is Lipschitz, replacing $\bm{U}$ with its eigenvector approximation changes 
$m(0, X_i, U_i)$ only slightly. Balancing functions of the augmenting eigenvectors therefore nearly balances $m(0, X_i, U_i)$, which is precisely what keeps the conditional bias small. Additionally, the conditional variance of the spatial weighting estimator depends linearly on $\sum_{Z_i = 0} w_i^2$, which is itself directly related to the variance of the weights, $(\sum_{Z_i = 0} w_i^2 - 1/\sum_{i = 1}^n (1-Z_i))/\sum_{i = 1}^n (1-Z_i)$. 

\begin{proof}
    Let $\bm{V}_H$ denote the matrix of $H$ eigenvectors of $\bm{S}$ included in the augmentation set and $\bm{V}_{-H}$ denote the matrix of the $n-H$ remaining eigenvectors of $\bm{S}$. First, we will bound $|\bm{V}_{-H} \bm{V}_{-H}^T \bm{U}|$, establishing that $\bm{U}$ is well-approximated by the first $H$ eigenvectors if $\bm{U}$ has a high Moran's I. We have
    \begin{align*}
        \bm{U}^T \bm{S} \bm{U} &= \bm{U}^T \bm{V}_H \Lambda_H \bm{V}_H^T \bm{U} + \bm{U}^T \bm{V}_{-H} \Lambda_{-H} \bm{V}_{-H}^T \bm{U}\\
        &\leq \lambda_1 \bm{U}^T \bm{V}_H \bm{V}_H^T \bm{U} + \lambda_{H + 1}\bm{U}^T \bm{V}_{-H} \bm{V}_{-H}^T \bm{U}\\
        &= \lambda_1 (1-\bm{U}^T \bm{V}_{-H} \bm{V}_{-H}^T \bm{U}) + \lambda_{H + 1}\bm{U}^T \bm{V}_{-H} \bm{V}_{-H}^T \bm{U}\\
        &= \lambda_1 + (\lambda_{H+1} - \lambda_1)\bm{U}^T \bm{V}_{-H} \bm{V}_{-H}^T \bm{U}\\
        \implies \bm{U}^T \bm{V}_{-H} \bm{V}_{-H}^T \bm{U} &\leq \frac{\lambda_1 - \bm{U}^T \bm{S} \bm{U}}{\lambda_1 - \lambda_{H +1}}\\
        &= \frac{\lambda_1 (1- \mathcal{I}(\bm{U};\bm{S}))}{\lambda_1 - \lambda_{H+1}}.
    \end{align*}Therefore, as $\mathcal{I}(\bm{U};\bm{S}) \rightarrow 1$, $\bm{U}^T \bm{V}_{-H} \bm{V}_{-H}^T \bm{U} \rightarrow 0$. This provides a bound for $|U_i - {\bm{V}_H}_i\bm{c}|$ where $\bm{c} = \bm{V}^T \bm{U}$. 
    \begin{align*}
        |U_i - {\bm{V}_H}_i\bm{c}| &= |(\bm{V}_{-H}\bm{V}_{-H}^T \bm{U})_i| \leq |\bm{V}_{-H}\bm{V}_{-H}^T \bm{U}| = \sqrt{\bm{U}^T \bm{V}_{-H}\bm{V}_{-H}^T \bm{U}} \leq \sqrt{\frac{\lambda_1 (1- \mathcal{I}(\bm{U};\bm{S}))}{\lambda_1 - \lambda_{H+1}}}.
    \end{align*}Since $m$ is Lipschitz-continuous in its third argument, then \begin{align*}
        |m(0, X_i, U_i) - m(0, X_i, {\bm{V}_H}_i\bm{c})| \leq \epsilon_1 |U_i - {\bm{V}_H}_i\bm{c}| \leq \epsilon_1 \sqrt{\frac{\lambda_1 (1- \mathcal{I}(\bm{U};\bm{S}))}{\lambda_1 - \lambda_{H+1}}}.
    \end{align*}for all $i$. Since $\max_i |m(0, X_i,  {\bm{V}_H}_i\bm{c}) - \bm{B}(\bm{X}_i^\text{aug})^T \bm{\xi}| < \epsilon_2$, \begin{align*}
        \max_i |m(0, X_i,  U_i) - \bm{B}(\bm{X}_i^\text{aug})^T \bm{\xi}| &\leq \max_i |m(0, X_i,  {\bm{V}_H}_i \bm{c}) - \bm{B}(\bm{X}_i^\text{aug})^T \bm{\xi}| \\
        &\hspace{1in}+ \max_i |m(0, X_i, U_i) - m(0, X_i, {\bm{V}_H}_i\bm{c})|\\
        &\leq \epsilon_2 + \epsilon_1 \sqrt{\frac{\lambda_1 (1- \mathcal{I}(\bm{U};\bm{S}))}{\lambda_1 - \lambda_{H+1}}}.
    \end{align*}The weights of the  spatial weighting estimator satisfy \begin{align*}
        \bigg \vert \sum_{i:Z_i = 0} w_i B_k (\bm{X}^{\text{aug}}_i) - \frac{1}{n_t}\sum_{i:Z_i = 1} B_k (\bm{X}^{\text{aug}}_i) \bigg \vert &\leq \delta_k, k = 1,2,\ldots, K,\\
        \bigg \vert \sum_{i = 1}^n (w_i(1-Z_i) - \frac{1}{n_t}Z_i) B_k (\bm{X}^{\text{aug}}_i) \bigg \vert &\leq \delta_k, k = 1,2,\ldots, K.\\
        \implies \bigg \vert \sum_{i = 1}^n (w_i(1-Z_i) - \frac{1}{n_t}Z_i) \bm{B} (\bm{X}^{\text{aug}}_i)^T\bm{\xi} \bigg \vert &= \bigg \vert \sum_{i = 1}^n (w_i(1-Z_i) - \frac{1}{n_t}Z_i) \sum_{k = 1}^K \xi_k B_k (\bm{X}^{\text{aug}}_i) \bigg \vert\\
        &= \bigg \vert \sum_{k = 1}^K \xi_k \bigg(\sum_{i = 1}^n (w_i(1-Z_i) - \frac{1}{n_t}Z_i)  B_k (\bm{X}^{\text{aug}}_i)\bigg) \bigg \vert\\
        &\leq \sum_{k = 1}^K |\xi_k| \delta_k= \epsilon_3.
    \end{align*}
    
    The absolute conditional bias is therefore
    \begin{align*}
        |\E(\hat{\tau}_{SW}|\bm{Z}, \bm{X}, \bm{U}) - \tau_{ATT}| &= \bigg\vert\frac{1}{n_t}\sum_{Z_i = 1} m(0, X_i, U_i) - \sum_{Z_i = 0} w_i m(0, X_i, U_i)\bigg\vert\\
        &= \bigg\vert\sum_{i = 1}^n (\frac{1}{n_t}Z_i - (1-Z_i)w_i) m(0, X_i, U_i)\bigg\vert \\
        &\leq \bigg\vert\sum_{i = 1}^n (\frac{1}{n_t}Z_i - (1-Z_i)w_i) \bm{B}(\bm{X}_i^\text{aug})^T \bm{\xi}\bigg\vert \\
        &\hspace{1in}+ \bigg(\epsilon_2 + \epsilon_1 \sqrt{\frac{\lambda_1 (1- \mathcal{I}(\bm{U};\bm{S}))}{\lambda_1 - \lambda_{H+1}}}\bigg) \sum_{i =1}^n |\frac{1}{n_t}Z_i - (1-Z_i)w_i|\\
        &\leq \epsilon_3 + 2\bigg[\epsilon_2 +  \epsilon_1 \sqrt{\frac{\lambda_1 (1- \mathcal{I}(\bm{U};\bm{S}))}{\lambda_1 - \lambda_{H+1}}}\bigg].
    \end{align*}

    Next, let $$l_i = \begin{cases} -w_i & Z_i = 0 \\ \frac{1}{n_t} & Z_i = 1 \end{cases}.$$ The conditional variance of the spatial weighting estimator is  \begin{align*}
        \text{Var}(\hat{\tau}_{SW}|\bm{Z}, \bm{X}, \bm{U}) &= \text{Var}(\sum_{i =1}^n l_i Y_i| Z_i, \bm{X}_i, U_i)\\
        &= \sum_{i = 1}^n l_i^2 \text{Var}(Y_i| Z_i, \bm{X}_i, U_i)\\
        &= \frac{\sigma^2_1}{n_t} + \sigma^2_0 \sum_{i:Z_i = 0}w_i^2.
    \end{align*}
\end{proof}

\section{Simulation}
\label{sec:simulation}

Using the observed covariates $\bm{X}$ and binary treatment $\bm{Z}$ from our Superfund dataset, we create datasets subject to unmeasured spatial confounding. We estimate the ATT using the spatial weighting estimator of Section \ref{sec:extension}, augmenting the covariate matrix $\bm{X}$ with $J = 150$ eigenvectors. The objective of the simulations is twofold. First, we investigate whether the spatial weighting estimator can mitigate confounding bias from three distinct classes of unmeasured confounders with different patterns of spatial autocorrelation. Second, we evaluate its performance while varying the complexity of the outcome model. All simulation code can be found at \url{https://github.com/NSAPH-Projects/spatial_regression_weighting}.

For $i = 1, \ldots, n = 1429,$ we generate unmeasured confounders $\bm{U}$ using three different mechanisms, each inducing a different spatial autocorrelation structure. For all three mechanisms, we generate $\bm{U}$ with a two-stage procedure: 1) we simulate $U_i \sim \mathcal{N}(Z_i, 0.1^2)$ and 2) we set $U_i = \sum_{j = 1}^n W_{ij} U_j / \sum_{j = 1}^n W_{ij}$. The first class of unmeasured confounders (\textit{cluster-based}) is designed to be constant within states, hence we define $W_{ij} = I(C_i = C_j)$. The second class of unmeasured confounders (\textit{adjacency-based}) is designed to be smooth across adjacencies, hence we define $W_{ij} = (A^{100})_{ij}$. The third class of unmeasured confounders (\textit{distance-based}) is designed to vary smoothly with geographical distance, hence we define $W_{ij} = \text{exp}(-d_{ij}/ 50\text{km})$. Figure \ref{fig:simU} provides a plot of one observation of $\bm{U}$ from each class. 

We further consider three possible outcome models. The linear outcome model generates outcome as $Y_i = \bm{\beta}^T \bm{X}_i + Z_i - 0.2 U_i + \mathcal{N}(0, 0.1^2)$. The linear-interaction outcome model generates outcome as $Y_i = \bm{\beta}^T \bm{X}_i + Z_i - 0.5 U_i + U_i Z_i + \mathcal{N}(0, 0.1^2)$. The nonlinear-interaction outcome model generates outcome as $Y_i = \bm{\beta}^T \bm{X}_i + \tau Z_i + Z_i \sin(U_i) - U_i^2 + U_i Z_i(X_{2i}+1) + 0.3 Z_i X_{1i}^2 + \mathcal{N}(0, 0.1^2)$. We set $\beta = (-0.44,0.46,-0.69,-1.45, 0.57,-1.02,-0.02,-0.94,1.10,-0.48,-0.71,$ $ -0.937, -0.091, -0.576, 0.234, 1.578, 0.936, -0.537)^T$. These coefficients were generated from a standard normal distribution. For each of the $9$ data-generating scenarios produced by the three classes of unmeasured confounders and three outcome models, we create $M = 1000$ datasets of size $n = 1429$, with $\bm{X}$ and $\bm{Z}$ drawn directly from our Superfund dataset.

Under each data-generating scenario, we estimate the ATT using six different approaches. The first approach (\textit{OLS}) fits an ordinary least squares regression of $Y$ on $X$ and $Z$ and extracts the coefficient estimate of $Z$. Approaches two through four employ generalized least squares to fit the random effects (\textit{RE}), conditional autoregressive (\textit{CAR}), and Gaussian process (\textit{GP}) models described in Section \ref{sec:threeexamples}, likewise extracting the coefficient estimate of $Z$. The fifth approach (\textit{spatial coordinates}) incorporates the latitude and longitude of each Superfund site as two additional covariates in an augmented inverse propensity score estimator, following the framework proposed by \citesupp{gilbert2021causal}. For additional details on this estimator, see below. The sixth approach (\textit{spatial weighting}) implements our proposed methodology from Section \ref{sec:extension}. In order to adjust for a heterogeneous class of unmeasured confounders with spatial autocorrelation, we augment $\bm{X}$ with $50$ eigenvectors corresponding to the largest eigenvalues from each spatial regression model above, resulting in $J = 150$. 
We use the package \url{sbw} for implementation \citepsupp{zubizarreta2015stable}. 

We evaluate the performance of the six approaches in estimating the ATT through bias and root mean squared error. Figure \ref{fig:boxplot} presents boxplots of the ATT estimates and Table \ref{tab:simresults} presents numerical results. 

\begin{figure}[!ht]
    \centering
    \includegraphics[width=\linewidth]{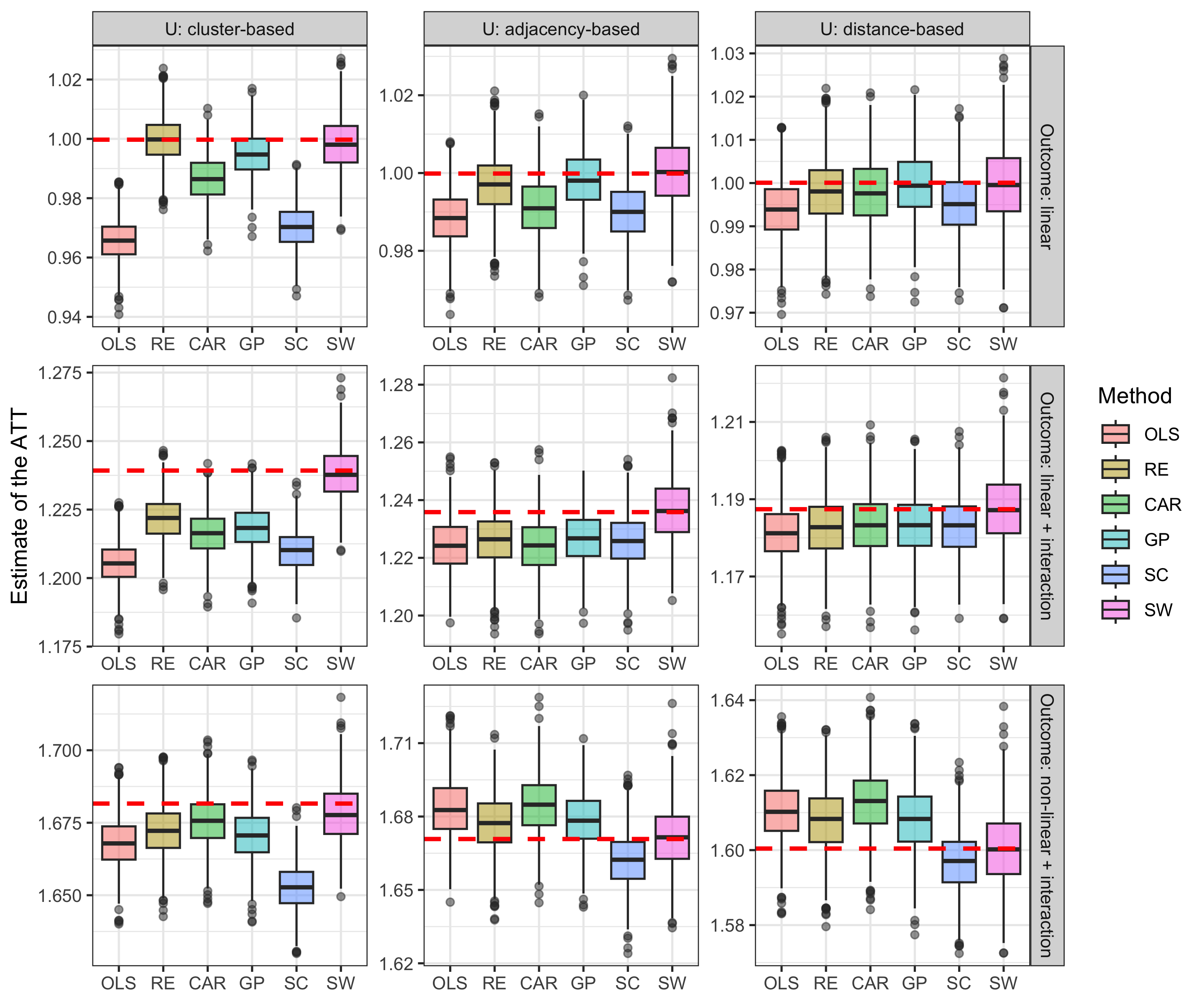}
    \caption{Estimates of the ATT derived from six approaches. OLS, ordinary least squares; RE, random effects; CAR, conditional autoregressive; GP, Gaussian process; SC, spatial coordinates; SW, spatial weighting. The red dotted line represents the true ATT. }
    \label{fig:boxplot}
\end{figure}

We briefly highlight two key observations. First, the results under the linear outcome model, in the first row of the boxplot grid, provide empirical support for our findings in Section \ref{sec:space-diagnostics}. Specifically, the random effects model yields unbiased estimates under cluster-based confounding and the Gaussian process model yields unbiased estimates under distance-based confounding. 
However, when nonlinearities or heterogeneity of treatment effects by $\bm{U}$ is present in the outcome model, all three estimators suffer from bias. 

Adjustment for spatial coordinates can reduce bias from unmeasured spatial confounding when the confounder is a measurable, or nearly continuous, function of spatial coordinates and the treatment exhibits non-spatial variation \citepsupp{gilbert2021causal}. These assumptions are not explicitly
satisfied by any of the confounding mechanisms considered, and the spatial coordinates estimator exhibits substantial bias across most data-generating scenarios. 

In contrast, our proposed spatial weighting estimator remains approximately unbiased across all nine data‐generating scenarios and---in terms of bias---outperforms each competing model in seven of the nine scenarios, although sometimes exhibits larger uncertainty. These simulation results demonstrate that the spatial weighting estimator can flexibly accommodate complex nonlinearities and heterogeneity of treatment effects, while simultaneously adjusting for unmeasured confounders that are constant within states, smooth across adjacencies, or vary continuously over space.

\subsection*{Spatial coordinates estimator}
For the \textit{spatial coordinates} approach, we follow \citesupp{gilbert2021causal} and incorporate latitude and longitude as two additional covariates into an augmented inverse propensity weighting (AIPW) estimator. Specifically, we estimate the propensity score $\pi(\bm{X}_i, \text{lat}_i, \text{long}_i) = \mathbb P(Z_i = 1|\bm{X}_i, \text{lat}_i, \text{long}_i)$ and outcome model $m_0(\bm{X}_i, \text{lat}_i, \text{long}_i) = \E(Y_i|Z_i = 0, \bm{X}_i, \text{lat}_i, \text{long}_i)$ with the SuperLearner package in R using a combination of candidate learners, including multivariate adaptive regression splines (SL.earth), generalized additive models (SL.gam), generalized linear
models (SL.glm), mean regression (SL.mean), and an interaction model (SL.glm.interaction). The AIPW estimator of the ATT is \citepsupp{mercatanti2014debit, moodie2018doubly}:
\begin{align*}
    \hat{\tau}_{AIPW} &= \frac{1}{\sum_{i = 1}^n Z_i}\bigg[\sum_{i = 1}^n Y_i Z_i - \frac{Y_i(1-Z_i)\hat{\pi}(\bm{X}_i, \text{lat}_i, \text{long}_i) + \hat{m}_0(\bm{X}_i, \text{lat}_i, \text{long}_i)(Z_i - \hat{\pi}(\bm{X}_i, \text{lat}_i, \text{long}_i))}{1-\hat{\pi}(\bm{X}_i, \text{lat}_i, \text{long}_i)}\bigg].
\end{align*}

\begin{figure}[!ht]
    \centering
    \begin{subfigure}[t]{0.32\textwidth}
        \centering
        \includegraphics[height=1.5in]{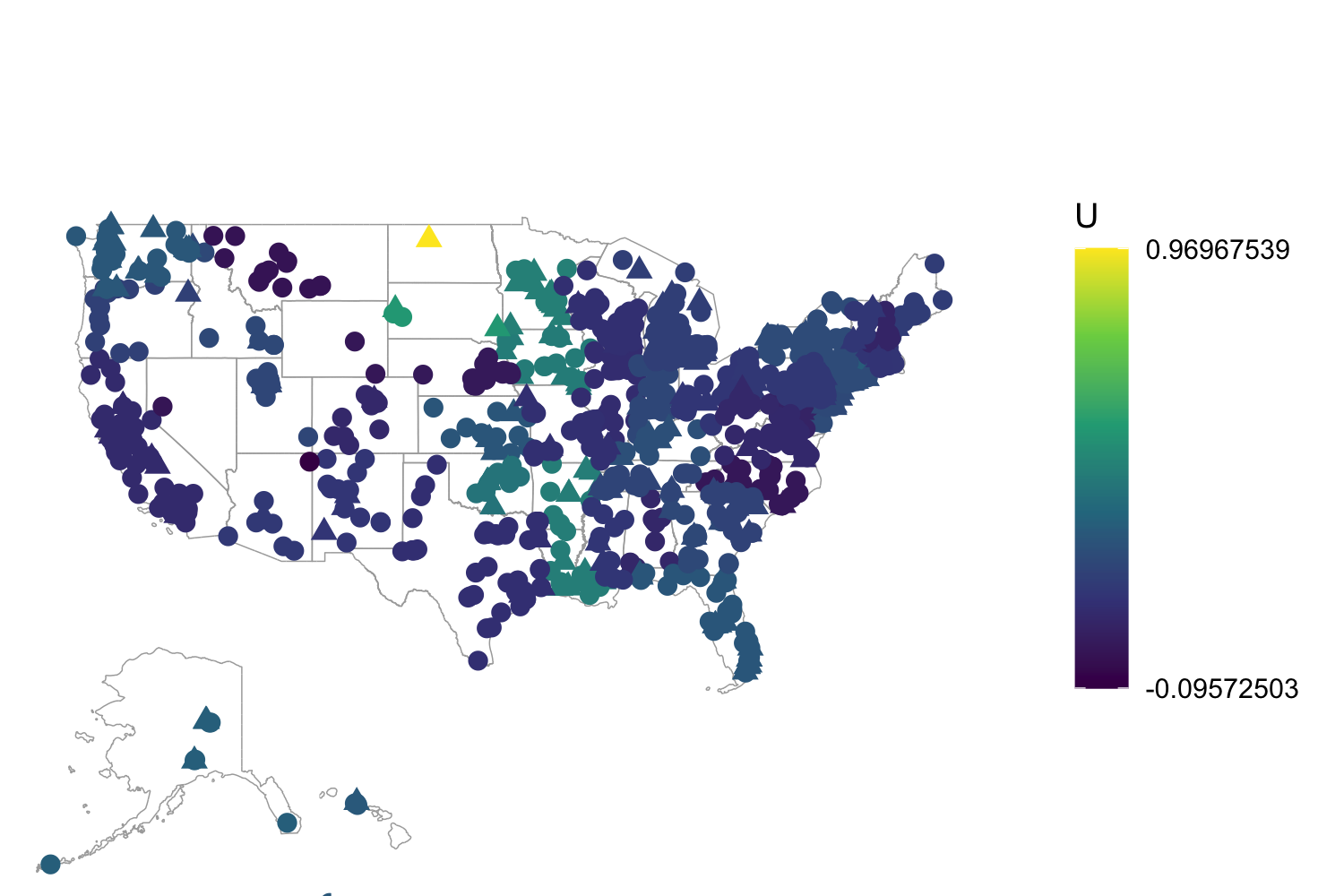}
        \caption{Cluster-based.}
    \end{subfigure}
    \hfill
    \begin{subfigure}[t]{0.32\textwidth}
        \centering
        \includegraphics[height=1.5in]{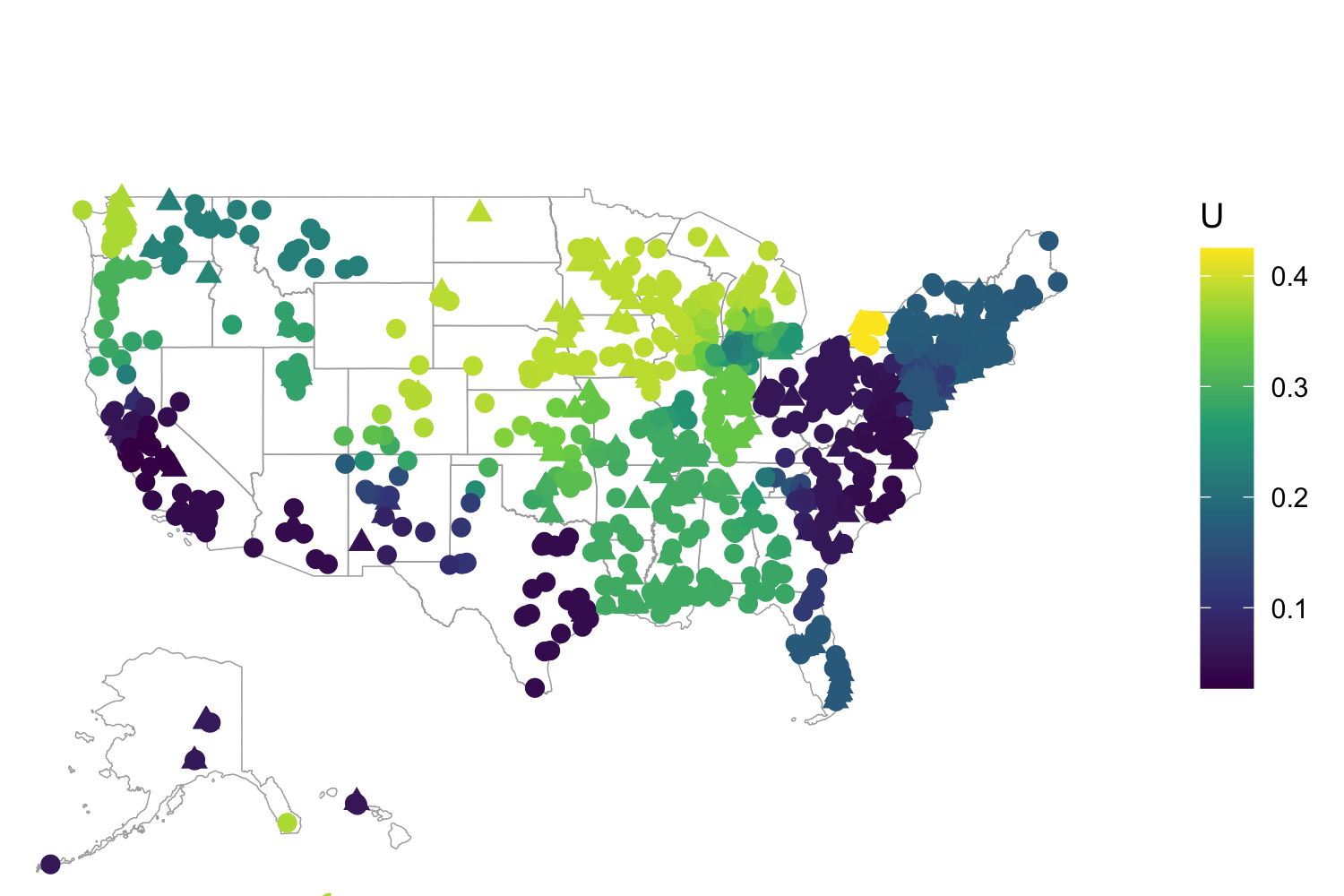}
        \caption{Adjacency-based.}
    \end{subfigure}
    \hfill
    \begin{subfigure}[t]{0.32\textwidth}
        \centering
        \includegraphics[height=1.5in]{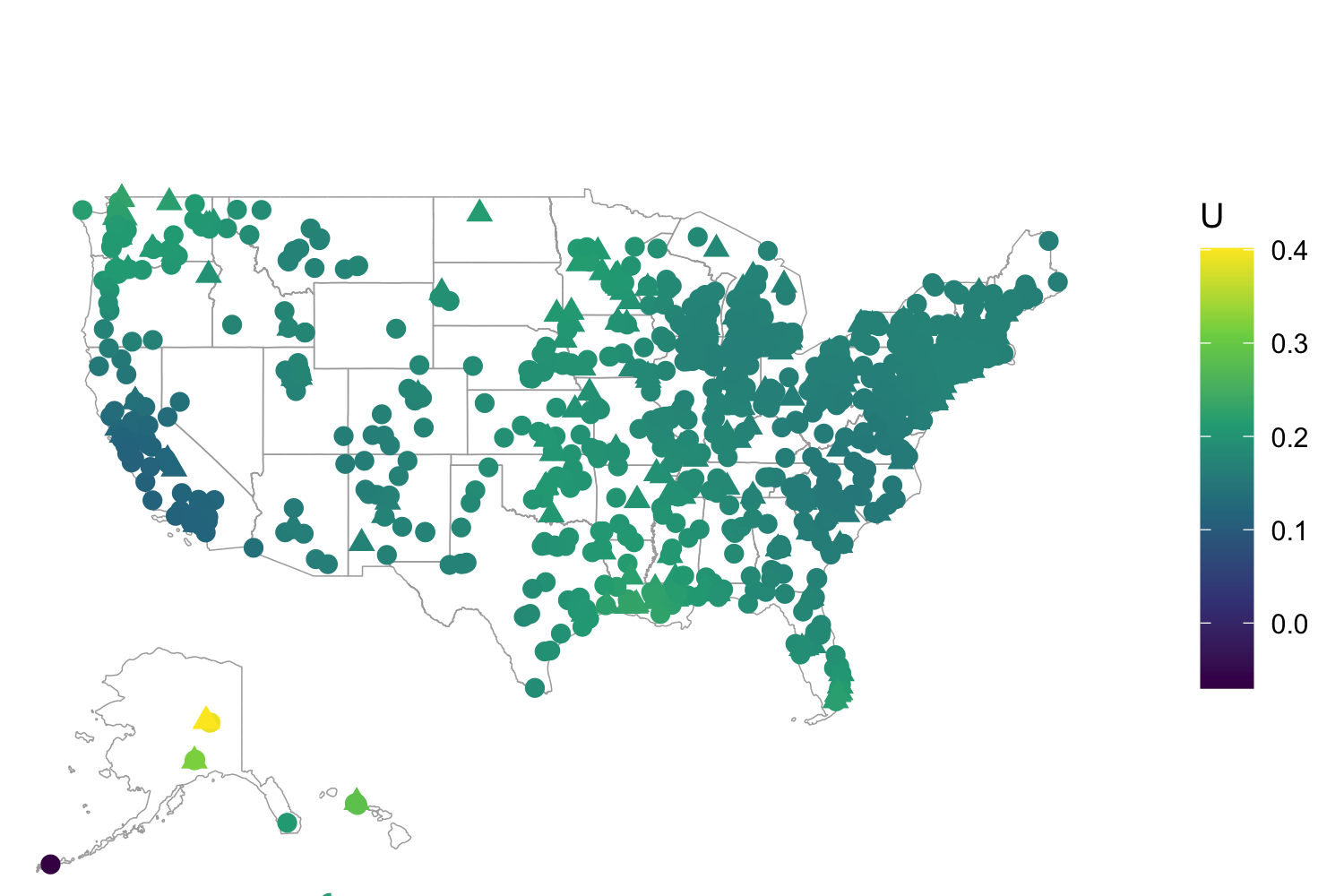}
        \caption{Distance-based.}
    \end{subfigure}
    \caption{One observation of $\bm{U}$ for each of the three confounding mechanisms (cluster-based, adjacency-based, distance-based). In the simulation, we reverse the conventional two stage process of generating $(\bm{X}, \bm{U})$ followed by $\bm{Z}|(\bm{X}, \bm{U})$. Instead, we fix $(\bm{X}, \bm{Z})$ to their observed values in the Superfund dataset and generate $\bm{U}|\bm{Z}$, in order to retain the observed treatment and confounder structure and closely reflect the real data. Although this data-generating process reverses the conventional order of generating $(\bm{X}, \bm{U})$ and $\bm{Z}$, we do not believe this choice should materially affect the conclusions of the simulation.}
    \label{fig:simU}
    \end{figure}

\begin{table}[!ht]
\centering
\begin{tabular}{llrrrrrr}
Unmeasured confounding & Outcome Model & OLS & RE & CAR & GP & SC & SW \\ 
  \hline
  \multicolumn{7}{c}{Bias} \\
\hline
\multirow{3}{*}{cluster-based} & linear & -3.39 & -0.01 & -1.33 & -0.49 & -2.95 & -0.16 \\ 
 & linear-interaction & -3.38 & -1.75 & -2.31 & -2.08 & -2.92 & -0.13 \\
 & nonlinear-interaction & -1.35 & -0.94 & -0.63 & -1.11 & -2.89 & -0.36 \\ 
 \multirow{3}{*}{adjacency-based} & linear & -1.14 & -0.30 & -0.89 & -0.17 & -0.99 & 0.05 \\ 
 & linear-interaction & -1.16 & -0.96 & -1.17 & -0.91 & -0.98 & 0.06 \\ 
 & nonlinear-interaction & 1.23 & 0.63 & 1.37 & 0.77 & -0.84 & 0.09 \\
  \multirow{3}{*}{distance-based} & linear &-0.61 & -0.21 & -0.23 & -0.05 & -0.49 & -0.05 \\ 
   & linear-interaction & -0.61 & -0.47 & -0.42 & -0.42 & -0.45 & -0.01 \\ 
   & nonlinear-interaction & 1.00 & 0.76 & 1.23 & 0.77 & -0.36 & 0.00 \\ 
   \hline
\multicolumn{7}{c}{RMSE} \\
  \hline
  \multirow{3}{*}{cluster-based} & linear & 3.46 & 0.75 & 1.54 & 0.89 & 3.04 & 0.91 \\ 
   & linear-interaction & 3.46 & 1.93 & 2.45 & 2.23 & 3.02 & 0.94 \\ 
   & nonlinear-interaction & 1.59 & 1.30 & 1.11 & 1.41 & 3.01 & 1.07 \\ 
    \multirow{3}{*}{adjacency-based} & linear &  1.34 & 0.81 & 1.19 & 0.76 & 1.23 & 0.89 \\ 
   & linear-interaction & 1.48 & 1.36 & 1.52 & 1.28 & 1.36 & 1.07 \\ 
   & nonlinear-interaction & 1.72 & 1.36 & 1.83 & 1.37 & 1.42 & 1.26 \\ 
  \multirow{3}{*}{distance-based} & linear & 0.93 & 0.78 & 0.81 & 0.74 & 0.87 & 0.89 \\ 
   & linear-interaction & 0.97 & 0.93 & 0.93 & 0.89 & 0.88 & 0.92 \\
   & nonlinear-interaction & 1.30 & 1.16 & 1.53 & 1.15 & 0.90 & 0.97 \\
  \hline
  
\end{tabular}
\caption{Bias and variance of the ATT estimates in simulation. All values have been multiplied by $10^2$. OLS, ordinary least squares; RE, random effects; CAR, conditional autoregressive; GP, Gaussian process; SC, spatial coordinates; SW, spatial weighting. }
\label{tab:simresults}
\end{table}

\clearpage
\section{Data application details}
\label{suppsec:dataapp}
\subsection*{Definition of binary treatment}
Figure \ref{fig:treatment-definition} presents a schematic summarizing the classification of sites as treated, control, or ineligible for inclusion in our study. Site listing and deletion dates were obtained from the Superfund National Priorities List (NPL) Where You Live Map \citesupp{EPA2025SearchSitesWhereYouLive}. Remediation start dates were obtained by web scraping site profile pages (\url{https://cumulis.epa.gov/supercpad/SiteProfiles/index.cfm?fuseaction=second.schedule\&id=SEMS_ID}). As described in Section \ref{sec:motivatingdata}, Superfund site remediation follows a multistage process that includes site assessment, official listing, planning of remedial design, implementation of cleanup actions, and post-construction monitoring until the site is considered protective of human health and the environment. Our analysis focuses on quantifying the effect of the cleanup phase itself, which we define as the period between the start of remedial action and the site's deletion from the NPL. 

\begin{figure}[!ht]
    \centering
    \includegraphics[width=0.7\linewidth]{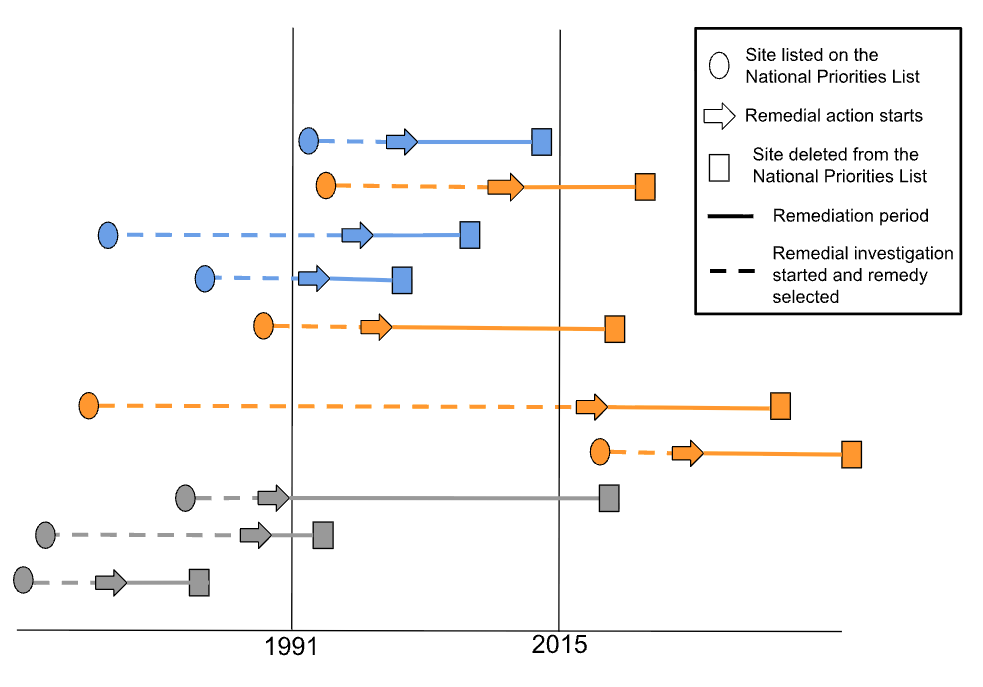}
    \caption{A diagram illustrating the definition of binary treatment according to listing date, remedial action start date, and site deletion date. The study includes all sites currently on the National Priorities List (NPL) excluding those where remedial action had started before January 1, 1991. Sites are defined as treated (in blue) if remediation was successfully completed by December 31, 2015, i.e. the site was deleted from the NPL. Sites are defined as control (in orange) if remediation was not successfully completed by December 31, 2015, i.e. the site was not yet deleted from the NPL, regardless of whether remedial action had started. In gray, we provide three examples of sites that do not meet our study inclusion criteria.}
    \label{fig:treatment-definition}
\end{figure}

\subsection*{Dataset characteristics}
Table \ref{tab:dataset} reports summary statistics for the binary treatment, outcome, and measured confounders included in the Superfund site-level analysis ($n = 1429$). There was 1 missing value in median household income, 2 missing values in median house value, 1 missing value in median year housing structure built, and 188 missing values in each of presence of a metal, presence of a volatile organic compound, presence of contamination through solids or soil, presence of contamination through a liquid medium, and presence of contamination through air, soil gas, or landfill gas. Missing values were addressed using single imputation. For stable estimation, we further excluded sites with fewer than $10$ Medicaid-covered live births within the site area and surrounding 
2-kilometer buffer, resulting in a final dataset of $n = 1079$ Superfund sites.

\renewcommand{\arraystretch}{1.4}
\begin{table}[!ht]
\footnotesize
    \begin{tabular}{p{1.65cm}| p{5cm} p{2.75cm} p{5.6cm}}
    & Variables & Mean(sd) & Data Source\\
    \hline \hline
    Binary treatment & Superfund site remediation (remediation action starts and site deleted from the NPL) between 1991-2015 & 0.181 (0.385) & Superfund National Priorities List (NPL) Where You Live Map \citepsupp{EPA2025SearchSitesWhereYouLive, EPA2025cleanup, EPA2025NPLBoundaries}. Remediation dates obtained from site profile webpages (\url{https://cumulis.epa.gov/supercpad/SiteProfiles/index.cfm?fuseaction=second.schedule\&id=SEMS_ID})\\
    \hline
    \multirow{2}{*}{Outcome} & \% newborns that are ``small vulnerable'' during 2016--2018  & 13.5 (3.60)  & \multirow{2}{6cm}{Medicaid and CHIP TAF Inpatient file obtained from the Centers for Medicare and Medicaid Services \citepsupp{cms_taf_ip}} \\
    & \% newborns with congenital anomalies during 2016--2018 & 13.6 (5.28) & \\
    \hline
    \multirow{14}{1.5cm}{Measured confounders} & population density (people/mi$^2$) & 1710 (2760) &\multirow{11}{6cm}{1990 U.S. Decennial Census, downloaded at the census tract-level from NHGIS \citepsupp{Schroeder2025NHGIS, USCensus1990}}\\
    & Proportion of Hispanic residents & 0.0656 (0.120) &  \\
    & Proportion of Black residents & 0.107 (0.163) & \\
    & Proportion of Asian residents & 0.0229 (0.0491) & \\
    & Proportion of Indigenous residents & 0.00814 (0.0340) & \\
    & Proportion of renter occupied housing &0.320 (0.160) & \\
    & Median household income (\$) & 32400 (11200) & \\
    & Median house value (\$) & 93100 (60200) & \\ 
    & Proportion of residents in poverty & 0.116 (0.0855) & \\
    & Proportion of residents with a high school diploma & 0.741 (0.125) & \\
    & Median year housing structure built & 1964 (9.97) & \\
    \cline{2-4}
    & Site score & 42.8 (11.3) & \multirow{6}{6cm}{Contaminants of Concern \citepsupp{EPA2025COCSpreadsheet}, Superfund National Priorities List (NPL) Where You Live Map \citepsupp{EPA2025SearchSitesWhereYouLive}}\\
    & Presence of a metal or metalloid & 0.746 (0.435) &\\
    & Presence of a volatile organic compound & 0.712 (0.453) &\\
    & Presence of contamination through solids or soil & 0.840 (0.367) & \\
     & Presence of contamination through a liquid medium & 0.879 (0.327) & \\
      & Presence of contamination through air, landfill gas, or soil gas & 0.131 (0.337) & \\
    \end{tabular}
    \caption{Description of final Superfund site-level dataset ($n = 1079$) and data sources.}
    \label{tab:dataset}
\end{table}

\clearpage
\subsection*{Birth outcomes}
Algorithm \ref{alg1} describes how the ZCTA-level counts of total live births, small vulnerable newborns, and newborns with congenital anomalies over 2016--2018 were obtained from the Medicaid and CHIP T-MSIS Analytic File (TAF) Inpatient file. 
\singlespacing
\small
\begin{algorithm}[H]
\caption{Birth Outcome ZCTA-level Dataset Construction (2016--2018)}
\label{alg1}
\DontPrintSemicolon
\SetKwInOut{Input}{Input}
\SetKwInOut{Output}{Output}

\Input{Annual inpatient (IP) header files and denominator/eligibility (DE) files for years $y \in \{2016, 2017, 2018\}$.}
\Output{ZCTA-level counts of total live births, small vulnerable newborns, and newborns with congenital anomalies over 2016--2018.}

\ForEach{$y \in \{2016, 2017, 2018\}$}{
\textbf{1) Read in IP header records.}\;
Load IP header data for year $y$, retaining claim id, beneficiary id, diagnoses, service begin date, service end date, admission hour, discharge date, discharge hour. Keep only records containing at least one ICD-10-CM diagnosis code beginning with \texttt{Z38}.\;

\textbf{2) Merge with DE file.}\;
Load the DE file for year $y$, retaining ZCTA, birth date, and beneficiary state code. Merge to IP data on beneficiary id.\;

\textbf{3) Drop non-service payment records.}\;
Exclude claims with claim type code  $\in \{2, 4, \texttt{B}, \texttt{D}, \texttt{V}, \texttt{X}\}$.\;

\textbf{4) Create claim-level indicators.}\;
\textbf{(a)} Define a binary flag for small vulnerable newborn, equal to $1$ if any ICD-10-CM diagnosis code begins with one of: $
\texttt{P070},\ \texttt{P071},\ \texttt{P072},\ \texttt{P073},\ \texttt{P05},\ 
\texttt{Z3A0--Z3A36},\ \texttt{O35}$
and $0$ otherwise.\\
\textbf{(b)} Define a binary flag for congenital anomalies, equal to $1$ if any ICD-10-CM diagnosis begins with \texttt{Q}, and $0$ otherwise.\;

\textbf{5) Aggregate to the beneficiary-year level.}\;
For each beneficiary, take the maximum of each flag across all claims in year $y$ and further de-duplicate by (ZCTA, service begin date, service end date, admission date, admission hour, discharge date, discharge hour, birth date). 
}

\textbf{Post-processing:}\;
\begin{enumerate}[label=\alph*)]
\item \textbf{Merge across years.}
Append the beneficiary-level datasets for $y \in \{2016, 2017, 2018\}$ into a single file.
\vspace{-1em}
\item \textbf{Aggregate to the ZCTA level.}
Sum the counts of total live births, small vulnerable newborns, and newborns with congenital anomalies within each ZCTA.
\end{enumerate}
\end{algorithm}
\normalsize
We only have access to birth claims in the Inpatient file, which likely explains the undercounting in Table \ref{tab:statecounts}, as some states record birth claims only in the Other Services file \citepsupp{auty2024comparing}. However, we do not believe resulting missingness in total birth count is necessarily systematically related to our outcomes of interest, the percentage of small vulnerable newborns and the percentage of congenital anomalies. We exclude sites in Alabama, Colorado, and Rhode Island from our final analysis because the reported statistics for these states were irregular, with estimated proportions of small vulnerable newborns and congenital anomalies implausibly outside the expected 5–20\% range \citepsupp{lawn2023small, marden1964congenital}.

\doublespacing

\newpage
\begin{table}[!ht]
\footnotesize
\renewcommand{\arraystretch}{0.9}
    \centering
    \begin{tabular}{p{0.95in}ccccccc}
    state & \# live births & \makecell[l]{\# live births\\(CDC)} & error & \makecell[l]{\# small\\vulnerable} & \makecell[l]{\% small\\vulnerable} & \makecell[l]{\# congenital\\anomalies} & \makecell[l]{\% congenital\\anomalies} \\
    \hline
    \hline
    Alabama & 15651 & 87410 & -82.09 & 6654 & 0.43 & 2795 & 0.18 \\ 
  Alaska & 12602 & 12345 & 2.08 & 1569 & 0.12 & 1534 & 0.12 \\ 
  Arizona & 119946 & 128379 & -6.57 & 13493 & 0.11 & 11309 & 0.09 \\ 
  Arkansas & 61707 & 49105 & 25.66 & 7555 & 0.12 & 6476 & 0.10 \\ 
  California & 243504 & 604780 & -59.74 & 31456 & 0.13 & 28897 & 0.12 \\ 
  Colorado & 17759 & 75615 & -76.51 & 6074 & 0.34 & 3451 & 0.19 \\ 
  Connecticut & 45601 & 39598 & 15.16 & 5981 & 0.13 & 6636 & 0.15 \\ 
  Delaware & 14260 & 14512 & -1.74 & 2093 & 0.15 & 2252 & 0.16 \\ 
  D.C. & 9510 & 11833 & -19.63 & 1432 & 0.15 & 1141 & 0.12 \\ 
  Florida & 337327 & 324540 & 3.94 & 44981 & 0.13 & 60482 & 0.18 \\ 
  Georgia & 207425 & 178772 & 16.03 & 25928 & 0.12 & 27359 & 0.13 \\ 
  Hawaii & 20649 & 15812 & 30.59 & 2732 & 0.13 & 3274 & 0.16 \\ 
  Idaho & 25160 & 23987 & 4.89 & 3326 & 0.13 & 2147 & 0.09 \\ 
  Illinois & 184921 & 177171 & 4.37 & 26415 & 0.14 & 19969 & 0.11 \\ 
  Indiana & 86032 & 100304 & -14.23 & 12188 & 0.14 & 12664 & 0.15 \\ 
  Iowa & 47595 & 45263 & 5.15 & 6031 & 0.13 & 6168 & 0.13 \\ 
  Kansas & 42165 & 34712 & 21.47 & 5236 & 0.12 & 4763 & 0.11 \\ 
  Kentucky & 81656 & 77473 & 5.40 & 10629 & 0.13 & 12801 & 0.16 \\ 
  Louisiana & 77550 & 114769 & -32.43 & 13507 & 0.17 & 11203 & 0.14 \\ 
  Maine & 11921 & 14658 & -18.67 & 1581 & 0.13 & 1553 & 0.13 \\ 
  Maryland & 10385 & 84967 & -87.78 & 1463 & 0.14 & 2492 & 0.24 \\ 
  Massachusetts & 80719 & 58146 & 38.82 & 11207 & 0.14 & 15563 & 0.19 \\ 
  Michigan & 151402 & 139674 & 8.40 & 22772 & 0.15 & 21142 & 0.14 \\ 
  Minnesota & 38701 & 65573 & -40.98 & 4587 & 0.12 & 5385 & 0.14 \\ 
  Mississippi & 50838 & 69425 & -26.77 & 8373 & 0.16 & 4687 & 0.09 \\ 
  Missouri & 77770 & 85278 & -8.80 & 11067 & 0.14 & 9645 & 0.12 \\ 
  Montana & 16225 & 14566 & 11.39 & 2234 & 0.14 & 1600 & 0.10 \\ 
  Nebraska & 3559 & 25825 & -86.22 & 444 & 0.12 & 358 & 0.10 \\ 
  Nevada & 32023 & 48973 & -34.61 & 4114 & 0.13 & 5396 & 0.17 \\ 
  New Hampshire & 5533 & 9044 & -38.82 & 714 & 0.13 & 560 & 0.10 \\ 
  New Jersey & 53582 & 94093 & -43.05 & 5318 & 0.10 & 5882 & 0.11 \\ 
  New Mexico & 45371 & 39381 & 15.21 & 6029 & 0.13 & 6229 & 0.14 \\ 
  New York & 347412 & 331669 & 4.75 & 42206 & 0.12 & 54667 & 0.16 \\ 
  North Carolina & 175426 & 152964 & 14.68 & 25122 & 0.14 & 21584 & 0.12 \\ 
  North Dakota & 8270 & 7931 & 4.27 & 905 & 0.11 & 1108 & 0.13 \\ 
  Ohio & 192631 & 168816 & 14.11 & 27625 & 0.14 & 25243 & 0.13 \\ 
  Oklahoma & 82353 & 79152 & 4.04 & 10104 & 0.12 & 9567 & 0.12 \\ 
  Oregon & 59163 & 58316 & 1.45 & 7193 & 0.12 & 6989 & 0.12 \\ 
  Pennsylvania & 155760 & 135422 & 15.02 & 23210 & 0.15 & 25296 & 0.16 \\ 
  Rhode Island & 525 & 15647 & -96.64 & 140 & 0.27 &  82 & 0.16 \\ 
  South Carolina & 84936 & 84468 & 0.55 & 12038 & 0.14 & 10265 & 0.12 \\ 
  South Dakota & 13238 & 11198 & 18.22 & 1563 & 0.12 & 1252 & 0.09 \\ 
  Tennessee & 99352 & 117056 & -15.12 & 14237 & 0.14 & 13603 & 0.14 \\ 
  Texas & 549241 & 543037 & 1.14 & 68704 & 0.13 & 72485 & 0.13 \\ 
  Utah & 42669 & 35913 & 18.81 & 6370 & 0.15 & 3634 & 0.09 \\ 
  Vermont & 2298 & 7175 & -67.97 & 267 & 0.12 & 334 & 0.15 \\ 
  Virginia & 98906 & 91281 & 8.35 & 14890 & 0.15 & 15320 & 0.15 \\ 
  Washington & 114767 & 103346 & 11.05 & 13545 & 0.12 & 14479 & 0.13 \\ 
  West Virginia & 18666 & 27217 & -31.42 & 2784 & 0.15 & 2250 & 0.12 \\ 
  Wisconsin & 76596 & 70609 & 8.48 & 10126 & 0.13 & 7624 & 0.10 \\ 
  Wyoming & 4922 & 6661 & -26.11 & 686 & 0.14 & 464 & 0.09 \\  
      \end{tabular}
    \caption{Birth statistics by state. Gold-standard counts of Medicaid-covered live births are obtained from the \href{https://wonder.cdc.gov/natality.html}{Centers for Disease Control and Prevention (CDC) WONDER Natality online database}. Error is equal to (\# live births - \# CDC live births)/\# live births. }
    \label{tab:statecounts}
\end{table}

\begin{figure}
    \begin{subfigure}[t]{\textwidth}
        \centering
        \includegraphics[height=3.75in]{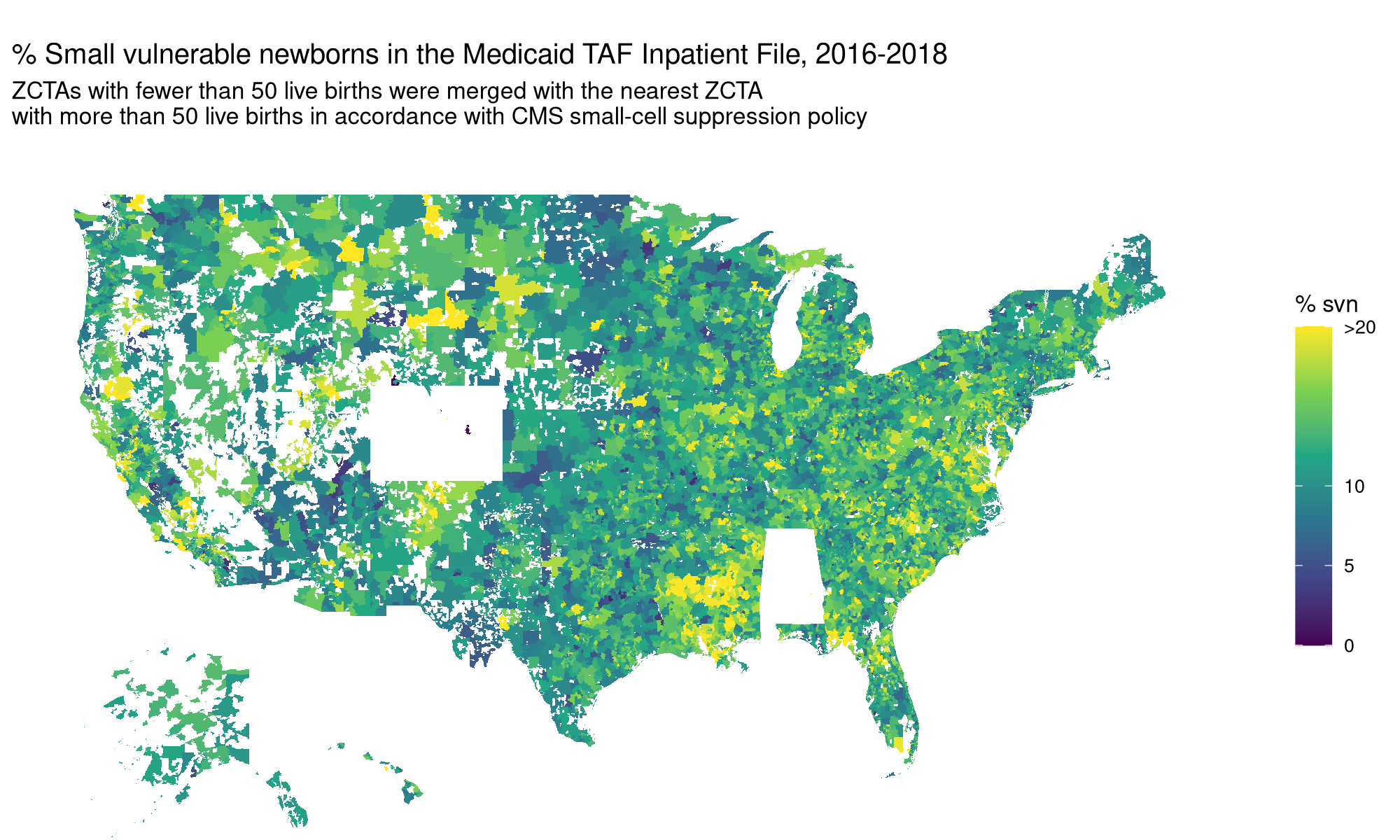}
        \caption{percentage of small vulnerable newborns--defined as preterm, small for gestational age, or low birth weight---among Medicaid-covered births in the United States, 2016--2018. ZCTAs in Colorado, Alabama, and Rhode Island are omitted due to irregularities in reporting (see Table \ref{tab:statecounts}).}
    \end{subfigure}
    \\
    \begin{subfigure}[t]{\textwidth}
        \centering
        \includegraphics[height=3.75in]{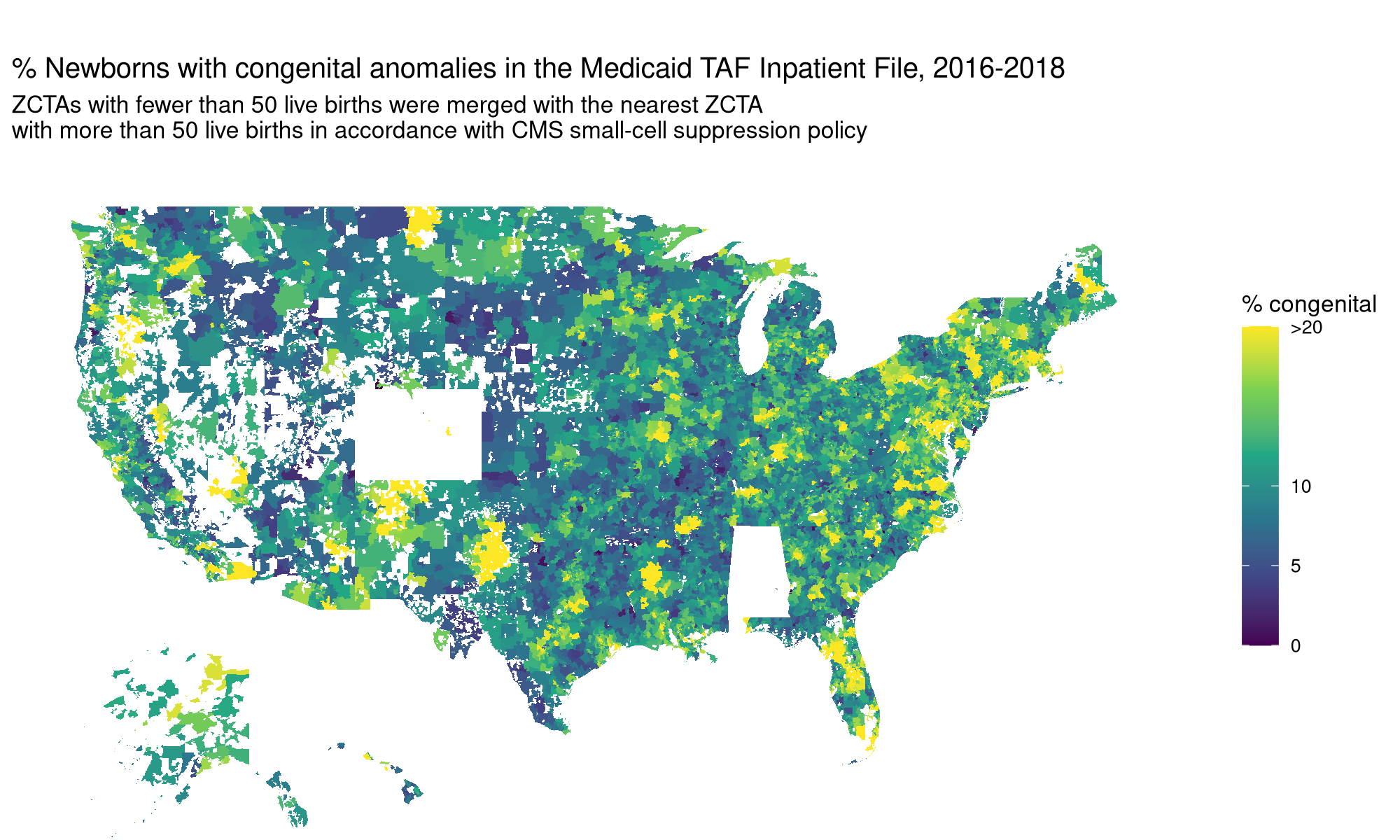}
        \caption{percentage of newborns with small congenital anomalies among Medicaid-covered births in the United States, 2016--2018. ZCTAs in Colorado, Alabama, and Rhode Island are omitted due to irregularities in reporting (see Table \ref{tab:statecounts}).}
    \end{subfigure}
    \caption{Two birth outcomes considered in our analysis.}
    \label{fig:outcomemaps}
\end{figure}

In the main analysis, the spatial weighting estimator approximately balances the means of the $17$ observed covariates and the means of each of the $10$ leading eigenvectors from each of the RE, CAR, and GP models. 
Formally, the resulting weighting problem can be written as:
\begin{align*}
    &\min_{\bm{w}} \sum_{i:Z_i = 0} w_i^2 \text{, subject to} \begin{cases} \bigg \vert \sum_{i:Z_i = 0} w_i  X^*_{ji} - \frac{1}{n_t}\sum_{i:Z_i = 1} X^*_{ji} \bigg \vert  \leq 0.001, j = 1,\ldots,17 \\
    \bigg \vert \sum_{i:Z_i = 0} w_i  v^{*\text{RE}}_{ki} - \frac{1}{n_t}\sum_{i:Z_i = 1} v^{*\text{RE}}_{ki} \bigg \vert  \leq 0.01, k = 1,\ldots,10 \\
    \bigg \vert \sum_{i:Z_i = 0} w_i  v^{*\text{CAR}}_{ki} - \frac{1}{n_t}\sum_{i:Z_i = 1} v^{*\text{CAR}}_{ki} \bigg \vert  \leq 0.01, k = 1,\ldots,10 \\
    \bigg \vert \sum_{i:Z_i = 0} w_i  v^{*\text{GP}}_{ki} - \frac{1}{n_t}\sum_{i:Z_i = 1} v^{*\text{GP}}_{ki} \bigg \vert  \leq 0.01, k = 1,\ldots,10 \\
    \sum_{i:Z_i = 0} w_i = 1, w_i \geq 0 \text{ }\forall i
    \end{cases}
\end{align*}where
\begin{itemize}
    \item $X^*_{ji} = (X_{ji} - \frac{1}{n}\sum_{i = 1}^nX_{ji})/\sqrt{\frac{1}{n}\sum_{i = 1}^n (X_{ji} - \frac{1}{n}\sum_{i = 1}^nX_{ji})^2}$.
    \item $\bm{v}^\text{RE}_{k}$ is the eigenvector associated with the $k$th largest eigenvalue of the spatial covariance matrix $\bm{S}^{\mathrm{RE}}$ and $v_{ki}^{*\text{RE}} = (v_{ki}^\text{RE} - \frac{1}{n}\sum_{i = 1}^nv_{ki}^\text{RE} )/\sqrt{\frac{1}{n}\sum_{i = 1}^n (v_{ki}^\text{RE} - \frac{1}{n}\sum_{i = 1}^nv_{ki}^\text{RE} )^2}$.  
    \item $\bm{v}^\text{CAR}_{k}$ is the eigenvector associated with the $k$th largest eigenvalue of the spatial covariance matrix $\bm{S}^{\mathrm{CAR}}$ and $v_{ki}^{*\text{CAR}} = (v_{ki}^\text{CAR} - \frac{1}{n}\sum_{i = 1}^nv_{ki}^\text{CAR} )/\sqrt{\frac{1}{n}\sum_{i = 1}^n (v_{ki}^\text{CAR} - \frac{1}{n}\sum_{i = 1}^nv_{ki}^\text{CAR} )^2}$.  
    \item $\bm{v}^\text{GP}_{k}$ is the eigenvector associated with the $k$th largest eigenvalue of the spatial covariance matrix $\bm{S}^{\mathrm{GP}}$ and $v_{ki}^{*\text{GP}} = (v_{ki}^\text{GP} - \frac{1}{n}\sum_{i = 1}^nv_{ki}^\text{GP} )/\sqrt{\frac{1}{n}\sum_{i = 1}^n (v_{ki}^\text{GP} - \frac{1}{n}\sum_{i = 1}^nv_{ki}^\text{GP} )^2}$.  
\end{itemize}

We assess the sensitivity of our results to two key design choices, the linearity of the basis functions $B_k$, and the specified balance tolerances. Table \ref{tab:sensitivity-basis} presents a summary of the results.

\begin{table}[!ht]
\footnotesize
    \centering
    \begin{tabular}{l|p{1.1in}|p{2.2in}|p{1.1in}|p{1.1in}}
    & & & \multicolumn{2}{c}{ATT estimate}\\
\hline
    Model & $B_k$ & balancing thresholds & \% Small Vulnerable & \% Congenital Anomalies \\
    \hline
        1 (main) & linear only & $0.001$ for standardized measured covariates, 
        0.01 for standardized latent covariates & -0.604 & 0.0672\\
        2 & linear only & automatic tuning via \citepsupp{wang2020minimal} & -0.601 & 0.0448\\
        3 & quadratic and pairwise interactions among $\bm{X}$ & $0.001$ for standardized measured covariates, 
        $0.01$ for standardized latent covariates, $0.01$ for all quadratic/interaction terms & -0.660 & 0.0831 \\
    \end{tabular}
    \caption{Variations on the spatial weighting estimator.}
    \label{tab:sensitivity-basis}
\end{table}

We further examine the sensitivity of our findings to treatment timing. Specifically, we compare (i) all treated sites to control sites with remediation started after December 31, 2015, an ideal but small control group ($n=394$); (ii) all treated sites to control sites with remediation started after January 1, 2005 ($n=460$); (iii) treated sites with remediation started before January 1, 2001 to all control sites ($n=1003$); and (iv) treated sites with remediation started after January 1, 2001 to all control sites ($n = 894)$. Table \ref{tab:sensitivity-treatment-timing} presents the results. 

Finally, we assess the sensitivity of our results to discrepancies in birth claim reporting by splitting the dataset into sites in states with more than 20\% deviation from the CDC WONDER Natality count ($n = 305$) and those in other states $(n = 774)$, as summarized in Table \ref{tab:statecounts}. 
Table \ref{tab:sensitivity-errorprone} presents the results. 

\begin{table}[!ht]
    \centering
    \renewcommand{\arraystretch}{1.3}
    \begin{tabular}{p{4in}p{1in}p{1in}}
    & \multicolumn{2}{c}{ATT estimate}\\
\hline
    &\makecell[l]{ \% Small \\Vulnerable} & \makecell[l]{\% Congenital \\Anomalies}\\
    \hline
     All treated sites versus control sites with remediation started after December 31, 2015 & -0.968 & -0.457\\
     All treated sites to control sites with remediation started after January 1, 2005 & -0.970 & -0.305\\
     Treated sites with remediation started before January 1, 2001 versus all control sites & -0.700 & -0.0136\\
     Treated sites with remediation started after January 1, 2001 versus all control sites & -0.737 & -2.92\\
    \end{tabular}
    \caption{Sensitivity analysis to treatment timing. }
    \label{tab:sensitivity-treatment-timing}
\end{table}

\begin{table}[!ht]
    \centering
    \renewcommand{\arraystretch}{1.3}
    \begin{tabular}{p{4in}p{1in}p{1in}}
    & \multicolumn{2}{c}{ATT estimate}\\
\hline
    &\makecell[l]{ \% Small \\Vulnerable} & \makecell[l]{\% Congenital \\Anomalies}\\
    \hline
    Sites in states with $>20$\% deviation from the CDC WONDER Natality Birth Count (AL, AR, CA, CO, HI, KS, LA, MD, MA, MN, MS, NE, NV, NH, NJ, RI, VT, WV, WY) & -1.262 & -0.299\\
        Sites in states with $<20$\% deviation from the CDC WONDER Natality Birth Count (AK, AZ, CT, DE, DC, FL, GA, ID, IL, IN, IA, KY, ME, MI, MO, MT, NM, NY, NC, ND, OH, OK, OR, PA, SC, SD, TN, TX, UT, VA, WA, WI) 
         & -0.609 & -0.00793
    \end{tabular}
    \caption{Sensitivity analysis to birth claim reporting. }
    \label{tab:sensitivity-errorprone}
\end{table}

\clearpage
 \bibliographystylesupp{asa}
 \bibliographysupp{main}
 
\end{document}